\documentclass[12pt,american]{article}
\usepackage[labelsep=space]{caption}
\usepackage{verbatim}
\usepackage{geometry}
\usepackage{setspace}
\usepackage{rotating} 
\usepackage{booktabs}
\usepackage{todonotes}
\usepackage{mathrsfs}
\usepackage{amsthm}
\usepackage{amsmath,amsfonts,amssymb}
\usepackage{gobble}
\usepackage{zref-xr}
\usepackage[section]{placeins}
\usepackage{mathtools}
\usepackage{amsfonts}
\usepackage{amssymb}
\usepackage{cancel}\usepackage{pifont}\usepackage{amsbsy}
\usepackage{bm}
\usepackage[sort,authoryear,longnamesfirst]{natbib}
\setcitestyle{citesep={,}}  
\usepackage{hyperref}
\hypersetup{colorlinks, citecolor=blue!50!black, filecolor=blue!50!black, linkcolor=red!50!black, urlcolor=blue!50!black}
\usepackage{xr}
\usepackage{rotating}
\usepackage{multirow}
\usepackage{graphicx}
\usepackage{environ}         \usepackage{etoolbox}        \usepackage{epstopdf}
\usepackage{epsfig}
\usepackage{url,hyperref}
\usepackage{amsfonts}
\usepackage{enumerate}
\usepackage{appendix}
\usepackage{multirow}
\usepackage{multicol}
\usepackage{fmtcount}
\usepackage{bbm}
\usepackage{indentfirst}
\usepackage{booktabs,threeparttable}
\usepackage{paralist}
\usepackage{mathrsfs}
\usepackage{pdfpages}
\usepackage{adjustbox}
\usepackage{subfig}
\usepackage{float}
\usepackage{siunitx}
\usepackage{tikz}
\usepackage{dcolumn}
\newcolumntype{d}{D{.}{.}{-1}}

\usepackage{algorithm}
\usepackage{algpseudocode}
\usepackage{float} 


\usetikzlibrary{shapes,arrows,arrows.meta,trees,snakes,positioning,fit,shadows,shapes.geometric,decorations.markings}
\usepackage{array}
\usepackage{mathtools}
\newcolumntype{L}[1]{>{\raggedright\let\newline\\\arraybackslash\hspace{0pt}}m{#1}}
\newcolumntype{C}[1]{>{\centering\let\newline\\\arraybackslash\hspace{0pt}}m{#1}}
\newcolumntype{R}[1]{>{\raggedleft\let\newline\\\arraybackslash\hspace{0pt}}m{#1}}
\makeatletter
\newcommand{\moverlay}{\mathpalette\mov@rlay}
\newcommand{\mov@rlay}[2]{\leavevmode\vtop{\baselineskip\z@skip{}
\lineskiplimit-\maxdimen\ialign{\hfil#1##\hfil\cr#2\cr\cr}}}
\makeatother

\newcommand{\extratext}[1]{\relax}

\DeclareMathAlphabet{\mathitbf}{OML}{cmm}{b}{it}
\DeclareMathOperator*{\argmin}{arg\,min}

\usepackage{etoolbox}

\makeatletter
\patchcmd{\abstract}{\quotation}{\quotation\noindent\hspace{-\parindent}\hspace{-1cm}}{}{}
\patchcmd{\endabstract}{\endquotation}{\endquotation}{}{}
\makeatother
\theoremstyle{plain}
\newtheorem{theorem}{Theorem}\newtheorem{proposition}{Proposition}\newtheorem{example}{Example}

\newtheorem{assumption}{Assumption}
\makeatletter
\newtheorem*{assumptionUC*}{Assumption UC}
\newenvironment{assumptionUC}[1][]{\begin{assumptionUC*}[#1]\phantomsection \def\@currentlabel{UC}}{\end{assumptionUC*}}
\makeatother

\newtheorem{lemma}{Lemma}

\newtheorem{corollary}{Corollary}\theoremstyle{remark}
\newtheorem{remark}{Remark}

\makeatletter
\renewenvironment{proof}[1][\proofname] {\par\pushQED{\qed}\normalfont\topsep6\p@\@plus6\p@\relax\trivlist\item[\hskip\labelsep\bfseries#1\@addpunct{.}]\ignorespaces}{\popQED\endtrivlist\@endpefalse}
\makeatother
 
\usepackage{apptools}
\usepackage{chngcntr}   \AtAppendix{\counterwithin{lemma}{section}
\renewcommand{\thelemma}{\thesection.\arabic{lemma}}}
\numberwithin{equation}{section}
\renewcommand{\theequation}{\thesection.\arabic{equation}}

\usepackage{xifthen}
\newcounter{change}
\renewcommand{\thechange}{\arabic{change}} 

   {\par }

\usepackage{titlesec}
\titleformat{\section}
  {\normalfont\Large\bfseries}{\thesection}{1em}{}
\titleformat{\subsection}
  {\normalfont\large\bfseries}{\thesubsection}{1em}{}
\titleformat{\subsubsection}
  {\normalfont\normalsize\bfseries}{\thesubsubsection}{1em}{}

\begin{document}

\title{\hspace*{-0.5cm}Low-Rank Estimation of Nonlinear Panel Data Models\thanks{I am deeply grateful to my advisors, Yuichi Kitamura, Donald Andrews, and Xiaohong Chen, for their invaluable support and guidance. I also thank Timothy Christensen, Max Cytrynbaum, Kengo Kato, Ed Vytlacil, and participants of the Econometrics Prospectus Workshops at Yale for helpful comments. All errors are my own.}
}
\author{
Kan Yao\thanks{Corresponding author. Email: \texttt{kan.yao@yale.edu}.} \\ 
 \\[1em]
}

\date{\today}
\maketitle

\begin{abstract}

\normalsize
\noindent This paper investigates nonlinear panel models with interactive fixed effects and introduces a general framework for parameter estimation under potentially nonconvex objective functions. We propose a computationally feasible two-step estimation procedure. In the first step, nuclear-norm regularization (NNR) is used to obtain preliminary estimators of the coefficients of interest, factors, and factor loadings. The second step involves an iterative procedure for post-NNR inference, improving the convergence rate of the coefficient estimator. We establish the asymptotic properties of both the preliminary and iterative estimators. We also study the determination of the number of factors. Monte Carlo simulations demonstrate the effectiveness of the proposed methods in determining the number of factors and estimating the model parameters. In our empirical application, we apply the proposed approach to study the cross-market arbitrage behavior of U.S. nonfinancial firms.

\vspace{0.5cm}

	\bigskip
\end{abstract}
\newpage \section{Introduction}
\onehalfspacing
Unobserved heterogeneity is of broad interest in both reduced-form
and structural work in economics and other social sciences. Empirical data often include individual units sampled over time from diverse backgrounds, with factors unobserved by econometricians.
Accommodating this heterogeneity in a flexible, yet parsimonious manner is challenging but essential in practice.

Panel data, which involve observations on individual units over time,
provide a valuable framework to model latent structures within low-dimensional manifolds. 
This is often achieved by incorporating unobserved individual and
time effects into the model, controlling for unobserved covariates that remain either time or cross-sectionally invariant. A fixed effects approach imposes no distributional assumptions on these unobserved effects, allowing them to be arbitrarily related to observed covariates. A notable example is two-way fixed effects
estimation for difference-in-differences, which has become a leading method in applied economics. 

The additive structure, however, assumes that the influence of unobserved factors on individual units remains constant over time, failing to capture more complex dynamics. Alternatively, fixed effects can interact multiplicatively, giving rise to interactive fixed effects (IFEs) or factor structures. This multiplicative form provides a more flexible representation of heterogeneity, as it allows common time-varying shocks (factors) to affect cross-sectional units with individual-specific sensitivities (factor loadings). For example, they can account for aggregate shocks with heterogeneous impacts on agents in macroeconomic models or capture multidimensional individual heterogeneity with time-varying effects in microeconomic models. This flexibility motivated the discussion of interactive effects
in the econometrics literature. \citet*{bai_panel_2009} and \citet*{moon_linear_2015,moon_dynamic_2017} study the linear regression model, treating individual
and time effects as nuisance parameters estimated by least squares. \citet*{chen_nonlinear_2021} and \citet*{wang_maximum_2022} extend the least squares estimator to nonlinear panel data models with convex log-likelihood functions.

However, many important and widely used economic models, such as random coefficients choice models, do not have globally convex objective functions. The problem becomes even more challenging when time-varying unobserved heterogeneity is represented through interactive fixed effects. As a result, panel models with interactive fixed effects and potentially nonconvex objective functions remain largely unexplored. Existing approaches in the literature are not directly applicable because they rely on closed-form solutions or global convexity.

To address this gap, the present paper develops a general framework for parameter estimation in panel models with interactive fixed effects under potentially nonconvex objective functions. 
The framework therefore accommodates a broad class of economically important nonlinear models whose objective functions are not globally convex with respect to the parameters of interest, such as random coefficients choice models.
The proposed estimation procedure consists of two steps. In the first step, we use nuclear-norm regularization (NNR) to obtain preliminary estimators of the coefficients of interest, factors,
and factor loadings. In the second step, we use the first-step NNR estimator as an initial value for a localized post-regularization iterative procedure, which refines the estimator for the parameters of interest.

In this context, we make three contributions. First, we develop a first-step nuclear-norm-regularized estimator for nonlinear panel models with interactive fixed effects. By regularizing the latent interactive component through its nuclear norm, the method avoids requiring the researcher to specify the number of factors in advance. We establish consistency and convergence rates for the preliminary estimator and show that the number of factors can be estimated consistently from the singular values of the NNR estimator. In the case of convex objective functions, the first-step estimator is obtained from a convex optimization problem, thereby avoiding the multiple local minima associated with rank-constrained estimation. 

Second, we establish new post-regularization inference results. Although the first-step NNR estimator is consistent, it suffers from shrinkage bias and generally converges too slowly for standard inference. To address this problem, we develop a second-step iterative procedure that starts from the NNR estimator and removes the regularization bias introduced in the first step. We show that, when initialized by the NNR estimator, this iterative procedure satisfies a contraction mapping property and reaches the usual parametric rate after a logarithmic number of iterations. We then derive the asymptotic distribution of the resulting estimator, characterize the incidental parameter bias terms, and provide bias corrections for post-regularization inference.

Establishing these convergence and distributional results is technically challenging because it requires controlling both the asymptotic bias and the stochastic error of the iterative estimator at each iteration. The proposed second-step procedure is conceptually related to the iterative refinements in \citet*{moon_nuclear_2019} and \citet*{hong_profile_2023}. However, unlike those papers, which study linear models and exploit closed-form updates, our setting involves nonlinear objective functions for which no closed-form update is available. Moreover, although our procedure allows objective functions that need not be globally convex, in the special case of single-index models with convex link functions, we show that our second-step estimator attains the same asymptotic distribution as the estimator studied by \citet*{chen_nonlinear_2021}. 

Third, the paper provides numerical algorithms for implementing the proposed estimators. For single-index panel models with convex loss functions, we compute the first-step estimator using a modified Alternating Direction Method of Multipliers (ADMM) algorithm. For random coefficients panel models with interactive fixed effects, we develop a new majorization-minimization (MM) algorithm that constructs a surrogate objective. It separates the update of the distribution of the random coefficients from that of the interactive fixed effects, while incorporating the nuclear-norm penalty. 
Monte Carlo simulations indicate that these first-step algorithms have good finite-sample performance and that the associated rank estimator recovers the true number of latent factors. The simulations also show that the iterative procedure improves estimation accuracy relative to the first-step estimator.

We also provide an empirical illustration of the proposed procedure by revisiting \citet{ma_nonfinancial_2019}. Specifically, we study the joint debt-equity financing decisions of U.S. nonfinancial firms using multiple firm-level financial datasets. The estimates produce economically sensible outcomes that align with financial intuition. For example, firms are more likely to issue equity and retire debt when the cost of debt is high and the cost of equity is low, consistent with the view that firms may act as cross-market arbitrageurs in their own securities. In addition, by allowing for both interactive fixed effects and random coefficients, the model suggests the presence of heterogeneity in sensitivities to market valuation measures.

This paper relates to three branches of the literature. First, it adds to the extensive literature on panel data models with IFEs. \citet*{bai_panel_2009}
and \citet*{moon_linear_2015,moon_dynamic_2017} propose estimators
based on quasi-maximum likelihood estimation and principal component analysis. Much previous work has focused on linear models, while nonlinear models have recently attracted attention.
\citet*{chen_nonlinear_2021} extend
\citet*{fernandez-val_individual_2016}'s results to nonlinear panel data models with IFEs. \citet*{chen_quantile_2021} study quantile factor models, and \citet*{gao_binary_2023-1,ando_bayesian_2022} study binary panel choice models. Relatedly, \citet*{boneva_discrete-choice_2017} and \citet*{chen_common_2025} extend the common correlated effects (CCE) framework of \citet*{pesaran_estimation_2006} to nonlinear settings.
We refer to \citet*{fernandez-val_fixed_2018} for a recent survey. We complement and extend the literature by generalizing previous results to a broad range of M-estimators.

Second, our work contributes to the burgeoning literature on nuclear norm penalization, a technique that has gained popularity in estimating low-rank matrices in statistics and econometrics. This approach has been explored in various works such as \citet*{agarwal_noisy_2012,agarwal_causal_2021,athey_matrix_2021,belloni_high_2023,beyhum_square-root_2019,candes_exact_2009,chernozhukov_inference_2019-1,chernozhukov_inference_2023,fan_generalized_2017,feng_nuclear_2023,hong_profile_2023,huang_low-rank_2018,miao_high-dimensional_2022,ma_detecting_2021,negahban_estimation_2011,negahban_unified_2012}, among others. Most of these works focus on establishing estimation error bounds, specifically in the Frobenius norm, for the NNR estimators.
The convergence rate of our estimator is consistent with the existing results of the mean and quantile regressions, indicating that NNR can be successfully extended to more complex models without compromising performance.

More recent studies have advanced statistical inference for NNR-based estimators. For instance, linear models have been studied by \citet*{armstrong_robust_2023, chernozhukov_inference_2023, miao_high-dimensional_2022, moon_nuclear_2019, hong_profile_2023}, among others. Two main debiasing strategies have emerged to address shrinkage bias. The first, inspired by the debiased-Lasso framework in the statistics literature, typically requires relatively few iterations. \citet*{chernozhukov_inference_2023} combine a rotation-based debiasing step with sample splitting, while \citet*{choi_inference_2024} show that sample splitting is not essential. \citet*{armstrong_robust_2023} apply minimax linear estimation theory to construct robust debiased estimators. 

The second line of research, to which our paper contributes, employs iterative bias-correction procedures, as exemplified by \citet*{hong_profile_2023}, \citet*{miao_high-dimensional_2022}, and \citet*{moon_nuclear_2019}. Specifically, \citet*{ChenMiaoSu2025} study the asymptotic distribution of factor and loading estimators in logistic panel models without covariates. In a closely related and independently developed paper, \citet*{Zeleneev_tractable_2026} study estimation of single-index panel models with interactive fixed effects. In particular, they provide concrete convex optimization procedures for both steps and establish their global convergence guarantees under the convex link function assumption. Our paper, on the other hand, accommodates nonconvex model specifications, including the important class of random coefficients models, and addresses distinct theoretical and computational challenges. 
 
Lastly, this paper is also related to the literature on estimating the parameters of high-dimensional models using $\ell_1$ regularization. Rather than listing all relevant papers, we refer the interested reader to the comprehensive textbook by \citet*{Hastie_lasso_2015} and the recent survey by \citet*{belloni_high-dimensional_2018}, and instead focus on a few key references. Specifically, \citet*{chetverikov_selecting_2025} develop a method called bootstrapping after cross-validation for selecting the penalty parameter in $\ell_1$-penalized M-estimators in high dimensions. \citet*{stadler_l1-penalization_2010} derive an oracle inequality for the $\ell_1$-penalized approach to estimating mixture regression models, while \citet*{beyhum_high-dimensional_2024} study high-dimensional nonconvex Lasso-type M-estimators and establish their convergence rates. However, none of these studies consider panel data settings or nuclear norm regularization.

The remainder of the paper is organized as follows. Section~\ref{sec:model} introduces the class of IFE models. Section~\ref{sec:est} proposes a general M-estimation framework and presents the main estimation procedure. In Section~\ref{sec:main}, we examine the asymptotic properties of the estimators, including the consistency and rate of convergence of the initial estimator and the rank estimator, as well as the convergence and asymptotic distribution of the second-step estimator. Section~\ref{sec:MC} outlines the implementation algorithms and provides Monte Carlo results. Section~\ref{sec:empirical} contains the empirical application. Finally, Section~\ref{sec:conc} concludes. All proofs and additional results are provided in the Appendix.

\paragraph*{Notation.}

For a natural
number $m\in\mathbb{N}$, we introduce the notation $[m]=\{1,\ldots,m\}$. For a vector $v=\left(v_{1},\ldots,v_{p}\right)^{\prime}\in\mathbb{R}^{p}$,
we define its $\ell_{1}$-norm as $\|v\|_{1}=\sum_{j=1}^{p}\left|v_{j}\right|$,
and its $\ell_{2}$-norm norm as $\|v\|=\sqrt{\sum_{j=1}^{p}v_{j}^{2}}$.
For a matrix $A=\left(a_{it}\right)_{i,t}\in\mathbb{R}^{N\times T}$,
we define its entry-wise $\ell_{1}$-norm as $\|A\|_{1}=\sum_{i=1}^{N}\sum_{t=1}^{T}\left|a_{it}\right|$,
its Frobenius norm as $\|A\|_{F}=\sqrt{\sum_{i=1}^{N}\sum_{t=1}^{T}a_{it}^{2}}$,
its infinity norm as $\|A\|_{\infty}=\max\left\{ \left|a_{it}\right|:i\in[N],t\in[T]\right\} $,
its nuclear norm as $\|A\|_{*}=\operatorname{trace}\left(\sqrt{A^{\prime}A}\right)$,
its spectral norm as $\|A\|=\sup_{x:\|x\|=1}\sqrt{x^{\prime}A^{\prime}Ax}$, and its rank by $\operatorname{rank}(A)$. Moreover, for a real matrix, we use $\sigma_s(\cdot)$ to denote its $s$-th largest singular value. For a real and symmetric matrix, we use $\mu_s(\cdot), \mu_{\max }(\cdot)$ and $\mu_{\min }(\cdot)$ to denote its $s$-th largest eigenvalue, the largest and smallest eigenvalues, respectively.

For a real number $a$, let $\operatorname{sgn}(a)=1$ if $a \geq 0$ and $\operatorname{sgn}(a)=-1$ if $a<0$. For a square matrix $A$ whose $j$-th diagonal element is denoted as $A_{j j}$, define $\operatorname{sgn}(A)$ as a diagonal matrix whose $j$-th diagonal element is equal to $\operatorname{sgn}\left(A_{j j}\right)$.
We also define the operations $\vee$ and $\wedge$ as $a\vee b=\max\{a,b\}$ and $a\wedge b=\min\{a,b\}$. Finally, for sequences $\left\{ a_{m}\right\} _{m=1}^{\infty}$
and $\left\{ b_{m}\right\} _{m=1}^{\infty}$ we use $a_{m}\lesssim b_{m}$ as the shorthand for the inequality $a_{m}\leq\overline{c}b_{m}$ for some finite positive $\overline{c}$ independent of $a_{m}$ and $b_m$ for sufficiently large $m$. $a_{m}\asymp b_{m}$ means that $a_{m}\lesssim b_{m}$ and $b_{m}\lesssim a_{m}$.
 \section{Model\label{sec:model}}
Let the data observations be denoted as $\left\{ \left(Y_{it},X_{it}\right):i\in\left[N\right],t\in\left[T\right]\right\} $,
where $Y_{it}\in\mathcal{Y}\subset\mathbb{R}$ is the outcome variable
and $X_{it}\in\mathcal{X}\subset\mathbb{R}^{d_x}$ is a vector of exogenous
covariates for a fixed $d_x$. The indices $i$ and $t$ denote individuals
and time periods, respectively. In matrix
notation, let $Y$ and $X_{j}$ be $N\times T$ matrices representing
the outcome and the $j$-th covariate for $j=1,...,d_x$. Let $X=\left(X_{j}\right)_{j\in\left[d_x\right]}$ collect all covariates. The support of $\left(Y_{it},X_{it}\right)$ is given by $\mathcal{Y}\times\mathcal{X}$.

We assume that for each individual $i$ and time $t$, the outcome $Y_{it}$ is allowed to depend on both the observed covariates $X_{it}$ and the latent interactive effects given by individual-specific
loadings $\lambda_{i}\in\mathbb{R}^{r}$ and time factors $f_{t}\in\mathbb{R}^{r}$, where $r$ is a fixed rank. The effects $\lambda_{i}$ and $f_t$, although unobserved by the econometrician, may confound the effect of $X_{it}$ on $Y_{it}$. Following the fixed effects approach, we treat the realizations of $\left\{\lambda_i\right\}_{i=1}^N$ and $\left\{f_t\right\}_{t=1}^T$ as unrestricted parameters to be estimated.

Formally, we consider a class of nonlinear panel data models in which the true values of the common parameter vector $\theta \in \mathbb{R}^p$ and the $N\times r$ and $T\times r$ matrices of fixed effects, $\Lambda=\left(\lambda_{1},...,\lambda_{N}\right)^{\prime}$ and $F=\left(f_{1},...,f_T\right)^{\prime}$, are defined by the solution to the population optimization problem,\footnote{The model is identified only up to a normalization; see Remark~\ref{remark:normlization}. A formal discussion of the identification assumptions is provided in Assumption~\ref{assu:id-2}.}
\begin{align}
\left(\theta_{0},\Lambda_{0},F_{0}\right)=&\argmin_{\theta\in\Theta,\Lambda\in\Phi_{\lambda}^{N},F\in\Phi_{f}^{T}}\mathbb{E}\left[\frac{1}{NT}\sum_{i=1}^{N}\sum_{t=1}^{T}\ell\left(W_{it};\theta,\lambda_{i}^{\prime}f_{t}\right)\right]
\label{ref:Model}
\end{align}
where $\ell$ is a known loss function, which may be nonconvex with respect to $\theta$ and $\lambda_i'f_t$, and where $W_{it} = (Y_{it}, X_{it})$. 
The parameter spaces satisfy \(\Theta\subset\mathbb R^p\), 
\(\Phi_\lambda\subset\mathbb R^r\), and \(\Phi_f\subset\mathbb R^r\) such that the matrix parameters belong to the product spaces
\(\Lambda\in\Phi_\lambda^N\) and \(F\in\Phi_f^T\).
The expectation $\mathbb{E}[\cdot]$ is taken with respect to the distribution of the data, conditional on the fixed realizations of unobserved individual and time effects for all $N$ and $T$.
The model is also referred to as a factor model with factor loadings $\lambda_{i}$
and common factors $f_{t}$, and we will use the terms ``factor''
and ``interactive fixed effect'' synonymously. The conventional
additive structure is a special case of the factor structure with
$r=2$, $\lambda_{i}=\left(\lambda_{i1},1\right)^{\prime}$, and $f_{t}=\left(1,f_{1t}\right)^{\prime}$. 

Let $\pi_{0,it} = \lambda_{0i}^{\prime} f_{0t}$ denote the true value of $\pi_{it}$ that generates the data, where $\pi_{0,it} \in \Phi \subset \mathbb{R}$ for each $i$ and $t$. The matrix $\Pi_0 \in \Phi^{N \times T}$ collects all the fixed effects. Similarly, the matrix $\Pi = (\pi_{it})$ is treated as a parameter to be estimated. Thus, we can rewrite Model \eqref{ref:Model} as
\begin{align}
\left(\theta_{0},\Pi_{0}\right)= & \argmin_{\theta\in\Theta,\Pi\in\Phi^{N\times T}}\mathbb{E}\left[\frac{1}{NT}\sum_{i=1}^{N}\sum_{t=1}^{T}\ell\left(W_{it};\theta,\pi_{it}\right)\right]\label{ref:obj}\\
 & \qquad\text{s.t. }\pi_{it}=\lambda_{i}^{\prime}f_{t}=\sum_{s=1}^{r}\lambda_{is}f_{ts}\nonumber 
\end{align}

\begin{remark}\label{remark:normlization}
Similar to linear factor models, the factor
loadings $\Lambda_0$ and common factors $F_0$ in \eqref{ref:Model} cannot be separately identified without imposing normalization. The loss function is invariant under the transformation $\lambda_{i} \mapsto A^{\prime}\lambda_{i}$ and $f_{t} \mapsto A^{-1}f_{t}$ for any non-singular $r \times r$ matrix $A$. On the other hand, $\Pi_0$ in \eqref{ref:obj} can be uniquely identified, as it is invariant to the normalization used to eliminate this indeterminacy. For further discussion on normalization in the context of linear factor models, see \citet*{robertson_iv_2015}.
\end{remark}

\subsection{Examples}
Some examples of models that fall within this framework and the scope
of our methodology are as follows. \begin{example}[Linear Panel]
\label{ref:ex1} The linear panel model is generated according to
\begin{equation}
y_{it}=X_{it}^{\prime}\theta_{0}+\pi_{0,it}+\varepsilon_{it}\label{ref:linear}
\end{equation}
where $\varepsilon_{it}$ is the idiosyncratic error term with $\mathbb{E}\left[\varepsilon_{it}\mid X_{it}\right]=0$.
We can estimate $\theta$ and $\pi_{it}$ using $\ell\left(W_{it};\theta,\pi_{it}\right)=\frac{1}{2}\left(y_{it}-X_{it}^{\prime}\theta-\pi_{it}\right)^{2}$ and $W_{it}=(y_{it},X_{it})$.
\end{example} \begin{example}[Linear Quantile Panel] \label{ref:ex2}Consider
the linear model in (\ref{ref:linear}) with the quantile restriction:
$\mathrm{P}\left(\varepsilon_{it}\leq0\mid X_{it};\theta_{0},\pi_{0,it}\right)=\tau$.
We estimate $\theta$ and $\pi_{it}$ by using $\ell\left(y_{it},X_{it}^{\prime}\theta+\pi_{it}\right)=\rho_{\tau}\left(y_{it}-X_{it}^{\prime}\theta-\pi_{it}\right)$,
where $\rho_{\tau}\left(t\right)=t\left(\tau-K\left(-t/h\right)\right)$ represents a smoothed version of the usual check function. Here, $\tau\in\left(0,1\right)$ denotes a specified quantile
of interest, $K$ is a CDF-type kernel function, and $h$ is the bandwidth parameter.

\end{example} \begin{example}[Binary Choice Panel] \label{ref:ex3}
The binary choice panel model is characterized as 
\begin{equation}
y_{it}=\mathbb{I}\left\{ X_{it}^{\prime}\theta_{0}+\pi_{0,it}-\varepsilon_{it}\geq0\right\} \label{ref:binary}
\end{equation}
where $X_{it}$ and $\varepsilon_{it}$ are independent and $\varepsilon_{it}$
is distributed according to a known cumulative distribution function
$F$. Since $\mathrm{P}\left(y_{it}=1\mid z_{it}\right)=F\left(z_{it}\right)$,
we define the objective function as $-\ell\left(W_{it};\theta,\pi_{it}\right)=y_{it}\log F\left(X_{it}^{\prime}\theta+\pi_{it}\right)+\left(1-y_{it}\right)\log\left[1-F\left(X_{it}^{\prime}\theta+\pi_{it}\right)\right]$, where $W_{it}=(y_{it},X_{it})$.
\end{example} 
\begin{example}[Random Coefficients Logit Panel]\label{ref:ex4}
Consider a binary choice model such that 
\[
y_{it}=\mathbb{I}\left\{ X_{it}^{\prime}\beta_{0,it}+\pi_{0,it}-\varepsilon_{it}\geq0\right\} 
\]
where $\varepsilon_{it}$ is a random variable with the Type I extreme
value distribution, independent across individuals $i$ over time $t$, and $\beta_{0,it}$
is a $p$-dimensional vector of random coefficients, assumed to be i.i.d. across individuals
over time.\footnote{We could relax the i.i.d. assumption by imposing a parametric model upon $\beta_{0,it}$ over time $t$, e.g., an AR(1) process, but we do not pursue this approach
here.} Specifically, we assume $\beta_{0,it}\sim \mathcal{N}\left(\bar{\beta}_{0},\Sigma_{0}\right)$, where $\bar{\beta}_{0}$ is a $p$-dimensional mean vector, and $\Sigma_{0}$ is a positive-definite $p\times p$ covariance matrix. The theoretical results presented in this paper also hold for other distributions of $\beta_{0,it}$. Let $f\left(\beta\right)$ denote the probability density function of the random coefficient $\beta_{it}$ evaluated at $\beta$. We define the parameter set $\theta=\left(\bar{\beta},\Sigma\right)$. The objective function is given by $$\ell\left(W_{it};\theta,\pi_{it}\right)=-\log\left[\int L_{it}(\beta,\pi_{it})f\left(\beta\right)d\beta\right]$$
where the integrand is given by $L_{it}(\beta,\pi_{it})=p_{it}(\beta,\pi_{it})^{y_{it}}(1-p_{it}(\beta,\pi_{it}))^{1-y_{it}}$, with $p_{it}\left(\beta,\pi_{it}\right)=\frac{\exp\left(X_{it}^{\prime}\beta+\pi_{it}\right)}{1+\exp\left(X_{it}^{\prime}\beta+\pi_{it}\right)}$, and $W_{it}=(y_{it},X_{it})$.
\end{example} 
 \section{Estimation\label{sec:est}}
In this subsection, we describe our two-step estimation procedure for the model: the initial estimator is based on nuclear norm regularization (NNR). It is followed by a second step iterative estimator. As we will show in Sections~\ref{subsec:consistency} and~\ref{subsec:rate}, under a set of regularity conditions, the first step estimation provides consistent estimators for $\theta_0$, $\Lambda_0$ and $F_0$. The second step uses the estimators from the first step as the initial values, and iteratively updates estimators to enhance their properties. The asymptotic normality of the iterative estimator is established in Section~\ref{sec:inf}.

\subsection{First Step: Nuclear Norm Regularization}
We define the loss function as 
\begin{equation}\label{ref:Lnt_Def}
\mathcal{L}_{NT}\left(\theta,\Pi\right)\coloneqq\frac{1}{NT}\sum_{i=1}^{N}\sum_{t=1}^{T}\ell\left(W_{it};\theta,\pi_{it}\right)
\end{equation}

Motivated by the literature on nuclear norm regularization, we propose estimating $\left(\theta_{0},\Pi_{0}\right)$ by minimizing the following penalized criterion function:
\begin{equation}
(\hat{\theta},\hat{\Pi})=\argmin_{\theta\in\Theta,\Pi\in\Phi^{N\times T}}\mathcal{L}_{NT}\left(\theta,\Pi\right)+\nu\|\Pi\|_{*}\label{ref:Estimation1}
\end{equation}
where $\nu$ is a tuning parameter,\footnote{The tuning parameter $\nu$ depends on $N$ and $T$, but we omit the subscript for simplicity.} and $\|\cdot\|_{*}$ denotes the nuclear
norm, which regularizes the singular values of the matrix $\Pi$. The estimation
problem at hand is inherently high-dimensional, as it involves estimating a total of $p+r\left(N+T\right)$ parameters, where the number of parameters grows linearly with $N$ and $T$. Specifically, $p$ represents the number of parameters associated with $\theta$, and $r(N + T)$ corresponds to those related to the low-rank matrix $\Pi$.

Our estimator can be viewed as a generalization of Lasso, with a key difference being that we impose sparsity not directly on individual parameters but on the singular values of the matrix $\Pi$. This regularization ensures that the number of non-zero singular values of $\Pi$ is relatively sparse compared to $N$ and $T$. Since we assume that the rank $r$ of the matrix $\Pi_0$ is fixed and much smaller than both $N$ and $T$, the nuclear norm regularization reduces the complexity by encouraging a low-rank structure in $\Pi$, similar to how Lasso enforces sparsity in individual coefficients.

Similar to Lasso, \eqref{ref:Estimation1} can be seen as a convex relaxation of the following problem:
\begin{align}
 & \min_{\theta\in\Theta,\Pi\in\Phi^{N\times T}}\mathcal{L}_{NT}\left(\theta,\Pi\right)\nonumber \\
 & \quad\quad\textrm{s.t.}\quad\text{rank}\left(\Pi\right)\leq r\label{ref:Model1}
\end{align}
where the rank $r$ is predetermined, and $\theta$, $\lambda_{i}$
and $f_{t}$ are jointly estimated, as studied by \citet*{bai_panel_2009} and \citet*{chen_nonlinear_2021},
among others. While the formulation in \eqref{ref:Model1} seems appealing,
it presents two challenges. 

First, the rank constraint makes \eqref{ref:Model1} a nonconvex optimization problem, even when the loss function is convex. 
In contrast, our proposed estimator, defined in \eqref{ref:Estimation1}, avoids directly penalizing or constraining the rank of the estimated interactive fixed effects matrix. Instead, it seeks a $\hat{\Pi}$ with a small nuclear norm, serving as a convex relaxation of \eqref{ref:Model1}. Just as $\ell_1$ minimization is the tightest convex relaxation of the combinatorial $\ell_0$ minimization problem, nuclear norm minimization offers the tightest convex relaxation of the NP-hard rank minimization problem \citep*{candes_matrix_2010}.

Second, when the loss function $\mathcal{L}_{NT}$ is nonconvex, as in random coefficients logit panel models, there are no guarantees of convergence for the estimators defined in \eqref{ref:Model1}. Global convexity plays a crucial role in the literature on interactive fixed effects. For instance, \citet*{chen_nonlinear_2021} rely on global convexity to establish the consistency of $\theta$ and to bound the remainder terms in their stochastic expansions. In contrast, our penalized estimator accommodates a broader class of important economic models without global convexity.

Note that the NNR-based initial estimation does not require the number of factors $r$ to be known beforehand. When $r$ is unknown, we propose estimating $r$ using singular value thresholding (SVT) as follows:
\begin{equation}
 \hat{r}=\sum_{s=1}^{N\wedge T}\mathbb{I}\left\{ \sigma_{s}\left(\hat{\Pi}\right)\geq\sqrt{NT}\left(\nu\left\Vert \hat{\Pi}\right\Vert \right)^{1/2}\right\} 
\label{ref:r}   
\end{equation}
where $\sigma_{s}\left(\hat{\Pi}\right)$ denotes the $s$-th largest
singular value of $\hat{\Pi}$. 
Alternatively, $r$ can be estimated using other methods, such as the panel information criteria (IC) and panel criteria (PC) methods of \citet*{bai_determining_2002}, or the eigenvalue ratio (ER) and growth ratio (GR) methods of \citet*{ahn_eigenvalue_2013}, among others. These methods remain valid as long as $\hat{\theta}$ is consistent. In Section~\ref{subsec:rate}, we first establish the consistency of $\hat{\theta}$ and $\hat{\Pi}$, then we prove that $\hat{r}=r$ with probability approaching one. Consequently, for the second-step estimation, we assume $r$ is known.

\subsection{Second Step: Iterative Estimation}
It is widely recognized that NNR estimators
are subject to shrinkage bias, which complicates statistical inference. Theorem~\ref{thm:main-3} establishes the convergence rate of the initial estimator, which, although consistent, is slower than the $\sqrt{NT}$ rate. To address this issue, we propose a post nuclear-norm-regularization estimator to improve the convergence rate.

For simplicity, we focus primarily on maximum likelihood estimation. We assume that the outcome is generated by
\[
Y_{it}\mid X_{it};\theta, \lambda_{i},f_{t}\sim g\left(\cdot|X_{it};\theta,\lambda_{i},f_{t}\right)
\]
where $g\left(\cdot\right)$ is a known probability density function
with respect to some dominating measure. The log-likelihood function is then given by
\[
\ell\left(W_{it};\theta, \lambda_{i}^{\prime}f_{t}\right)=-\log g\left(Y_{it}\mid X_{it};\theta,\lambda_{i},f_{t}\right)
\]

Since the common factors $F$ and factor loadings $\Lambda$ cannot be separately identified without imposing normalization (see Remark~\ref{remark:normlization}), we impose the same normalization on both $F$ and $\Lambda$ as in \citet*{bai_panel_2009}, without loss of generality. Specifically, we normalize such that $F^{\prime}F/T=\mathbb{I}_{r}$ and
$\Lambda^{\prime}\Lambda/N$ is diagonal with non-increasing diagonal elements. It is important to note that the ultimate asymptotic results
remain unchanged regardless of the specific normalization chosen for
$F$ and $\Lambda$.

To refine the initial estimators, we employ a localized iterative procedure. Starting with the initial estimators $\left(\hat{\theta},\hat{\Lambda},\hat{F}\right)$, the procedure iteratively updates these estimators by searching within a local region of radius $d_{NT}=c\log\left(N\land T\right)\gamma_{NT}$, where $c$ is a positive constant, and $\gamma_{NT}$ represents the convergence rate of the first-step estimator, as specified in Corollary~\ref{cor:main-r} below. 
The procedure is defined as follows: 
\begin{enumerate}
    \renewcommand{\labelenumi}{\textbf{Step \arabic{enumi}.}}
    \item For $m=0$, set $\left(\hat{\theta}^{\left(0\right)},\hat{\Lambda}^{\left(0\right)},\hat{F}^{\left(0\right)}\right)=\left(\hat{\theta},\hat{\Lambda},\hat{F}\right)$,
    the preliminary consistent estimator of $\left(\theta_{0},\Lambda_{0},F_{0}\right)$.

    \item Given $\hat{\theta}^{\left(m\right)}$, update the
    estimators of $\left\{ \Lambda_{0},F_{0}\right\} $ to $\left\{ \hat{\Lambda}^{(m+1)},\hat{F}^{(m+1)}\right\} $
    by solving
    \[
    \left\{ \hat{\Lambda}^{(m+1)},\hat{F}^{(m+1)}\right\} \in\underset{\Lambda\in\mathcal{B}\left(\hat{\Lambda}^{(m)},\sqrt{N}d_{NT}\right),F\in\mathcal{B}\left(\hat{F}^{(m)},\sqrt{T}d_{NT}\right)}{\operatorname{argmin}}\mathcal{L}_{NT}\left(\hat{\theta}^{(m)};\Lambda,F\right)
    \]

    \item Given $\left\{ \hat{\Lambda}^{(m+1)},\hat{F}^{(m+1)}\right\} $,
    update the estimator of $\theta_{0}$ to $\hat{\theta}^{(m+1)}$ according
    to 
    \[
\hat{\theta}^{(m+1)}=\underset{ \theta\in\mathcal{B}\left(\hat{\theta}^{(m)},d_{NT}\right)}{\operatorname{argmin}}\mathcal{L}_{NT}\left(\theta;\hat{\Lambda}^{(m+1)},\hat{F}^{(m+1)}\right)
    \]

    \item Iterate Steps 2-3 until a convergence criterion is met.
\end{enumerate} \section{Theoretical Results\label{sec:main}}
In this section, we examine the asymptotic properties of both the initial and iterative estimators under an asymptotic framework where $N$ and $T$ tend to infinity jointly. We first study the consistency and convergence rate of the initial regularized estimator. We then turn to the second-step iterative estimator and derive its convergence properties and asymptotic distribution.

Recall the definition of $\mathcal{L}_{NT}$ in \eqref{ref:Lnt_Def}. We further define
\[
\mathcal{\bar{L}}\left(\theta,\Pi\right)=\mathbb{E}\left[\mathcal{L}_{NT}\left(\theta,\Pi\right)\right],\quad\mathcal{\tilde{L}}_{NT}\left(\theta,\Pi\right)=\mathcal{L}_{NT}\left(\theta,\Pi\right)-\mathbb{E}\left[\mathcal{L}_{NT}\left(\theta,\Pi\right)\right]
\]
and the excess risk function as 
\[
\mathcal{E}\left(\theta,\Pi\right)=\mathcal{\bar{L}}\left(\theta,\Pi\right)-\mathcal{\bar{L}}\left(\theta_{0},\Pi_{0}\right)
\]
Let $\underline{c}$ and $\overline{c}$ denote generic positive constants
that may vary in their occurrences. 
\begin{assumption}[Sampling and Covariates] \label{assu:data} (i) The data sequences $\left\{ \left(W_{it}\right)\right\} _{i\in[N],t\in[T]}$
obey the model (\ref{ref:Model});\\
(ii) The data on units $\left(W_{it}\right)_{t\in\left[T\right]}$
are independent across $i$, and for each $i$, the sequence $\left(W_{it}\right)_{t\in\left[T\right]}$
is stationary and exponentially $\beta$-mixing with respect to $t$, conditional on $\Pi$. Specifically,
there exist constants $C>0$ and $\mu>0$ such that $\forall q\geqslant1$,
$$\sup_{1\leqslant i\leqslant N}\gamma_{i}(q;W)\leqslant C\exp(-\mu q)$$
where $\gamma_{i}(q;W)\coloneqq\frac{1}{2}\sup_{\bar{t} \geq 1}\sup\sum_{l=1}^{L}\sum_{l'=1}^{L'}\left|\mathrm{P}\left(A_{l}\cap B_{l^{\prime}}\right)-P\left(A_{l}\right)P\left(B_{l^{\prime}}\right)\right|$ with the second supremum over all finite partitions $\left\{A_l\right\}$ in the $\sigma$-field generated by $(\left\{W_{i t}\right\}_{t \leq \bar{t}})$ and $\left\{B_{l^{\prime}}\right\} $ in the $\sigma$-field generated by $(\left\{W_{i t}\right\}_{t \geq \bar{t}+q})$;
\\
 (iii) As $\left(N,T\right)\rightarrow\infty,N/T\rightarrow\kappa^{2},0<\kappa<\infty$;\\
(iv) The covariate $X_{it}$ has bounded support; \\
(v) The rank of $\Pi_0$ is fixed, denoted as $r$.
\end{assumption} 

\begin{assumption}[Lipschitz continuity] \label{assu:Lip} 
Suppose that $\ell\left(w;\theta,\pi\right)$ is Lipschitz continuous in $\left(\theta,\pi\right)$ in the sense that $$\left|\ell\left(w;\theta_{1},\pi_{1}\right)-\ell\left(w;\theta_{2},\pi_{2}\right)\right|\leqslant L\left(w\right)\left[\left\Vert \theta_{1}-\theta_{2}\right\Vert +\left|\pi_{1}-\pi_{2}\right|\right]$$
for a measurable function $L\left(\cdot\right)$ with $\max_{i\in[N],t\in[T]}\mathbb{E}\left[L\left(W_{it}\right){}^{4}\right]<\infty$, where the data sequence $\left\{ \left(W_{it}\right)\right\} _{i\in[N],t\in[T]}$ satisfies Assumption~\ref{assu:data}. 
\end{assumption} 
Assumption \ref{assu:data} imposes basic regularity on the data generating process. Assumption \ref{assu:data}(ii) limits inter-temporal
dependence within the data. Specifically, $\mu$ controls the strength
of this temporal dependence. We adopt the assumption of exponentially
beta-mixing for simplicity, but similar results can be achieved with
polynomially beta-mixing data. If the observations $\left\{ \left(W_{it}\right)\right\} _{i\in[N],t\in[T]}$
are i.i.d. across $i$ and over $t$, our theoretical results will
hold without imposing Assumption \ref{assu:data}(ii). Assumption
\ref{assu:data}(iii) imposes constraints on the relative rates at which $N$ and $T$ tend to infinity.\footnote{For the first-step estimator, it is sufficient that $T^{-1}\log N = o(1)$, which requires only that $T$ increase sufficiently fast relative to $\log N$. This weaker condition guarantees the convergence rate of the first-step estimates and is implied by the joint asymptotic regime $N/T \to \kappa^{2}$ adopted here.}
It is the standard assumption in the large-$T$ panel data literature \citep*{hahn_asymptotically_2002,bai_panel_2009,chen_nonlinear_2021}.
Assumption \ref{assu:data}(iv) is included for convenience, but could
be relaxed to allow covariates with sufficiently light tails. For
instance, if the covariates are standard Gaussian, $c_x$ can grow like
$\sqrt{\log(pNT)}$ with probability approaching one. This avenue is not
pursued here to maintain a clear and concise exposition. Additionally, we focus only on the case of fixed $r$; the analysis of approximately low-rank $\Pi_0$, analogous to approximate sparsity in the Lasso literature, is beyond the scope of this paper. 

Assumption~\ref{assu:Lip} imposes a Lipschitz continuity condition on the loss function. In the specific case where the loss function is globally Lipschitz with a constant independent of $y$ or $x$, this assumption holds trivially. 

\subsection{Consistency}\label{subsec:consistency}
We show that the first-step regularized estimator is consistent, which is the starting point for establishing its rate of convergence. This section states the high-level conditions required for consistency, while Proposition~\ref{prop:ep} and Appendix~\ref{apx:rsc} present primitive conditions.
\begin{assumption}[Identification] \label{assu:id-2}(i) $\theta_{0}$
lies in the interior of $\Theta$, and for each $i$ and $t$, $\pi_{0,it}$
lies in the interior of $\Phi$, where $\Theta$ and $\Phi$ are compact
convex subsets of $\mathbb{R}^{p}$ and $\mathbb{R}$, respectively. Furthermore, $\Phi_\lambda$ and $\Phi_f$ are also compact
convex subsets of $\mathbb{R}^r$;\\
(ii) For
all $\eta>0$, there exists an $\varepsilon>0$ such that for all $N\geq1$
and $T\geq1$,
\[
 \underset{\substack{\left(\theta,\Pi\right)\in\Theta\times\Phi^{N\times T}\\
\left\Vert \theta-\theta_{0}\right\Vert ^{2}+\frac{1}{NT}\left\Vert \Pi-\Pi_{0}\right\Vert _{F}^{2}\geq\eta^{2}
}
}{\inf}\mathcal{E}\left(\theta,\Pi\right)\geq\varepsilon
\]
\end{assumption} 
This is an identification assumption that imposes restrictions on the shape of the excess risk function. When $N$ and $T$ are fixed, the condition is satisfied if the loss function is continuous and $(\theta_{0},\Pi_{0})$
is its unique minimizer. In the context of panel data, this condition needs to be satisfied uniformly in $N$ and $T$. 

Thanks to the penalty term in \eqref{ref:Estimation1}, we can show that with probability approaching one, $\left(\hat{\theta},\hat{\Pi}\right)$ belongs to a restricted set
\[
\mathcal{B}=\left\{ \left(\theta,\Pi\right)\in\Theta\times\Phi^{N\times T}:\left\Vert \Pi\right\Vert _{*}\leq c_{\mathcal{L}}\nu^{-1}+\left\Vert \Pi_{0}\right\Vert _{*}\right\} 
\]
where $c_{\mathcal{L}}=\mathcal{\bar{L}}\left(\theta_{0},\Pi_{0}\right)+1$. It is formally established in the Appendix.
This result is important because it allows the analysis to restrict attention to a smaller set $\mathcal{B}$ within the parameter space.
\begin{assumptionUC} \label{assu:obj-2}$\sup_{\left(\theta,\Pi\right)\in\mathcal{B}}\left|\mathcal{L}_{NT}\left(\theta,\Pi\right)-\mathcal{\bar{L}}\left(\theta,\Pi\right)\right|=o_{p}\left(1\right)$
\end{assumptionUC} 
Assumption~\ref{assu:obj-2} assumes uniform convergence. In the low-dimensional context, a similar condition is usually required on a compact set which does not depend on $N$ and $T$. The main difference here is that the radius of $\mathcal{B}$ grows with $N$ and $T$. We provide sufficient conditions for Assumption~\ref{assu:obj-2} to hold in Proposition~\ref{prop:ep}.

\begin{lemma}[Consistency] \label{lem:consistency-2}Under Assumptions
\ref{assu:id-2} and \ref{assu:obj-2}, if $\nu\left\Vert \Pi_{0}\right\Vert _{*}=o\left(1\right)$,
then we have $\left\Vert \hat{\theta}-\theta_{0}\right\Vert ^{2}+\frac{1}{NT}\left\Vert \hat{\Pi}-\Pi_{0}\right\Vert _{F}^{2}=o_{p}\left(1\right)$.
\end{lemma} 
The result depends on the condition that $\nu\left\Vert \Pi_{0}\right\Vert _{*}=o(1)$,
which implies that the added penalty term has a negligible effect
on the objective function when evaluated at the true values. This condition on $\nu$ is satisfied in the subsequent theorems and corollaries.

\subsection{Rate of Convergence}\label{subsec:rate}
Similar to other high-dimensional settings such as those in \citet{negahban_unified_2012}, we impose an invertibility condition involving another restricted set, denoted as $\mathcal{A}$, which we describe next and use to derive our main results.
Before stating the next condition, we first introduce some notation. 

Let $\Pi_{0}=UDV^{\prime}$ represent the singular value decomposition
(SVD) of $\Pi_{0}$, where $U$ and $V$ are the matrices of singular vectors corresponding to all singular values. In particular, let $U_{0}\in\mathbb{R}^{N\times r}$ and $V_{0}\in\mathbb{R}^{T\times r}$
denote the columns of $U$ and $V$ associated with the non-zero singular
values. 
Let $P_U=U_0U_0'$ and $P_V=V_0V_0'$. For an $N\times T$ matrix $\Delta$, we define the operators
\begin{equation}
\mathcal{P}\left(\Delta\right)\coloneqq P_U\Delta+\Delta P_V-P_U\Delta P_V,\qquad\mathcal{M}\left(\Delta\right)\coloneqq\Delta-\mathcal{P}\left(\Delta\right)\label{eq:PM-def}
\end{equation}
The operator $P(\cdot)$ is the projection onto the tangent space associated with the
row and column spaces of $\Pi_0$, while $M(\cdot)$ is the projection onto its orthogonal
complement.

Next, we define the restricted set as 
\begin{equation}
\mathcal{A}=\left\{ \left(\delta,\Delta\right)\in\mathbb{R}^{p}\times\mathbb{R}^{N\times T},\text{s.t.}\|\Delta\|_{*}-4\|\mathcal{P}\left(\Delta\right)\|_{*}-\sqrt{NT}\left\Vert \delta\right\Vert \leq0\right\} \label{ref:cone}
\end{equation}

It can be shown that $\left(\hat{\theta}-\theta_{0},\hat{\Pi}-\Pi_{0}\right)$ lies in this ``cone'' under certain conditions. 
Namely, the set contains matrices $\Pi$ that are close to $\Pi_{0}$,
in the sense that the part that cannot be explained by $\lambda_{0i}$
and $f_{0t}$ is small in terms of nuclear norm. 

\begin{assumption}[Restricted
Strong Convexity] \label{assu:rsc} For all $\left(\delta,\Delta\right)\in\mathcal{A}$, there exists a universal constant $c_{RSC}>0$ such that 
\[
\mathcal{E}\left(\theta_{0}+\delta,\Pi_{0}+\Delta\right)\geq c_{RSC}\left\Vert \delta\right\Vert ^{2}+\frac{c_{RSC}}{NT}\left\Vert \Delta\right\Vert _{F}^{2}
\]

\end{assumption} 
Assumption \ref{assu:rsc} is based on Restricted Strong Convexity (RSC), which relaxes the definition of strong
convexity by only needing strong convexity in certain directions or over a subset of the ambient
space \citep*{negahban_unified_2012,wainwright_high-dimensional_2019}. 
This assumption represents a version of a widely used condition in the matrix completion literature, although verifying it typically requires imposing additional structure on the parameter space. See also the discussions in \citet*{moon_nuclear_2019} and \citet*{Zeleneev_tractable_2026}, following their Assumptions 1 and A.3, respectively. 
More discussion of Assumption~\ref{assu:rsc} is provided in Appendix~\ref{apx:rsc}.

Finally, to control the empirical risk associated with the estimation errors, we define the set 
\[
\mathcal{V}=\left\{ (\theta,\Pi)\in\mathcal{B}:\left\Vert \delta\right\Vert ^{2}+\frac{1}{NT}\left\Vert \Delta\right\Vert _{F}^{2}\leq c_{l}\right\}.
\]
We also introduce the norm $\rho(\cdot,\cdot)$, defined as
\[
\rho(\delta,\Delta)=\left[\left\Vert \delta\right\Vert ^{2}+\frac{1}{NT}\left\Vert \Delta\right\Vert _{*}^{2}\right]^{1/2}.
\]
According to Lemma~\ref{lem:consistency-2}, it suffices to analyze the convergence rate of the empirical risk within the restricted set $\mathcal{V}$.

\begin{proposition}
\label{prop:ep}Under Assumptions \ref{assu:data} and \ref{assu:Lip},
and with $\nu=c_{\nu}\frac{\psi_{NT}c_{\varepsilon,NT}}{\sqrt{NT}}$
for some constant $c_{\nu}\geq2$, 
there exist positive sequences $\left\{ \psi_{NT}\right\} $, and $\left\{ c_{\varepsilon,NT}\right\}$ such that 
\[
\lim_{\left(N,T\right)\rightarrow\infty}\mathrm{P}\left(\sup_{\left(\theta,\Pi\right)\in\mathcal{V}}\frac{\left|\mathcal{\tilde{L}}_{NT}\left(\theta,\Pi\right)-\mathcal{\tilde{L}}_{NT}\left(\theta_{0},\Pi_{0}\right)\right|}{\rho\left(\theta-\theta_{0},\Pi-\Pi_{0}\right)\vee c{}_{\varepsilon,NT}}\leq\psi_{NT}c_{\varepsilon,NT}\right)=1
\]
where $c_{T}=\lceil2\mu^{-1}\log\left(NT\right)\rceil$, $d_{T}=\lfloor T/\left(2c_{T}\right)\rfloor$,
$c_{\varepsilon,NT}=\frac{\sqrt{c_{T}}}{\sqrt{N\land d_{T}}}$ and
$\psi_{NT}/\left(C_{L}\log\left(c_{T}+1\right)\right)\rightarrow\infty$. Moreover, Assumption~\ref{assu:obj-2} holds.
\end{proposition} 
The proof of Proposition~\ref{prop:ep} is based on the empirical process theory and does not rely on the differentiability of the loss function. It derives sufficient conditions under which Assumption~\ref{assu:obj-2} holds and provides a stochastic bound that will be used for the subsequent theorem.
Specifically, we have $c_{\varepsilon,NT}=O\left(\sqrt{\log\left(NT\right)/\left(N\wedge T\right)}\right)$, since, under our exponential $\beta$-mixing setting, the dependence length satisfies $c_T=O(\log (N T))$. The scaling sequence $\psi_{N T}$ is allowed to diverge slowly, at a rate exceeding $O(\log \left(\log (N T)\right))$, so that the probabilistic bound in Proposition~\ref{prop:ep} holds uniformly over the local parameter space $\mathcal{V}$.

We now present our first main result for estimating $\left(\theta_{0},\Pi_{0}\right)$.
\begin{theorem}[Convergence Rate] \label{thm:main-3}
Under Assumptions~\ref{assu:data}-\ref{assu:rsc}, let $\nu=c_{\nu}\frac{\psi_{NT}c_{\varepsilon,NT}}{\sqrt{NT}}$
for some constant $c_{\nu}\geq2$. With probability approaching one, we have
\[\left\Vert \hat{\theta}-\theta_{0}\right\Vert ^{2}+\frac{1}{NT}\left\Vert \hat{\Pi}-\Pi_{0}\right\Vert _{F}^{2}\lesssim \frac{\psi_{NT}^{2}rc_{T}}{N\land d_{T}}\]
where $c_{\varepsilon,NT}$ and $\psi_{NT}$ are specified in Proposition
\ref{prop:ep}. \end{theorem} 

Theorem~\ref{thm:main-3} establishes the convergence rates of the estimation errors for $\hat{\theta}$ and $\hat{\Pi}$ in the $\ell_{2}$ norm. Here, $\theta$ is a low-dimensional parameter vector, while $\Pi$ is a high-dimensional matrix with $NT$ elements; hence, the Frobenius norm of $\Pi$ is normalized by $1/\sqrt{NT}$ to ensure comparability. For simplicity, consider the case of i.i.d. panel data across $i$ over $t$. When the regularization parameter is set to $\nu \asymp \frac{\sqrt{N}\lor\sqrt{T}}{NT}$, the convergence rate of our estimator under the Euclidean norm is of the order of $\sqrt{1 / (N \land T)}$, with all other factors ignored. These results are consistent with previous work on penalized mean and quantile regression models for panel data \citep*{athey_matrix_2021, feng_nuclear_2023, moon_nuclear_2019, belloni_high_2023}, while the present paper develops a unified M-estimation framework that extends these results to a broader class of models. Finally, note that Theorem~\ref{thm:main-3} does not require differentiability of the loss function; however, subsequent asymptotic normality results rely on smoothness conditions to establish limiting distributions.

Next, we impose an assumption on the common factors and factor loadings.

\begin{assumption}[Strong factors] \label{assu:factor}
Let $\frac{1}{N} \sum_{i=1}^N \lambda_{0i} \lambda_{0i}^{\prime} \xrightarrow{p} \Sigma_\Lambda > 0$ and
$\frac{1}{T} \sum_{t=1}^T f_{0t} f_{0t}^{\prime} \xrightarrow{p} \Sigma_F > $
for some positive definite matrices \(\Sigma_\Lambda, \Sigma_F \in \mathbb{R}^{r \times r}\). Furthermore, the matrix $\Pi_0$ has $r$ distinct nonzero singular values.
\end{assumption}

Assumption~\ref{assu:factor} formalizes the strong factor condition commonly imposed in the literature; see, for example, \citet*{bai_panel_2009} and \citet*{chen_nonlinear_2021}. The requirement that the singular values $\sigma_1, \ldots, \sigma_r$ are distinct is analogous to Assumption G in \citet*{bai_inferential_2003}, and closely related to Assumption 4.2 in \citet*{chernozhukov_inference_2023}. 
In particular, Theorem~\ref{thm:main-3} does not depend on Assumption~\ref{assu:factor}.\footnote{\citet*{armstrong_robust_2023} use NNR and
study robust inference for weak factors in the linear panel
model.}

The following corollary establishes the consistency of $\hat{r}$ and the mean squared convergence rates of $\hat{\Lambda}$ and $\hat{F}$.
\begin{corollary} \label{cor:main-r}
Suppose that the conditions
in Theorem \ref{thm:main-3} and Assumption \ref{assu:factor} hold. Without loss of generality,
we impose the normalization that $F^{\prime}F/T=\mathbb{I}_{r}$ and
$\Lambda^{\prime}\Lambda/N$ is diagonal with nonincreasing diagonal
elements. Then, \\
(i) $P\left(\hat{r}=r\right)\rightarrow1$
as $N,T\rightarrow\infty$; \\
(ii) $\frac{1}{\sqrt{N}}\left\Vert \hat{\Lambda}-\Lambda_{0}\hat{S}\right\Vert _{F}=O_{p}\left(\gamma_{NT}\right)$,
and $\frac{1}{\sqrt{T}}\left\Vert \hat{F}-F_{0}\hat{S}\right\Vert _{F}=O_{p}\left(\gamma_{NT}\right)$,
where $\hat{S}=\operatorname{sgn}\left(\hat{F}^{\prime} F_0 \right)$ and $\gamma_{NT}=\psi_{NT}c_{\varepsilon,NT}$, with $\psi_{NT}$ and $c_{\varepsilon,NT}$ specified in Proposition \ref{prop:ep}.\end{corollary} 
Nuclear-norm regularization does not require prior knowledge or specification of the number of factors $r$. Corollary \ref{cor:main-r} indicates that $\hat{r}$ is a consistent estimator of $r$. Thus, in what follows, we assume that the number of factors has been correctly selected. 
Without loss of generality, we further assume that $\hat{\mathrm{S}}=\mathbb{I}_r$ to simplify the notation.

 \subsection{Asymptotic Normality}\label{sec:inf}

Here, we use the shorthand notation $\ell_{it}\left(\theta,\pi\right)=\ell\left(W_{it};\theta, \pi\right)$
for convenience, where $\theta\in\Theta$ and $\pi\in\Phi$. Additionally, we omit the function arguments when
they are evaluated at the true parameter values $\left(\theta_{0},\pi_{0,it}\right)$,
e.g., $\ell_{it}=\ell\left(W_{it};\theta_0, \pi_{0,it}\right)$.

To study the asymptotic behavior of the iterative estimator $\hat{\theta}^{\left(m+1\right)}$, it is necessary to introduce additional quantities that characterize how the parameter $\theta$ interacts with the incidental components $\pi_{it}$. 

Let $\Xi_{it}$ denote a $p$-dimensional
vector defined by the following population weighted least squares
projection for each component of $\mathbb{E}\left[\partial_{\theta_{k}\pi}\ell_{it}\right]$
onto the space spanned by the incidental parameters, under a metric
given by $\mathbb{E}\left[\partial_{\pi^{2}}\ell_{it}\right]$. Specifically,
$\Xi_{it,k}=\lambda_{i,k}^{*\prime}f_{0t}+\lambda_{0i}^{\prime}f_{t,k}^{*}$,
where $\left(\lambda_{i,k}^{*},f_{t,k}^{*}\right)$ is defined as
the solution to the following optimization problem,
\[
\left(\lambda_{i,k}^{*},f_{t,k}^{*}\right)\in\underset{\lambda_{i,k},f_{t,k}}{\operatorname{argmin}}\sum_{i,t}\mathbb{E}\left[\partial_{\pi^{2}}\ell_{it}\right]\left(\frac{\mathbb{E}\left[\partial_{\theta_{k}\pi}\ell_{it}\right]}{\mathbb{E}\left[\partial_{\pi^{2}}\ell_{it}\right]}-\lambda_{i,k}^{\prime}f_{0t}-\lambda_{0i}^{\prime}f_{t,k}\right)^{2}
\]
In addition, we define the operator $D_{\theta\pi^{q}}\ell_{it}=\partial_{\theta\pi^{q}}\ell_{it}-\Xi_{it}\partial_{\pi^{q+1}}\ell_{it}$ for $q=0,1,2$. Intuitively, these operators remove the component of $\partial_{\theta\pi^q}\ell_{it}$
that can be explained by the individual and time fixed effects.
Furthermore, define a $p \times p$ matrix
$$
\bar{W}_{NT}:=\frac{1}{NT}\sum_{i=1}^N \sum_{t=1}^T \mathbb{E}\left[\partial_{\theta \theta} \ell_{i t}-\partial_{\pi^2} \ell_{i t} \Xi_{i t} \Xi_{i t}^{\prime}\right]
$$
which will be used to characterize the information matrix for $\theta$.

We now introduce the regularity conditions that are required for the asymptotic results.
\begin{assumption} \label{assu:inf}
Let $\mathcal{B}_{\epsilon}^{0} \subset \mathbb{R}^{p+1}$ denote a bounded set that contains a $\epsilon$-neighborhood of $(\theta_0, \pi_{0,it})$ for all $i \in [N]$ and $t \in [T]$, uniformly over $N$ and $T$. \\
 (i) The function $\left(\theta,\pi\right)\mapsto\ell_{it}(\theta,\pi)$
is four times continuously differentiable over $\mathcal{B}_{\epsilon}^{0}$.
The partial derivatives of $\ell_{it}(\theta,\pi)$ with respect to
the elements of $\left(\theta,\pi\right)$ up to the fourth order
are bounded in absolute value over $\left(\theta,\pi\right)\in\mathcal{B}_{\epsilon}^{0}$
by a function $M\left(W_{it}\right)>0$ such that $\max_{i\in[N],t\in[T]}\mathbb{E}\left[\left|M\left(W_{it}\right)\right|^{8+\iota}\right]$
is uniformly bounded for some $\iota>0$ over all $N$ and $T$; \\
(ii) There exist positive constants $b_{\min}$ and $b_{\max}$ such
that for all $\left(\theta,\pi\right)\in\mathcal{B}_{\epsilon}^{0}$,
$b_{\min}\leq\mathbb{E}\left[\partial_{\pi^{2}}\ell_{it}\left(\theta,\pi\right)\right]\leq b_{\max}$;\\
(iii) The matrix $\bar{W}_{NT}$ is uniformly positive definite, i.e., $\inf _{N, T} \sigma_{\min }\left(\bar{W}_{N T}\right) \geq c>0$.
\end{assumption} 
Assumption~\ref{assu:inf} strengthens the conditions
used for Theorem~\ref{thm:main-3}. 

Assumption~\ref{assu:inf}(i) and (ii) require the loss
function to exhibit sufficient smoothness and local strong convexity. Assumption~\ref{assu:inf}(iii) is a generalized noncollinearity condition, similar to those commonly imposed in the factor literature, such as Assumption~A in \citet*{bai_panel_2009} and Assumption~1(vi) in \citet*{chen_nonlinear_2021}. In the linear case, Assumption~\ref{assu:inf}(iii) simplifies to requiring $ \sum_{i=1}^N \sum_{t=1}^T \mathbb{E}\left[\left(X_{i t}-\Xi_{i t}\right)\left(X_{i t}-\Xi_{i t}\right)^{\prime}\right]$ to be positive definite. Intuitively, this condition ensures that the covariates exhibit sufficient variation across individuals and over time, thus guaranteeing the identification of $\theta$.
\begin{theorem} \label{thm:inf-conv}
Suppose Assumptions~\ref{assu:data}-\ref{assu:inf} hold. Then

(i) For $m=0,1,2,....$, we have 
\[
\begin{aligned}\hat{\theta}^{\left(m+1\right)}-\theta_{0} & =C^{\left(0\right)}\left(\hat{\theta}^{\left(m\right)}-\theta_{0}\right)+ C^{\left(1\right)}\\
 & +o_{p}\left(\left\Vert \hat{\theta}^{\left(m\right)}-\theta_{0}\right\Vert \right)+o_{p}\left(\left(NT\right)^{-1/2}\right)
\end{aligned}
\]
where $C^{\left(0\right)}\stackrel{p}{\longrightarrow}\bar{C}^{\left(0\right)}$
with $\|\bar{C}^{\left(0\right)}\|\in[0,1)$, and $C^{\left(1\right)}=O_{p}\left(\left(NT\right)^{-1/2}\right)$;

(ii) For any $m\geq-\left(\frac{1}{2}\log\left(NT\right)+\log(\gamma_{NT})\right)/\log\left(\|\bar{C}^{\left(0\right)}\|\right)-1$
with $\bar{C}^{(0)}\neq0$, $$\hat{\theta}^{(m+1)}-\theta_{0}=O_{p}\left((NT)^{-1/2}\right)$$
where $\gamma_{NT}=\psi_{NT}c_{\varepsilon,NT}$, with $\psi_{NT}$ and $c_{\varepsilon,NT}$ specified in Proposition \ref{prop:ep}.\end{theorem}

Theorem \ref{thm:inf-conv} examines the numerical convergence properties
of $\hat{\theta}^{\left(m+1\right)}$. Specifically, Theorem \ref{thm:inf-conv}(i)
guarantees the convergence of the iterative procedure, while
Theorem \ref{thm:inf-conv}(ii) indicates that after $O\left(\log\left(NT\right)\right)$
iterations, the iterative estimator $\hat{\theta}^{\left(m+1\right)}$
achieves the desired convergence rate for inference. It converges
to the true value $\theta_{0}$ much faster than $\gamma_{NT}$, which is the rate at which the initial estimator $\hat{\theta}$ converges in probability.
The upper bound of our contraction parameter, specifically the spectral norm of $\bar{C}^{(0)}$, is crucial in the iterative procedure. The closer $\|\bar{C}^{\left(0\right)}\|$
is to $0$, the faster the contraction and the numerical convergence of
the iterative procedure. Conversely, if $\|\bar{C}^{\left(0\right)}\|$
is close to $1$, the numerical convergence of $\hat{\theta}^{\left(m+1\right)}$
is slow. 

The following theorem establishes the asymptotic distribution of the second-step estimator $\hat{\theta}^{\left(m+1\right)}$.

\begin{theorem} \label{thm:inf-dist}Suppose Assumptions~\ref{assu:data}-\ref{assu:inf} hold, and the following limits exist, 
\[
\begin{aligned}\bar{W}_{\infty} & \coloneqq\lim_{N,T\rightarrow\infty}\frac{1}{NT}\sum_{i=1}^{N}\sum_{t=1}^{T}\mathbb{E}\left[\partial_{\theta\theta}\ell_{it}-\partial_{\pi^{2}}\ell_{it}\Xi_{it}\Xi_{it}^{\prime}\right]\\
\bar{B}_{\infty} & \coloneqq-\lim_{N,T\rightarrow\infty}\frac{1}{N}\sum_{i=1}^{N}\sum_{t=1}^{T}\left\{ \left[\sum_{\tau=t}^{T}f_{0t}^{\prime}H_{\left(\lambda\lambda\right)i}^{-1}f_{0\tau}\mathbb{E}\left[\partial_{\pi}\ell_{it}D_{\theta\pi}\ell_{i\tau}\right]\right]+\frac{1}{2}f_{0t}^{\prime}H_{\left(\lambda\lambda\right)i}^{-1}f_{0t}\mathbb{E}\left[D_{\theta\pi^{2}}\ell_{it}\right]\right\} \\
\bar{D}_{\infty} & \coloneqq-\lim_{N,T\rightarrow\infty}\frac{1}{T}\sum_{i=1}^{N}\sum_{t=1}^{T}\lambda_{0i}^{\prime}H_{\left(ff\right)t}^{-1}\lambda_{0i}\mathbb{E}\left[\partial_{\pi}\ell_{it}D_{\theta\pi}\ell_{it}+\frac{1}{2}D_{\theta\pi^{2}}\ell_{it}\right]
\end{aligned}
\]
where $H_{\left(\lambda\lambda\right)i}=\sum_{t=1}^{T}\mathbb{E}\left[\partial_{\pi^{2}}\ell_{it}\right]f_{0t}f_{0t}^{\prime}$, $H_{\left(ff\right)t}=\sum_{i=1}^{N}\mathbb{E}\left[\partial_{\pi^{2}}\ell_{it}\right]\lambda_{0i}\lambda_{0i}^{\prime}$, and $\bar{W}_{\infty}>0$. Then, for any $m\geq-\frac{1}{2}\log(NT)/\log\left(\|\bar{C}^{\left(0\right)}\|\right)-1$,
it follows that
\begin{equation}
\sqrt{NT}\left(\hat{\theta}^{\left(m+1\right)}-\theta_{0}-\frac{1}{T}\bar{W}_{\infty}^{-1}\bar{B}_{\infty}-\frac{1}{N}\bar{W}_{\infty}^{-1}\bar{D}_{\infty}\right)\stackrel{d}{\longrightarrow}\mathcal{N}\left(0,\bar{W}^{-1}_{\infty}\right)
\end{equation}
\end{theorem}
Theorem~\ref{thm:inf-dist} indicates that $\hat{\theta}^{\left(m+1\right)}$ contains two asymptotic bias terms associated with $\frac{1}{T}\bar{B}_{\infty}$ and $\frac{1}{N}\bar{D}_{\infty}$, respectively. For convex objective functions, the estimator is asymptotically equivalent to the one obtained from \eqref{ref:Model1}, as studied by \citet*{chen_nonlinear_2021}.

\subsection{Bias Correction}\label{sec:ABC}
This section describes methods for removing the asymptotic bias of the second‐step estimator. We first outline the analytical bias‐correction procedure and then discuss alternative approaches, including the split‐panel jackknife and the bootstrap.

The analytical correction is constructed using sample analogs of the expressions in Theorem~\ref{thm:inf-dist}, replacing the true values of $(\theta,\pi)$ with their second‐step estimates. Both analytical bias correction and variance estimation require consistent estimators of the quantities $\bar{B}{\infty}$, $\bar{D}{\infty}$, and $\bar{W}_{\infty}$ defined in Theorem~\ref{thm:inf-dist}. Let $\hat{B}$, $\hat{D}$, and $\hat{W}$ denote the corresponding sample analogs, obtained by substituting sample averages for expectations and replacing the true parameters with their second‐step estimates. For example,
$$\hat{W}=\frac{1}{NT}\sum_{i=1}^{N}\sum_{t=1}^{T}\partial_{\theta\theta}\hat{\ell}_{it}-\partial_{\pi^{2}}\hat{\ell}_{it}\hat{\Xi}_{it}\hat{\Xi}_{it}^{\prime}$$
where $\partial_{\theta\theta}\hat{\ell}_{it}=\partial_{\theta\theta}\ell_{it}\left(W_{it};\hat{\theta},\hat{\lambda_{i}}^{\prime}\hat{f_{t}}\right)$, $\partial_{\pi^{2}}\hat{\ell}_{it}=\partial_{\pi^{2}}\ell_{it}\left(W_{it};\hat{\theta},\hat{\lambda_{i}}^{\prime}\hat{f_{t}}\right)$, and $\hat{\Xi}_{it}$ is a $p$-dimensional vector with elements $\hat{\Xi}_{it,k}=\lambda_{i,k}^{\#\prime}\hat{f}_{t}+\hat{\lambda}_{i}^{\prime}f_{t,k}^{\#}$ where the pair $\left(\lambda_{i,k}^{\#},f_{t,k}^{\#}\right)$ is defined by
$$\left(\lambda_{i,k}^{\#},f_{t,k}^{\#}\right)\in\underset{\lambda_{i,k},f_{t,k}}{\operatorname{argmin}}\sum_{i,t}\left(\partial_{\pi^{2}}\hat{\ell}_{it}\right)\left(\frac{\partial_{\theta_{k}\pi}\hat{\ell}_{it}}{\partial_{\pi^{2}}\hat{\ell}_{it}}-\lambda_{i,k}^{\prime}\hat{f}_{t}-\hat{\lambda}_{i}^{\prime}f_{t,k}\right)^{2}$$
Once these sample analogs are constructed, the analytical bias correction of $\hat{\theta}$ reads
$$\hat{\theta}_{ABC}=\hat{\theta}^{\left(m+1\right)}-\frac{1}{T}\hat{W}^{-1}\hat{B}-\frac{1}{N}\hat{W}^{-1}\hat{D}$$

\begin{theorem}\label{thm:inf-ABC} 
Suppose that the assumptions in Theorem~\ref{thm:inf-dist} hold. Then $\hat{W}\stackrel{p}{\longrightarrow}\bar{W}_{\infty}$ and $$\sqrt{NT}\left(\hat{\theta}_{ABC}-\theta_{0}\right)\stackrel{d}{\longrightarrow}\mathcal{N}\left(0,\bar{W}_{\infty}^{-1}\right)$$
\end{theorem}

Although the analytical correction offers a convenient closed‐form adjustment, it requires estimating high‐order derivatives and moment matrices that may be sensitive to model specification and sample size. In practice, alternative approaches—such as the split‐panel jackknife and the bootstrap—can remove the leading bias terms without explicit analytical derivations.

Following \citet*{dhaene_split-panel_2015} and \citet*{chen_nonlinear_2021}, bias correction can also be implemented using the split‐panel jackknife method. Let $\hat{\theta}_{N / 2, T}^1$ and $\hat{\theta}_{N / 2, T}^2$ be the second-step estimators based on the subsamples $\{(i, t): i=1, \ldots, N / 2 ; t=1, \ldots, T\}$ 
and $\{(i, t): i=N / 2+1, \ldots, N ; t=1, \ldots, T\}$, respectively. 
Similarly, let $\hat{\theta}_{N, T / 2}^1$ and $\hat{\theta}_{N, T / 2}^2$ be the estimators obtained from the subsamples $\{(i, t): i=1, \ldots, N ; t=1, \ldots, T / 2\}$ 
and 
$\{(i, t): i=1, \ldots, N ; t=T / 2+1, \ldots, T\}$, respectively. The jackknife bias‐corrected estimator is
\[
\hat{\theta}_{\text{JBC}} = 3 \hat{\theta} 
- \frac{1}{2}\left(\hat{\theta}_{N / 2, T}^1+\hat{\theta}_{N / 2, T}^2\right)
- \frac{1}{2}\left(\hat{\theta}_{N, T / 2}^1+\hat{\theta}_{N, T / 2}^2\right).
\]
Under suitable homogeneity and stationarity conditions ensuring that the asymptotic biases of all estimators converge to the same limit,
\[
\sqrt{N T}\left(\hat{\theta}_{\text{JBC}}-\theta_0\right)
\xrightarrow{d} 
\mathcal{N}\left(0,\bar{W}_{\infty}^{-1}
\right)
\]
Other bias‐correction methods include the leave‐one‐out jackknife of \citet*{hughes_estimating_2022} and the bootstrap procedures of \citet*{higgins_bootstrap_2024} and \citet*{sun_k-step_2010}. However, these approaches are not directly applicable to models with interactive fixed effects. Extending them, particularly adapting the bootstrap of \citet*{higgins_bootstrap_2024} to account for time-specific or interactive effects, remains an interesting avenue for future research. \section{Monte Carlo Simulations}\label{sec:MC}
In this section, we assess the finite sample performance of our proposed approach using Monte Carlo simulations.  Section~\ref{subsec:MC_imp} outlines the main steps of the first-step estimation algorithm. Section~\ref{subsec:nu} discusses the choice of the tuning parameter $\nu$. Section~\ref{subsec:MC} presents the simulation designs and results.

\subsection{Implementation}\label{subsec:MC_imp}
This section outlines the computational procedures for the proposed estimation algorithms. We first present the single-index panel model and its implementation; the binary logit specification in Example \ref{ref:ex3} is treated as a special case. We then describe the random coefficient logit model in Example \ref{ref:ex4}. Additional details and formal descriptions can be found in Section~\ref{sec:alg}.

\paragraph{Single-Index Panel Model}
Algorithm~\ref{alg:logit} summarizes the first-stage estimation procedure for the class of single-index panel models, which includes the binary logit specification in Example~\ref{ref:ex3} as a special case. 
We introduce a slack variable $Z_{\Pi}$ to separate the low-rank interactive effects from the linear index component. This reformulation allows the optimization problem in \eqref{ref:Estimation1} to be expressed in an equivalent and computationally convenient form:
\begin{align}
 & \min_{\theta\in\mathbb{R}^{p},\Pi\in\mathbb{R}^{N\times T}}\frac{1}{NT}\sum_{i=1}^{N}\sum_{t=1}^{T}\ell\left(Y_{it},V_{it}\right)+\nu\|\Pi\|_{*}\nonumber \\
 & \qquad\text{s.t. }V=\sum_{j=1}^{p}X_{j}\theta_{j}+Z_{\Pi},\quad Z_{\Pi}-\Pi=0\label{eq:est-alg-main}
\end{align}
To solve the minimization problem (\ref{eq:est-alg-main}), we use a modified Alternating Direction Method of Multipliers (ADMM) algorithm, which relies on the following
augmented Lagrangian
\begin{align*}
\mathscr{L}\left(\theta,\Pi,V,Z_{\Pi},U_{p},U_{v}\right)	&=\frac{1}{NT}\sum_{i=1}^{N}\sum_{t=1}^{T}\ell\left(Y_{it},V_{it}\right)+\nu\|\Pi\|_{*} \\
&+\frac{\eta}{2NT}\left\Vert V-\sum_{j=1}^{p}X_{j}\theta_{j}-Z_\Pi+U_{p}\right\Vert _{F}^{2}+\frac{\eta}{2NT}\left\Vert Z_\Pi-\Pi+U_{v}\right\Vert _{F}^{2}
\end{align*}
Here, $U_{p}$ and $U_{v}$ are the scaled dual variables corresponding to the linear constraints $V=\sum_{j=1}^{p}X_{j}\theta_{j}+Z_{\Pi}$
and $Z_{\Pi}-\Pi=0$, and $\eta>0$ is the penalty parameter for constraint
violations. Due to the separability of the parameters in $\mathscr{L}$,
the ADMM algorithm proceeds by iteratively minimizing the augmented
Lagrangian in blocks with respect to the original variables, in this
case $\left(V,\theta,\Pi\right)$ and $Z_{\Pi}$, and then updating
the dual variables $U_{p}$ and $U_{v}$. 

\begin{algorithm}[h]
\caption{ALM Algorithm: Single-Index Panel Model}
\label{alg:logit}
\begin{algorithmic}[1]

\State \textbf{Input}: Observed data $\left(Y_{it},X_{it}\right)$, and initialization for $V^{(0)}$, $\theta^{(0)}$, $\Pi^{(0)}$, $Z_{\Pi}^{(0)}$, $U_{p}^{(0)}$, and $U_{v}^{(0)}$, parameters $\nu$, $\eta$.
\State \textbf{Initialize}: $k = 0$
\While{not converged}
    \State \textbf{Update parameters at iteration $k+1$:} 

    \State $V_{it}^{(k+1)} \gets V_{it}^{(k)} - h(Y_{it},V_{it}^{(k)}) - \eta (V_{it}^{(k)} - X_{it}^{\prime} \theta^{(k)} - Z_{\Pi, it}^{(k)} + U_{p, it}^{(k)})$
    \State where $h(y,v)=\nabla_v \ell(y,v)$
    
    \State $\theta^{(k+1)} \gets \left( \sum_{i=1}^{N} \sum_{t=1}^{T} X_{it} X_{it}^{\prime} \right)^{-1} \left( \sum_{i=1}^{N} \sum_{t=1}^{T} X_{it} A_{it}^{(k)}  \right)$
    \State where $A^{\left(k\right)}=V^{\left(k+1\right)}-Z^{\left(k\right)}+U_{p}^{\left(k\right)}$
    \State $\Pi^{(k+1)} \gets S_{NT\frac{\nu}{\eta}}( Z^{(k)} + U_v^{(k)})$\label{line:SVT}

    \State $Z_{\Pi}^{(k+1)} \gets \frac{1}{2} ( V^{(k+1)} - \sum_{j=1}^{p} X_j \theta_j^{(k+1)} + U_p^{(k)} + \Pi^{(k+1)} - U_v^{(k)})$

    \State $U_p^{(k+1)} \gets V^{(k+1)} - \sum_{j=1}^{p} X_j \theta_j^{(k+1)} - Z_{\Pi}^{(k+1)} + U_p^{(k)}$
    \State $U_v^{(k+1)} \gets Z_{\Pi}^{(k+1)} - \Pi^{(k+1)} + U_v^{(k)}$

    \State $k \gets k + 1$

\EndWhile

\State \textbf{Output}: Updated parameters $\theta^{(k+1)}$, $\Pi^{(k+1)}$.

\end{algorithmic}
\end{algorithm}

 As shown in Line \ref{line:SVT} of Algorithm \ref{alg:logit}, $\Pi$ is updated via singular value thresholding (SVT). Specifically, recall that the singular value decomposition
(SVD) of a matrix $A$ is 
\[
A=U_{A}\Sigma_{A}V_{A}^{\prime}\in\mathbb{R}^{N\times T},\quad\Sigma_{A}=\operatorname{diag}\left(\left\{ \sigma_{s}\right\} _{s\in\left[N\land T\right]}\right)
\]
For $\iota>0$, the soft-thresholding operator $S_{\iota}\left(\cdot\right)$
is defined as
\[
S_{\iota}\left(A\right)\coloneqq U_{A}S_{\iota}\left(\Sigma_{A}\right)V_{A}^{\prime},\quad S_{\iota}\left(\Sigma_{A}\right)=\operatorname{diag}\left(\left\{ \sigma_{s}-\iota\right\} _{+}\right)
\]
where $t_{+}$ is the positive part of $t$, namely, $t_{+}=\max\left(0,t\right)$. See Theorem 2.1 in \citet*{cai_singular_2008} for details. 

\paragraph{Random Coefficient Logit Panel}

Algorithm~\ref{alg:RC} outlines the first-stage estimation algorithm for the binary random coefficient logit panel model in Example~\ref{ref:ex4} using a majorization-minimization (MM) approach. For simplicity, we assume that all coefficients $\beta$ are random, but incorporating fixed coefficients, such as $\beta_{j,it} = \beta_j$ for some $j$, is straightforward. Our algorithm extends the work of \citet{train_discrete_2009} and \citet{james_mm_2017} by introducing a novel surrogate function to handle the penalty term. Specifically, to update $\Pi$ at step $k+1$, we use the surrogate function $Q\left(\Pi|\theta^{\left(k\right)},\Pi^{(k)}\right)$, defined as
$$
\begin{aligned}Q\left(\Pi\mid\theta^{(k)},\Pi^{(k)}\right)= & \mathcal{L}_{NT}\left(\theta^{(k)},\Pi^{(k)}\right)-\frac{1}{2NT}\sum_{i=1}^{N}\sum_{t=1}^{T}\left(\sum_{r=1}^{R}w_{itr}^{(k)}h\left(Y_{it},X_{it}^{\prime}\beta_{itr}^{\left(k\right)}+\pi_{it}^{(k)}\right)\right)^{2}\\
 & +\frac{1}{2NT}\sum_{i=1}^{N}\sum_{t=1}^{T}\left(\pi_{it}^{(k)}-\sum_{r=1}^{R}w_{itr}^{(k)}h\left(Y_{it},X_{it}^{\prime}\beta_{itr}^{\left(k\right)}+\pi_{it}^{(k)}\right)-\pi_{it}\right)^{2}+\nu\|\Pi\|_{*}
\end{aligned}
$$
where the closed-form solution is given by    
    $$ \Pi^{(k+1)} \gets S_{NT\nu} \left( \Pi^{(k)} - \left(\sum_{r=1}^R w^{(k)}_{i t r} h\left(Y_{it},X_{it}^{\prime}\beta_{itr}+\pi_{it}^{(k)}\right) \right)_{i,t} \right)$$
    with $h(y, v) = \Lambda(v)-y$ and $\Lambda(v) = \frac{1}{1 + e^{-v}}$. This is detailed in Line~\ref{line:SVT2} of Algorithm~\ref{alg:RC}.

\begin{algorithm}[h]
\caption{MM Algorithm: Binary Random Coefficient Logit Panel Model}
\label{alg:RC}
\begin{algorithmic}[1]
\State \textbf{Input}: Observed data $\left(Y_{it},X_{it}\right)$, and initialization for $\beta^{(0)}$, $\Sigma^{(0)}$, $\Pi^{(0)}$, parameters $\nu$.
\State \textbf{Initialize}: $k = 0$
\While{not converged}
    \State \textbf{Step 1: Simulate and Calculate Weights}
    \State Given $\left(\theta^{\left(k\right)},\Pi^{(k)}\right)$, for each \((i, t)\), draw \( R \) values of \( \beta \) from $N\left(\bar{\beta}^{\left(k\right)},\Sigma^{\left(k\right)}\right)$, labeled \( \beta_{itr} \)
    \State For each \( \beta_{itr} \), calculate \( p_{it}(\beta_{itr}, \pi_{it}^{(k)}) \) and the likelihood $L_{it}(\beta_{itr},\pi_{it}^{(k)})$ using Equations \eqref{eq:MM1} and \eqref{eq:MM2}.
    \State Using $L_{it}(\beta_{itr}\pi_{it}^{(k)})$, compute the weights $    w_{itr}^{(k)}$ using Equation \eqref{eq:MM3}.

    \State \textbf{Step 2: Update Parameters}
    \State 
    $\bar{\beta}^{(k+1)} = \frac{1}{NT} \sum_{i=1}^{N} \sum_{t=1}^{T} \left( \sum_{r=1}^{R} w_{itr}^{(k)} \beta_{itr} \right)$
    \State 
$    \Sigma^{(k+1)} = \frac{1}{NT} \sum_{i=1}^{N} \sum_{t=1}^{T} \left( \sum_{r=1}^{R} w_{itr}^{(k)} \beta_{itr} \beta_{itr}^{\prime} \right) - \bar{\beta}^{(k+1)}\bar{\beta}^{(k+1)^{\prime}}$
    \State 
    $    \Pi^{(k+1)} \gets S_{NT\nu} \left( \Pi^{(k)} - \left(\sum_{r=1}^R w^{(k)}_{i t r} h\left(Y_{it},X_{it}^{\prime}\beta_{itr}+\pi_{it}^{(k)}\right) \right)_{i,t} \right)$\label{line:SVT2}
    \State $k \gets k + 1$
\EndWhile

\State \textbf{Output}: Updated parameters $\beta^{(k+1)}$, $\Sigma^{(k+1)}$, and $\Pi^{(k+1)}$.
\end{algorithmic}
\end{algorithm}

\subsection{Selection of Tuning Parameters}\label{subsec:nu}
The selection of the tuning parameter $\nu$ is crucial for the
proper implementation of the estimator. The objective is to choose $\nu$ such that it dominates
the ``score" of our penalized estimator. This ensures
a desirable rate of convergence and permits the consistent estimation of the rank of $\Pi_{0}$, even in scenarios where it is not known a priori. Nonetheless,
 the computation of the ``score" is generally infeasible, as it depends on unknown true parameters $\left(\theta_{0},\Pi_{0}\right)$. This challenge becomes more complex in the context of time series data,
where the selection of $\nu$ must also account for the degree of
temporal dependence.

Several general approaches are commonly used to choose the optimal
tuning parameters. As is standard in the statistics and machine learning
literature, the most popular data-driven method is cross-validation.
Another common approach involves information criteria (IC).
Methods like grid search or randomized search are also employed, where
a range or a sample of tuning parameters is systematically explored
to identify the best value.

In this paper, we select the tuning parameter $\nu$ using a modified information criterion (IC), motivated by prior work on low-rank and grouped panel estimators \citep*[e.g.,][]{bai_determining_2002, su_identifying_2016, belloni_high_2023}. This criterion provides a practical way to balance model complexity and statistical fit. 

Specifically, for each candidate value of $\nu$, we compute the corresponding rank estimate $\hat{r}(\nu)$, as defined in Equation \eqref{ref:r},\footnote{The consistency of the rank estimator is established in Corollary~\ref{cor:main-r}.} and evaluate
\begin{equation}
\operatorname{IC}\left(\nu\right)=\mathcal{L}_{NT}\left(\hat{\theta}\left(\nu\right),\hat{\Pi}\left(\nu\right)\right)+\varrho_{NT}\cdot\hat{r}\left(\nu\right)\label{eqn:BIC}
\end{equation}
The first term on the right-hand side of \eqref{eqn:BIC} measures the overall fit to the data, and the second term introduces a penalty proportional to model complexity through the factor $\varrho_{NT}$.
In the simulations, we set $\varrho_{NT} = \frac{1}{2}\log(N \wedge T)\,\tfrac{N \vee T}{N T}$ for Design 1 and $\varrho_{NT} = \frac{1}{2}\tfrac{\log\log(\sqrt{NT})}{\sqrt{NT}}$ for Design 2. These choices serve as practical finite-sample rules and perform fairly well across the designs considered.

\subsection{Finite-Sample Performance}
\label{subsec:MC}
\paragraph{Data Generating Process}\label{subsec:MC_DGP}
To evaluate the finite-sample performance of the estimation procedure,
we consider two data generating processes (DGPs) that cover models
with convex and nonconvex objective functions. 

\textbf{Design 1:} The data are generated from the following
model
\[
y_{it}=\mathbb{I}\left\{ X_{it}^{\prime}\theta_{0}+\sum_{s=1}^{2}\lambda_{is}f_{ts}-\varepsilon_{it}\geq0\right\} 
\]
where $\mathbb{I}\left\{ \cdot\right\} $ denotes the indicator function, and the error term $\varepsilon_{it}$ follows a logistic
distribution. We set $p=3$ and $r=2$ and take the coefficients $\theta_{0}$
to satisfy $\theta_{0}=\left(1,1,1\right)^{\prime}$. 
The regressor vector
$X_{it}=(x_{it,1},x_{it,2},x_{it,3})^\prime$ consists of three variables, $\lambda_i=(\lambda_{i1},\lambda_{i2})^\prime$, and $f_t=(f_{t1},f_{t2})^\prime$. Specifically, we define $x_{it,k}=\tilde{x}_{it,k}+\rho\left(\lambda{}_{{ik}}^{2}+f_{tk}^{2}\right)$
for $k=1,2$ and $x_{it,3}=\tilde{x}_{it,3}$. The parameter $\rho=0.2$ governs the correlation between the covariates
and the fixed effects. 
Here, $\lambda_i$, $f_t$, and $\tilde{x}_{it}$ are mutually independent across both $i$ and $t$ and $\lambda_{i}\sim \mathcal{N}\left( \binom{1}{1},I_{2}\right)$ and $f_{t}\sim \mathcal{N}\left( \binom{0}{0},I_{2}\right)$. 

We consider two alternative specifications for the distribution of $\tilde{x}_{it}$:
\begin{itemize}
    \item DGP 1: The covariates are independently distributed as $\tilde{x}_{it}\sim \mathcal{N}\left(\begin{pmatrix}0\\
0\\0\end{pmatrix},4I_{3}\right)$ 
where $I_{k}$ denotes $k$ by $k$ identity matrix.
\item DGP 2: The components of $\tilde{x}_{it}$ follow a stationary AR(1) process. For $k=1,2,3$, $\tilde{x}_{i k, t}=\varrho \tilde{x}_{i k, t-1}+\sigma u_{i k, t}$, where $u_{i k, t} \sim \mathcal{N}(0,1)$ with $\varrho=0.2$, $\sigma=2$, and initial values drawn from $\mathcal{N}\sim\bigl(0,\tfrac{\sigma^{2}}{1-\varrho^{2}}\bigr)$.
\end{itemize}

\textbf{Design 2:} We extend Design 1 by allowing
random coefficients $\beta_{it}$. With a slight abuse of notation, 
\[
y_{it}=\mathbb{I}\left\{ X_{it}^{\prime}\beta_{0,it}+\sum_{s=1}^{2}\lambda_{is}f_{ts}-\varepsilon_{it}\geq0\right\} 
\]
The covariates $X_{it}$, factor loadings $\lambda_{is}$, and factors $f_{ts}$ follow the same distributions as in Design 1, and we adopt the same specifications of $\tilde{x}_{it}$. The coefficients $\beta_{0,it}\sim \mathcal{N}\left(\bar{\beta}_{0},\Sigma_{0}\right)$
with $\bar{\beta}_{0}=\left(1,1,1\right)^{\prime}$ and $\Sigma_{0}=0.3 I_{3}$.

\paragraph{Performance}\label{subsec:MC_results}
For each model design, we consider different combinations of sample sizes 
$(N,T) \in \{(100,100),\,(150,100),\,(200,100),\,(150,150),\,(200,150),\,(200,200)\}$ 
to examine how the first-step and second-step estimators behave as the sample size increases.
Each experiment is replicated $100$ times. 
The ADMM algorithm is used to estimate the models in Design 1, while the MM algorithm is applied to estimate the models in Design 2. Detailed descriptions of the implementation are provided in Section~\ref{subsec:MC_imp}.

To evaluate the performance of the proposed estimator, we report the root mean squared error (RMSE) for both steps, defined as $\operatorname{RMSE}(\hat{\theta})=\sqrt{S^{-1}\sum_{s=1}^{S}|\hat{\theta}^{(s)}-\theta_0|^2}/\|\theta_0\|$, and the average estimated rank, $\mathbb{E}[\hat{r}]=S^{-1}\sum_{s=1}^{S}\hat{r}_s$. We set $S=100$ Monte Carlo replications. These measures summarize the finite-sample accuracy of the estimators and the typical rank selected across simulations.

\begin{table}[htbp]
\caption{Simulation Results: $\hat{\theta}$ in Design 1}
\renewcommand{\arraystretch}{1.2} \begin{center}
\begin{tabular}{c c c c c c c}
\cmidrule[1pt](r{.1em}l){1-7}
 & \multicolumn{3}{c}{DGP 1} & \multicolumn{3}{c}{DGP 2} \\ 
\cmidrule[.6pt](r{.1em}l){2-4} \cmidrule[.6pt](r{.1em}l){5-7}
$(N,T)$ 
& \multicolumn{2}{c}{First step} & \multicolumn{1}{c}{Second step} 
& \multicolumn{2}{c}{First step} & \multicolumn{1}{c}{Second step} \\ 
\cmidrule[.6pt](r{.1em}l){2-3} \cmidrule[.6pt](r{.1em}l){4-4}
\cmidrule[.6pt](r{.1em}l){5-6} \cmidrule[.6pt](r{.1em}l){7-7}
 & RMSE & $\mathbb{E}[\hat{r}]$ & RMSE & RMSE & $\mathbb{E}[\hat{r}]$ & RMSE \\ 
\cmidrule[.6pt](r{.1em}l){1-7}
(100,100)  & 6.11 & 2.00 & 3.06 & 5.99 & 2.00 & 2.93 \\
(150,100)  & 4.95 & 2.00 & 2.41 & 5.61 & 2.00 & 2.88 \\
(200,100)  & 5.35 & 2.00 & 2.33 & 5.61 & 2.00 & 2.54 \\
(150,150)  & 4.35 & 2.00 & 2.11 & 4.56 & 2.00 & 2.24 \\
(200,150)  & 2.11 & 2.00 & 1.75 & 2.52 & 2.00 & 1.73 \\
(200,200)  & 1.88 & 2.00 & 1.57 & 2.03 & 2.00 & 1.65 \\
\bottomrule
\end{tabular}
\end{center}
\footnotesize{
Note: This table reports the RMSE for both the first-step and second-step estimators, as well as the average estimated rank from the first-step procedure. Results are presented for $\hat{\theta}$ under Design 1. RMSE values are expressed as percentages relative to the true parameter values (\%), and $\mathbb{E}[\hat{r}]$ denotes the mean estimated number of factors computed across 100 Monte Carlo replications.}
\label{tab:design1_0929}
\end{table}

\begin{table}[htbp]
\caption{Simulation Results: $\hat{\bar{\beta}}$ in Design 2}
\renewcommand{\arraystretch}{1.2} \begin{center}
\begin{tabular}{c c c c c c c}
\cmidrule[1pt](r{.1em}l){1-7}
 & \multicolumn{3}{c}{DGP 1} & \multicolumn{3}{c}{DGP 2} \\ 
\cmidrule[.6pt](r{.1em}l){2-4} \cmidrule[.6pt](r{.1em}l){5-7}
$(N,T)$ 
& \multicolumn{2}{c}{First step} & \multicolumn{1}{c}{Second step} 
& \multicolumn{2}{c}{First step} & \multicolumn{1}{c}{Second step} \\ 
\cmidrule[.6pt](r{.1em}l){2-3} \cmidrule[.6pt](r{.1em}l){4-4}
\cmidrule[.6pt](r{.1em}l){5-6} \cmidrule[.6pt](r{.1em}l){7-7}
 & RMSE & $\mathbb{E}[\hat{r}]$ & RMSE & RMSE & $\mathbb{E}[\hat{r}]$ & RMSE \\ 
\cmidrule[.6pt](r{.1em}l){1-7}
(100,100)  & 8.38 & 2.05 & 7.38 & 8.19 & 2.04 & 7.51 \\
(150,100)  & 7.68 & 2.00 & 6.50 & 7.68 & 2.00 & 6.71 \\
(200,100)  & 6.97 & 2.00 & 5.70 & 6.90 & 2.00 & 5.62 \\ 
(150,150)  & 6.08 & 2.00 & 4.53 & 6.33 & 2.00 & 4.90 \\
(200,150)  & 5.20 & 2.00 & 3.63 & 5.35 & 2.00 & 3.90 \\
(200,200)  & 3.02 & 2.00 & 2.76 & 2.74 & 2.00 & 2.13 \\
\bottomrule
\end{tabular}
\end{center}
\footnotesize{
Note: This table reports the RMSE for both the first-step and second-step estimators, as well as the average estimated rank from the first-step procedure. Results are presented for $\hat{\bar{\beta}}$ under Design 2. RMSE values are expressed as percentages relative to the true parameter values (\%), and $\mathbb{E}[\hat{r}]$ denotes the mean estimated number of factors computed across 100 Monte Carlo replications.}
\label{tab:design2_0929}
\end{table} 
Table~\ref{tab:design1_0929} summarizes the finite-sample performance of the proposed estimators. Several key observations can be made from the results. (i) At each sample size, the second-step estimators exhibit smaller RMSE compared to the first-step estimators; (ii) As the sample size increases, particularly when both $N$ and $T$ grow, the RMSE of the first-step estimators tends to decrease; (iii) Similarly, the RMSE of the second-step estimators generally decreases with increasing sample size; (iv) The rank estimator correctly estimates the rank as the sample size increases, with the average rank stabilizing at $\hat{r} = 2$ for larger samples. These observations confirm the consistency results for both the first-step and the second-step estimators. 
For example, the RMSE decreases from $3.06\%$ when $N=T=100$ to $1.57\%$ at $N=T=200$, consistent with the theoretical prediction that estimator performance improves with larger sample sizes.

Table~\ref{tab:design2_0929} summarizes the results for Design 2, which incorporates random coefficients. The findings parallel those of Design 1: the second-step estimator consistently improves upon the initial estimator, and both bias and RMSE decline with sample size. In this case, we focus on the accuracy of the estimated mean of the random coefficients, which is reliably recovered in larger samples.

In general, the simulations confirm the main theoretical insights. The two-step procedure delivers substantial efficiency gains, the estimators are consistent as both dimensions of the panel grow, and rank selection is reliable in sufficiently large samples. \section{Empirical Application}\label{sec:empirical}
In this section, we apply our methodology to revisit the joint financing behavior of U.S. nonfinancial corporations, focusing on how firms adjust debt issuance and equity repurchase decisions in response to relative valuations across capital markets. The objective is to illustrate the empirical applicability of our approach by extending the baseline model of \citet{ma_nonfinancial_2019} and to compare its benchmark fixed-effects logit with our proposed specification.

Empirical studies suggest that firms often act as cross-market arbitrageurs, substituting between debt and equity to exploit differences in financing costs. When debt is relatively inexpensive, firms tend to issue debt and repurchase equity, whereas undervalued equity prompts them to issue equity and retire debt. In a recent study, \citet{ma_nonfinancial_2019} documents that a substantial share of financing flows originate from nonfinancial firms that simultaneously issue in one market while repurchasing in another. About 45 percent of quarterly net equity repurchases (by value) occur in periods when firms are concurrently net issuers of debt, and roughly 50 percent of seasoned equity issuance takes place among firms that are net retiring debt. These coordinated patterns indicate that firms actively rebalance their capital structures to exploit valuation differentials across markets, with this behavior most pronounced among large, profitable, and financially unconstrained firms.

To quantify this relationship, the study estimates the probability that a firm simultaneously issues equity and retires debt using a panel logit model with firm fixed effects:
\begin{equation} 
\label{emp:eq1}
\operatorname{Pr}\left(y_{i t}=1 \mid X_{D, i t-1}, X_{E, i t-1}, Z_{i t}\right)=\Lambda\left(\lambda_i+\beta_D^{\prime} X_{D, i t-1}+\beta_E^{\prime} X_{E, i t-1}+\gamma^{\prime} Z_{i t}\right)
\end{equation}
where $\Lambda(\cdot)$ denotes the logistic cumulative distribution function.
The binary outcome variable $y_{it}$ equals one if the firm issues equity (negative net equity repurchases) and retires debt (negative net debt issuance) in the same quarter, and zero otherwise.
Net equity repurchases are measured as purchases minus sales of common and preferred stock (PRSTKC$-$SSTK), and net debt issuance is defined as long-term debt issuance minus long-term debt reduction (DLTIS$-$DLTR). 
The vectors $X_{D,it-1}$ and $X_{E,it-1}$ capture debt and equity market valuations measured at the end of the quarter $t-1$, while $Z_{it}$ includes additional firm-level controls that can affect financing behavior.

The model uses three firm-level indicators of market valuation. Two bonds-based measures, the credit spread and the term spread, capture relative valuations in the debt market ($X_{D,it-1}$). The firm-level credit spread is computed as the face-value-weighted average yield differential between the firm’s bonds and the nearest-maturity Treasury, while the term spread is the yield difference between the nearest-maturity Treasury and the three-month Treasury bill.

The valuation of the equity market ($X_{E,it-1}$) is proxied by the firm's specific value-price ratio (V / P), where $V$ denotes the intrinsic value of the equity estimated from the residual income model and $P$ represents the market price of equity \citep{dong_overvalued_2012}. Control variables include net income, cash holdings, capital expenditures, deviations from target leverage, asset growth, and firm size.\footnote{See \citet{ma_nonfinancial_2019} for detailed definitions and data construction.}

Before introducing our model, we first examine whether the data exhibit latent time-varying structure by estimating a two-way fixed-effects (TWFE) specification following \citet*{fernandez-val_individual_2016} and analyzing its residuals,
$$\widehat{u}_{TWFE,i t}=y_{i t}-\Lambda\left(\hat{\lambda}_i +\hat{f}_t+\hat{\beta}_{D}^{\prime} X_{D, i, t-1}+\hat{\beta}_{E}^{\prime} X_{E, i, t-1}+\hat{\gamma}^{\prime} Z_{i t}\right)$$
Figure~\ref{fig:resid_heatmap} shows systematic temporal and cross-sectional patterns: firms display persistent deviations across quarters (horizontal streaks), and many firms experience synchronized movements within the same periods (vertical color bands). These residual patterns point to time-varying unobserved heterogeneity.

\subsection{Data and Model}
We begin with a brief overview of the data and the model. Following \citet{ma_nonfinancial_2019}, this analysis integrates several primary data sources, combining firm-level information on bond prices, equity valuations, and accounting fundamentals with aggregate market indicators. Bond data are obtained primarily from the Trade Reporting and Compliance Engine (TRACE) and supplemented by Datastream and Mergent’s Fixed Income Securities Database (FISD). Equity returns and valuation measures come from CRSP and the Institutional Brokers’ Estimate System (IBES), while balance sheet and cash flow variables are drawn from Compustat.\footnote{Access to TRACE, Compustat, CRSP, and IBES is provided through Wharton Research Data Services (WRDS), while Datastream is accessed via the LESG portal.} The sample covers the period 2003-2024, producing a panel of $N = 212$ firms observed in $T = 88$ quarters.\footnote{Additional details on data sources, variable definitions, and sample construction are provided in Appendix~\ref{apx:emp-data}.}
Summary statistics for key variables appear in Table~\ref{tab:sumstats}.

We extend this framework by incorporating interactive fixed effects and random coefficients, allowing for unobserved common shocks and heterogeneous firm responses to market valuations. Formally, the latent index is given by
\begin{equation}
\label{emp:eq2}
\operatorname{Pr}\left(y_{i t}=1 \mid X_{D, i t-1}, X_{E, i t-1}, Z_{i t}\right)=\Lambda\left(\lambda_i^{\prime} f_t+\beta_{D, i t}^{\prime} X_{D, i t-1}+\beta_{E, i t}^{\prime} X_{E, i t-1}+\gamma^{\prime} Z_{i t}\right)
\end{equation}
where $f_t$ denotes latent factors that capture aggregate market shocks, $\lambda_i$ are the corresponding firm-specific factor loadings, and $\beta_{D,it}$ and $\beta_{E,it}$ represent random coefficients reflecting heterogeneity in sensitivities to debt and equity valuations. Specifically, we assume that
$\begin{pmatrix}
\beta_{D,it} \\
\beta_{E,it}
\end{pmatrix}
\sim
\mathcal{N}\!\left(
\begin{pmatrix}
\bar{\beta}_D \\
\bar{\beta}_E
\end{pmatrix},
\, \Sigma_{\beta}
\right)$.
We employ the same set of market valuation and control variables as in \eqref{emp:eq1}.

\begin{table}[h!]
\centering
\caption{Summary Statistics of Key Variables}
\label{tab:sumstats}
\small
\renewcommand{\arraystretch}{1.2} \begin{tabular}{lcccccc}
\toprule
\textbf{Variable} & \textbf{N} & \textbf{Mean} & \textbf{SD} & \textbf{P$_{10}$} & \textbf{Median} & \textbf{P$_{90}$} \\
\midrule
Credit spread         & 14{,}667 & \phantom{$-$}2.03 & 1.76 & \phantom{$-$}0.73 & \phantom{$-$}1.48 & \phantom{$-$}3.90 \\
Term spread           & 14{,}554 & \phantom{$-$}1.19 & 1.07 & $-0.28$ & \phantom{$-$}1.23 & \phantom{$-$}2.54 \\
V/P ratio             & 14{,}667 & \phantom{$-$}1.08 & 0.63 & \phantom{$-$}0.42 & \phantom{$-$}0.95 & \phantom{$-$}1.88 \\
Cash holdings         & 14{,}557 & \phantom{$-$}7.68 & 7.99 & \phantom{$-$}0.83 & \phantom{$-$}5.15 & \phantom{$-$}17.39 \\
Asset growth          & 14{,}622 & \phantom{$-$}6.34 & 14.45 & $-5.71$ & \phantom{$-$}4.57 & \phantom{$-$}19.97 \\
Size                  & 14{,}667 & \phantom{$-$}9.38 & 1.20 & \phantom{$-$}7.88 & \phantom{$-$}9.29 & \phantom{$-$}10.90 \\
Net income            & 14{,}639 & \phantom{$-$}1.52 & 2.06 & \phantom{$-$}0.05 & \phantom{$-$}1.47 & \phantom{$-$}3.52 \\
Capital expenditures  & 14{,}573 & \phantom{$-$}1.30 & 1.21 & \phantom{$-$}0.28 & \phantom{$-$}0.94 & \phantom{$-$}2.66 \\
\bottomrule
\end{tabular}
\vspace{0.3em}
\end{table}
 \subsection{Analysis}
We report our estimation results in Table~\ref{tab:logit_rc_comparison}. Model~A corresponds to the conventional panel logit specification with individual fixed effects, as defined in~\eqref{emp:eq1}, and serves as the benchmark. It is estimated using the conditional maximum likelihood approach following \citet{ma_nonfinancial_2019}. In contrast, Model~B adopts the specification in~\eqref{emp:eq2}, which extends the benchmark by allowing for random coefficients on lagged valuation measures and by incorporating interactive fixed effects. The rank of the interactive component is estimated using the procedure described in~\eqref{ref:r}; the first-step estimation yields an estimated rank of $\hat{r} = 2$. Accordingly, Table~\ref{tab:logit_rc_comparison} reports the second-step estimates, including the means and standard deviations of the random coefficients.

In general, the results reconfirm the importance of market valuations in shaping joint financing decisions of firms and are consistent with economic intuition: firms are more likely to simultaneously issue equity and retire debt when the cost of debt is high (e.g., elevated bond spreads or expected excess returns) and the cost of equity is low (e.g., a low V/P ratio or low expected excess equity returns). 

We find that the signs of the estimated effects are generally robust to the inclusion of latent factors, with variations in their magnitudes across specifications. For instance, the effect of lagged term spread is larger in our model than in the benchmark specification: a one-standard-deviation increase in term spread multiplies the odds of joint equity issuance and debt retirement by about 1.33 on average across firms, compared with roughly 1.25 in Model~A. We also uncover substantial heterogeneity in valuation sensitivities. In particular, the estimated standard deviation of the random coefficient on the lagged credit spread is 0.363 and statistically significant, indicating meaningful variation in firms' responses to credit market conditions.

We regard the comparison between our results and those from the benchmark model as an illustration that highlights the empirical value of allowing richer forms of unobserved heterogeneity. 
Firms adjust their financing decisions in response to evolving macrofinancial conditions and to unobserved, time-varying shocks—such as shifts in credit availability, monetary policy, or market liquidity—that jointly influence financing behavior and asset valuations. At the same time, firms differ in the extent to which they adjust their financing in response to observed market valuations.  By incorporating interactive fixed effects and random coefficients, the proposed approach flexibly captures both common dynamics and firm heterogeneity, demonstrating its applicability to complex empirical settings.
\begin{table}[ht!]
\begin{center}
\caption{Empirical Illustration: Comparison of Results}
\label{tab:logit_rc_comparison}
\small
\setlength{\tabcolsep}{10pt} \renewcommand{\arraystretch}{1.2} \begin{tabular}{lcccccc}
\toprule
 & \multicolumn{2}{c}{Model A} & \multicolumn{4}{c}{Model B} \\
\cmidrule(lr){2-3} \cmidrule(lr){4-7}
 & \multicolumn{2}{c}{$\beta$} & \multicolumn{2}{c}{Mean ($\bar{\beta}$)} & \multicolumn{2}{c}{SD ($\sigma_\beta$)} \\
\cmidrule(lr){2-3} \cmidrule(lr){4-5} \cmidrule(lr){6-7}
Variable & Coef. & S.E. & Coef. & S.E. & Coef. & S.E. \\
\midrule
L.Credit spread & \phantom{$-$}0.077 & 0.015 & \phantom{$-$}0.124 & 0.014 & \phantom{$-$}0.363 & 0.055 \\
L.Term spread   & \phantom{$-$}0.212 & 0.027 & \phantom{$-$}0.267 & 0.037 & \phantom{$-$}0.335 & 0.159 \\
L.V/P           & $-0.092$ & 0.046 & $-0.118$ & 0.056 & \phantom{$-$}0.316 & 0.341 \\
Net income      & $-0.039$ & 0.013 & $-0.044$ & 0.013 & \multicolumn{2}{c}{--} \\
L.Cash holding  & $-0.006$ & 0.005 & \phantom{$-$}0.023 & 0.003 & \multicolumn{2}{c}{--} \\
CAPX            & $-0.242$ & 0.037 & $-0.034$ & 0.020 & \multicolumn{2}{c}{--} \\
L.Size          & $-0.746$ & 0.057 & $-0.268$ & 0.008 & \multicolumn{2}{c}{--} \\
L.Asset growth  & \phantom{$-$}0.011 & 0.002 & \phantom{$-$}0.031 & 0.002 & \multicolumn{2}{c}{--} \\
\bottomrule
\end{tabular}
\end{center}
\end{table}  \section{Conclusion}
\label{sec:conc}
This paper explores nonlinear panel data models with interactive fixed effects and tackles nonconvex objective functions in large $N$ and $T$ settings. We propose a two-step estimation method using nuclear norm regularization and an iterative process. First, we obtain preliminary estimates of the slope coefficient, factors, and factor loadings using NNR. The second step refines these estimates through an iterative process. We establish the consistency of the preliminary estimator and derive the asymptotic distribution for the second-step estimator. The performance of the estimator is demonstrated through Monte Carlo simulations and an empirical application to a random coefficient binary logit model.

Several extensions remain open for future research. First, it would be valuable to explore whether post-NNR inference can be achieved within a finite number of iterations, extending the work of \citet{chernozhukov_inference_2023} and \citet{choi_inference_2024} to nonlinear panel settings. Second, while this paper accounts for certain heterogeneity in slope coefficients within the random coefficient logit panel model, allowing for even richer forms of coefficient heterogeneity, as discussed by \citet{bonhomme_fixed_2024}, presents an interesting avenue for further study. These topics are left for future exploration. \clearpage

\bibliography{mybib.bib}
\clearpage

\renewcommand{\appendixname}{Appendix} 
\appendixtitleon    \appendixtitletocon 

\begin{appendices}
\counterwithin{assumption}{section}
\setcounter{assumption}{0}
\renewcommand{\theassumption}{\thesection.\arabic{assumption}}
\counterwithin{lemma}{section}
\renewcommand{\thelemma}{\thesection.\arabic{lemma}}
\begin{center}
\textbf{\LARGE{}Appendix}
\par\end{center}

\addcontentsline{toc}{chapter}{Appendices}

\renewcommand{\thefootnote}{A\arabic{footnote}}
\renewcommand{\thetable}{A\arabic{table}}
\renewcommand{\thefigure}{A\arabic{figure}}

\setcounter{section}{0}\setcounter{footnote}{0} 
\setcounter{table}{0}
\setcounter{figure}{0}
\numberwithin{equation}{section}
\renewcommand{\theequation}{\thesection.\arabic{equation}} 

\renewcommand{\thesection}{\Alph{section}} \renewcommand{\thetheorem}{\thesection.\arabic{theorem}} \renewcommand{\thelemma}{\thesection.\arabic{lemma}} \renewcommand{\thecorollary}{\thesection.\arabic{corollary}} \renewcommand{\theproposition}{\thesection.\arabic{proposition}} \renewcommand{\thedefinition}{\thesection.\arabic{definition}} 

\setcounter{theorem}{0} \setcounter{lemma}{0} \setcounter{corollary}{0} \setcounter{proposition}{0} \setcounter{definition}{0} %
 \section{Implementation of the First-Step Estimator}\label{sec:alg}
This section describes the implementation of the first-step estimator. Conventional matrix-completion methods are based on linear formulations and are therefore ineffective for data with nonlinear structures. To address this limitation, we employ two optimization algorithms specifically tailored to our models. For the convex single-index panel model, we adopt a modified Alternating Direction Method of Multipliers (ADMM), which decomposes the overall optimization problem into smaller, more manageable subproblems. This decomposition enables efficient parallelization and ensures convergence in convex settings \citep{Boyd_ADMM_2011}. ADMM is particularly well suited to large-scale problems with separable objective functions.\footnote{Although ADMM is primarily designed for convex problems, it has been successfully applied to some nonconvex settings; see, for example, \citet*{Wang_ADMM_2019}.}

On the other hand, we develop a new method based on the Majorization-Minimization (MM) algorithm for the random coefficient panel model. This algorithm,
which can be seen as a generalization of the Expectation- Maximization (EM) algorithm, has been developed in the statistics literature \citep{Lange_MM_2016,Lange_MM_2004}. We utilize the MM algorithm to separate
random coefficients and interactive fixed effects. That is, we propose a new surrogate function to linearize the logit model so that we can find the closed-form solution to the IFEs. Sequential fixed effects estimation and majorization deliver estimates with significantly reduced computational and memory costs.

\subsection{Single-Index Panel Model}
We develop an Alternating Direction Method of Multipliers (ADMM) algorithm for estimating the first-stage parameters in the general single-index panel model. The binary logit model in Example~\ref{ref:ex3} serves as a leading illustration of this broader framework.

To reformulate the original problem, we also introduce slack variables. As a result, \eqref{ref:Estimation1}
is equivalent to 
\begin{align}
 & \min_{\theta\in\mathbb{R}^{p},\Pi\in\mathbb{R}^{N\times T}}\frac{1}{NT}\sum_{i=1}^{N}\sum_{t=1}^{T}\ell\left(Y_{it},V_{it}\right)+\nu\|\Pi\|_{*}\nonumber \\
 & \qquad\text{s.t. }V=\sum_{j=1}^{p}X_{j}\theta_{j}+Z_{\Pi},\quad Z_{\Pi}-\Pi=0\label{eq:est-alg}
\end{align}
To solve the minimization problem (\ref{eq:est-alg}), the proposed algorithm uses the following augmented Lagrangian 
\begin{align*}
\mathscr{L}\left(\theta,\Pi,V,Z_{\Pi},U_{p},U_{v}\right) & =\frac{1}{NT}\sum_{i=1}^{N}\sum_{t=1}^{T}\ell\left(Y_{it},V_{it}\right)+\nu\|\Pi\|_{*}\\
 & +\frac{\eta}{2NT}\left\Vert V-\sum_{j=1}^{p}X_{j}\theta_{j}-Z_{\Pi}+U_{p}\right\Vert _{F}^{2}+\frac{\eta}{2NT}\left\Vert Z_{\Pi}-\Pi+U_{v}\right\Vert _{F}^{2}
\end{align*}
Here, $U_{p}$ and $U_{v}$ are the scaled dual variables corresponding
to the linear constraints $V=\sum_{j=1}^{p}X_{j}\theta_{j}+Z_{\Pi}$
and $Z_{\Pi}-\Pi=0$, and $\eta>0$ is the penalty parameter for constraint
violations. Due to the separability of the parameters in $\mathscr{L}$,
the ADMM algorithm proceeds by iteratively minimizing the augmented
Lagrangian in blocks with respect to the original variables, in this
case $\left(V,\theta,\Pi\right)$ and $Z_{\Pi}$, and then updating
the dual variables $U_{p}$ and $U_{v}$.

The primal variables are updated by minimizing the augmented Lagrangian in blocks, and the scaled dual variables are then updated using the corresponding primal residuals:\begin{align*}
V^{(k+1)} & \gets \underset{V}{\arg\min}\sum_{i=1}^{N}\sum_{t=1}^{T}\ell\left(Y_{it},V_{it}\right)+\frac{\eta}{2}\left\Vert V-\sum_{j=1}^{p}X_{j}\theta_{j}^{\left(k\right)}-Z_{\Pi}^{\left(k\right)}+U_{p}^{\left(k\right)}\right\Vert _{F}^{2}\\
\theta^{(k+1)} & \gets \underset{\theta}{\arg\min}\frac{\eta}{2}\left\Vert V^{\left(k+1\right)}-\sum_{j=1}^{p}X_{j}\theta_{j}-Z_{\Pi}^{\left(k\right)}+U_{p}^{\left(k\right)}\right\Vert _{F}^{2}\\
\Pi^{(k+1)} & \gets \underset{\Pi}{\arg\min}\frac{\eta}{2}\left\Vert Z_{\Pi}^{\left(k\right)}-\Pi+U_{v}^{\left(k\right)}\right\Vert _{F}^{2}+NT\nu\|\Pi\|_{*}\\
Z_{\Pi}^{\left(k+1\right)} & \gets \underset{Z_{\Pi}}{\arg\min}\frac{\eta}{2}\left\Vert V^{\left(k+1\right)}-\sum_{j=1}^{p}X_{j}\theta_{j}^{\left(k+1\right)}-Z_{\Pi}+U_{p}^{\left(k\right)}\right\Vert _{F}^{2}+\frac{\eta}{2}\left\Vert Z_{\Pi}-\Pi^{\left(k+1\right)}+U_{v}^{\left(k\right)}\right\Vert _{F}^{2}\\
U_{p}^{\left(k+1\right)} & \gets\underset{U_{p}}{\arg\min}\frac{\eta}{2}\left\Vert V^{\left(k+1\right)}-\sum_{j=1}^{p}X_{j}\theta_{j}^{\left(k+1\right)}-Z_{\Pi}^{\left(k+1\right)}+U_{p}\right\Vert _{F}^{2}\\
U_{v}^{(k+1)} & \gets\underset{U_{v}}{\arg\min}\frac{\eta}{2}\left\Vert Z_{\Pi}^{\left(k+1\right)}-\Pi^{\left(k+1\right)}+U_{v}\right\Vert _{F}^{2}
\end{align*}

For single-index models, the loss function $\ell(y,v)$ is differentiable in $v$.
Let $h(y,v)=\nabla_{v}\ell(y,v)$ denote its derivative with respect to the index.
Then the first-step update for $V$ can be written as
\[
V^{(k+1)} \;\gets\;
V^{(k)} - \eta\left(V^{(k)} - X\theta^{(k)} - Z_{\Pi}^{(k)} + U_{p}^{(k)}\right)
     - h\left(Y,V^{(k)}\right),
\]
which corresponds to a gradient-descent or Fisher-scoring step depending on the form of $\ell$.

For instance, in the binary logit model considered in Example~\ref{ref:ex3},
the loss function is $\ell(y,v)=-yv+\log(1+\exp(v))$, whose derivative is
$h(y,v)=\Lambda(v)-y$ with $\Lambda(v)=(1+e^{-v})^{-1}$. Hence,
\[
h(Y_{it},V_{it}^{(k)})=\Lambda(V_{it}^{(k)})-Y_{it}
\]
recovering the update rule used in the binary logit specification of
Simulation~Design~1.

Additionally, following \citet*{cai_singular_2008} and \citet*{Yuan_ALM_2013}, we update $\Pi$ using singular value thresholding (SVT).
Let $A=U_{A}\Sigma_{A}V_{A}^{\prime}$ denote the singular value
decomposition (SVD) of matrix $A\in\mathbb{R}^{N\times T}$, where
$\Sigma_{A}=\operatorname{diag}\left(\left\{ \sigma_{s}\right\} _{s\in\left[N\land T\right]}\right)$.
For $\iota>0$, the soft-thresholding operator $S_{\iota}\left(\cdot\right)$
is defined as 
\[
S_{\iota}\left(A\right)\coloneqq U_{A}S_{\iota}\left(\Sigma_{A}\right)V_{A}^{\prime},\quad S_{\iota}\left(\Sigma_{A}\right)=\operatorname{diag}\left(\left\{ \sigma_{s}-\iota\right\} _{+}\right)
\]
where $t_{+}$ is the positive part of $t$, namely, $t_{+}=\max\left(0,t\right)$.
Further details can be found in Theorem 2.1 of \citet*{cai_singular_2008}.

Lastly, the remaining parameters, $\theta$ and slack variables, are
solved using standard least-squares estimation, which admits closed-form
solutions.

\subsection{Random Coefficient Logit Panel}
To expose the nature of the MM algorithm, we temporarily disregard
the penalty term. Suppose our goal is to minimize $\mathcal{L}_{NT}\left(\Psi\right):\mathbb{R}^{\dim\left(\Psi\right)}\mapsto\mathbb{R}$.
Let $\Psi^{\left(k\right)}$ denote the current iterate in the minimization
process. As suggested by the name, the majorization-minimization algorithm
proceeds in two steps. First, we construct a surrogate function $Q_{NT}\left(\Psi|\Psi^{\left(k\right)}\right)$.
Specifically, to update $\Psi$ at step $k+1$, we use a surrogate
function $Q_{NT}\left(\Psi|\Psi^{\left(k\right)}\right)$ that satisfies
\begin{enumerate}
    \item 
    $Q_{NT}(\Psi^{(k)}|\Psi^{(k)}) = \mathcal{L}_{NT}(\Psi^{(k)})$    
    \item  
$    Q_{NT}(\Psi|\Psi^{(k)}) \geq \mathcal{L}_{NT}(\Psi) \quad \text{for all } \Psi$
\end{enumerate}
Here, the function $Q_{NT}\left(\Psi|\Psi^{\left(k\right)}\right)$
is said to majorize $\mathcal{L}_{NT}\left(\Psi\right)$ at $\Psi^{\left(k\right)}$.
The current iterate $\Psi^{\left(k\right)}$ is treated as constant.
In the second step in the MM algorithm, we update $\Psi$ by choosing
$\Psi^{\left(k+1\right)}$ to minimize $Q_{NT}\left(\Psi|\Psi^{\left(k\right)}\right)$.
However, constructing a surrogate function that both majorizes $\mathcal{L}_{NT}(\Psi)$
at $\Psi^{(k)}$ and is easy to minimize is often a challenge in practice. This
step is analogous to the expectation (E) step in a well-structured
Expectation-Maximization (EM) algorithm.

For the random coefficient logit panel model, we define $\theta=\left(\bar{\beta},\Sigma\right)$
and $\Psi=\left(\theta,\Pi\right)$. The MM algorithm updates the
parameters $\theta=\left(\bar{\beta},\Sigma\right)$ and $\Pi$ separately
by using two surrogate functions $Q_{NT}\left(\theta|\Psi^{\left(k\right)}\right)$
and $Q_{NT}\left(\Pi|\Psi^{\left(k\right)}\right)$. The surrogate
function $Q_{NT}\left(\theta|\Psi^{\left(k\right)}\right)$ follows
from \citet{train_discrete_2009}, and we extend the work of \citet{train_discrete_2009}, \citet{james_mm_2017} and \citet*{chen_mm_2022} by introducing a new surrogate
function $Q_{NT}\left(\Pi|\Psi^{\left(k\right)}\right)$ which effectively
handles the penalty term.

In Example~\ref{ref:ex4}, the probability of $y_{it}=1$ for individual $i$ and
time $t$, given the random coefficient $\beta$ and fixed effect
$\pi_{it}$, is expressed as
\begin{equation}\label{eq:MM1}
p_{it}\left(\beta,\pi_{it}\right)=\frac{\exp\left(X_{it}^{\prime}\beta+\pi_{it}\right)}{1+\exp\left(X_{it}^{\prime}\beta+\pi_{it}\right)}
\end{equation}
The likelihood is given by 
\begin{equation}\label{eq:MM2}
L_{it}(\beta,\pi_{it})=p_{it}(\beta,\pi_{it})^{y_{it}}(1-p_{it}(\beta,\pi_{it}))^{1-y_{it}}
\end{equation}
Thus, the negative log-likelihood function is computed by integrating
the individual-specific and time-specific probabilities over all observed
data:
\[
\mathcal{L}_{NT}\left(\Psi\right)=-\sum_{i=1}^{N}\sum_{t=1}^{T}\log\left[\int L_{it}(\beta,\pi_{it})f\left(\beta\right)d\beta\right]
\]
The integral in the negative log-likelihood does not have a known
closed-form expression and is typically simulated numerically in estimation. 

Rather than trying to directly optimize $\mathcal{L}_{NT}\left(\Psi\right)$,
the strategy of the MM algorithm is to transfer optimization to simpler
surrogate functions that guarantees descent of the negative log-likelihood.
It refines the EM algorithm, as reviewed in \citet{train_discrete_2009}.
In mixed logit models, the EM surrogate function takes the form,

\[
Q_{NT}\left(\Psi|\Psi^{\left(k\right)}\right)=-\sum_{i=1}^{N}\sum_{t=1}^{T}\int\log\left[L_{it}(\beta,\pi_{it})f\left(\beta|\theta\right)g_{it}\left(\beta|\Psi^{\left(k\right)}\right)\right]d\beta
\]
In this expression, $g_{it}\left(\beta_{it}|\Psi^{\left(k\right)}\right)$
is an individual-and-time-specific probability density function of
the unobserved data, $\beta_{it}$, conditional on conditioned on
the observed data for individual at time $t$, and evaluated at $\Psi^{\left(k\right)}$.
This density is derived using Bayes' rule,
\[
g_{it}\left(\beta|\Psi^{\left(k\right)}\right)=\frac{L_{it}(\beta,\pi_{it}^{\left(k\right)})f\left(\beta|\theta^{\left(k\right)}\right)}{\int_{\beta^{\prime}}L_{it}(\beta^{\prime},\pi_{it}^{\left(k\right)})f\left(\beta^{\prime}|\theta^{\left(k\right)}\right)d\beta^{\prime}}
\]
Evaluating this integral requires simulation. For each individual
$i$ at time $t$, we draw $R$ samples from $ N\left(\bar{\beta}^{\left(k\right)},\Sigma^{\left(k\right)}\right)$,
label each as $\beta_{itr}$, and compute the weights using the
conditional likelihood function,
\begin{equation}\label{eq:MM3}
 w_{itr}^{(k)}=\frac{L_{it}(\beta_{itr},\pi_{it}^{(k)})}{\sum_{r^{\prime}=1}^{R}L_{it}(\beta_{itr^{\prime}},\pi_{it}^{(k)})}
\end{equation}
The superscript $k$ on the weights indicates that they are a function
of $\Psi^{\left(k\right)}$. As in maximum simulated likelihood,
$R$ must be large enough to eliminate the bias from simulating
the denominator of the weights \citet{Lee_SML_1995,train_discrete_2009}.

Using the approximation to $g_{it}\left(\beta|\Psi^{\left(k\right)}\right)$,
the EM surrogate function becomes
\begin{align*}
\tilde{Q}_{NT}\left(\Psi|\Psi^{\left(k\right)}\right) & =-\sum_{i=1}^{N}\sum_{t=1}^{T}\sum_{r=1}^{R}w_{itr}^{(k)}\log\left[L_{it}(\beta_{itr},\pi_{it})f\left(\beta_{itr}|\theta\right)\right]\\
 & =\underbrace{-\sum_{i=1}^{N}\sum_{t=1}^{T}\sum_{r=1}^{R}w_{itr}^{(k)}\log\left[L_{it}(\beta_{itr},\pi_{it})\right]}_{\tilde{Q}_{NT}\left(\Pi|\Psi^{\left(k\right)}\right)}\underbrace{-\sum_{i=1}^{N}\sum_{t=1}^{T}\sum_{r=1}^{R}w_{itr}^{(k)}\log\left[f\left(\beta_{itr}|\theta\right)\right]}_{Q_{NT}\left(\theta|\Psi^{\left(k\right)}\right)}
\end{align*}
Since the EM surrogate function is additively separable in the parameters
$\theta$ and $\Pi$ under the log operator, it simplifies the optimization
process by allowing the minimization to be performed over two independent
functions rather than a complex joint maximization. Specifically,
$Q_{NT}\left(\theta|\Psi^{\left(k\right)}\right)$ is the likelihood
function for a multivariate normal distribution with weighted observations
on $\beta$, where the solution for mean and covariance is obtained
through sample analogs (Train, 2009). As a result, $\theta^{\left(k+1\right)}=\left(\bar{\beta}^{\left(k+1\right)},\Sigma^{\left(k+1\right)}\right)$
admits a closed-form solution:
\begin{align*}
\bar{\beta}^{(k+1)} & =\frac{1}{NT}\sum_{i=1}^{N}\sum_{t=1}^{T}\left(\sum_{r=1}^{R}w_{itr}^{(k)}\beta_{itr}\right)\\
 \Sigma^{(k+1)} & =\frac{1}{NT}\sum_{i=1}^{N}\sum_{t=1}^{T}\left(\sum_{r=1}^{R}w_{itr}^{(k)}\beta_{itr}\beta_{itr}^{\prime}\right)-\bar{\beta}^{(k+1)}\bar{\beta}^{(k+1)^{\prime}}
\end{align*}
Note that $\tilde{Q}_{NT}\left(\Psi|\Psi^{\left(k\right)}\right)$
serves as the surrogate function for $\mathcal{L}_{NT}\left(\Psi\right)$
without the penalty term. Instead of optimizing over $\tilde{Q}_{NT}\left(\Pi|\Psi^{\left(k\right)}\right)+NT\nu\|\Pi\|_{*}$\footnote{Note that the original problem is reweighted by $NT$.},
this paper derives an alternative surrogate function, providing a
closed-form solution for $\Pi$. 

The main insight is to exploit a
natural feature of discrete-choice models and recognize that $\tilde{Q}_{NT}\left(\Pi|\Psi^{\left(k\right)}\right)$,
the weighted standard logit log-likelihood can be majorized with a
quadratic upper bound function. We define $Q_{N T}\left(\Pi \mid \Psi^{(k)}\right)$ as the surrogate function for the penalized objective
$\frac{1}{N T} \tilde{Q}_{N T}\left(\Pi \mid \Psi^{(k)}\right)+\nu\|\Pi\|_*$, which includes the scaled log-likelihood and nuclear-norm regularization. Specifically,
$$
\begin{aligned}
Q\left(\Pi\mid\theta^{(k)},\Pi^{(k)}\right)= & \mathcal{L}_{NT}\left(\theta^{(k)},\Pi^{(k)}\right)-\frac{1}{2NT}\sum_{i=1}^{N}\sum_{t=1}^{T}\left(\sum_{r=1}^{R}w_{itr}^{(k)}h\left(Y_{it},X_{it}^{\prime}\beta_{itr}^{\left(k\right)}+\pi_{it}^{(k)}\right)\right)^{2}\\
 & +\frac{1}{2NT}\sum_{i=1}^{N}\sum_{t=1}^{T}\left(\pi_{it}^{(k)}-\sum_{r=1}^{R}w_{itr}^{(k)}h\left(Y_{it},X_{it}^{\prime}\beta_{itr}^{\left(k\right)}+\pi_{it}^{(k)}\right)-\pi_{it}\right)^{2}+\nu\|\Pi\|_{*}
\end{aligned}
$$
where $h(y,v)=\Lambda(v)-y$ and $\Lambda(v)=\frac{1}{1+e^{-v}}$.
Here, the closed-form solution for Q$\left(\Pi|\Psi^{\left(k\right)}\right)$
is derived through SVT,
$$ \Pi^{(k+1)} \gets S_{NT\nu} \left( \Pi^{(k)} - \left(\sum_{r=1}^R w^{(k)}_{i t r} h\left(Y_{it},X_{it}^{\prime}\beta_{itr}+\pi_{it}^{(k)}\right) \right)_{i,t} \right)$$
Here, the threshold $N T \nu$ follows directly from the scaling of the quadratic term. \section{Proofs of Consistency and Convergence Rate Results}
As in \citet{berbee_convergence_1987}, we employ the blocking argument to deal with
$\beta$-mixing data. Similar arguments exist in different settings,
e.g. \citet*{chen_optimal_2015,semenova_inference_2023,belloni_high_2023,yu_rates_1994}.
Let block size be $c_{T}=\lceil2\mu^{-1}\log\left(NT\right)\rceil$
{} with $1\leq c_{T}\leq T/2$. For each $i\in\left[N\right]$, consider
a partition of $\left\{ W_{it}=\left(Y_{it},X_{it}\right)\right\} $
into $2d_{T}$ blocks of size $c_{T}$ each, where $d_{T}=\lfloor T/2c_{T}\rfloor$.
We denote random variables that correspond to block $p$ as $\left\{ W_{it}\right\} _{p}=\left\{ W_{i,\left(p-1\right)c_{T}+1},...W_{i,pc_{T}}\right\} $
for $p\in\left[2d_{T}\right]$. Then we obtain the existence of sequence
of random variables $\{W_{it}^{*}\}_{i\in[N],t\in[T]}$ such that 

\begin{itemize}
\item $\{W_{it}^{*}\}_{i\in[N],t\in[T]}$ is independent of $\{W_{it}\}_{i\in[N],t\in[T]}$; 
\item For a fixed $t$, $\{W_{it}^{*}\}_{i\in[N]}$ are independent; 
\item For a fixed $i$, $\{W_{it}^{*}\}_{p}$ and $\{W_{it}\}_{p}$ are
identically distributed for each block $p$; 
\item For a fixed $i$, the odd-numbered blocks $\{\{W_{it}^{*}\}_{1},\{W_{it}^{*}\}_{3},...,\{W_{it}^{*}\}_{2p-1}\}$
are independent, and the even-numbered blocks $\{\{W_{it}^{*}\}_{2},\{W_{it}^{*}\}_{4},...,\{W_{it}^{*}\}_{2p}\}$
are independent. 
\end{itemize}
Hence, denote $H_{l}$ as the odd-numbered block, $H_{l}^{\prime}$
as the even-number block, and $H_{R}$ as the remainder block of length
$T-2c_{T}d_{T}$. That is, for $l\in\left[d_{T}\right]$ 
\[
\begin{aligned}H_{l}= & \left\{ t:1+2(l-1)c_{T}\leq t\leq(2l-1)c_{T}\right\} \\
H_{l}^{\prime}= & \left\{ t:1+(2l-1)c_{T}\leq t\leq2lc_{T}\right\} \\
H_{R}= & \left\{ t:2c_{T}d_{T}+1\leq t\leq T\right\} 
\end{aligned}
\]
Throughout the appendix, let $\underline{c}$ and $\overline{c}$
denote generic positive constants that may vary across their occurrences.

\subsection{Proof of Results in Section~\ref{subsec:consistency}}
\begin{lemma}
\label{lem:nc}
We have $\left(\theta_{0},\Pi_{0}\right)\in\mathcal{B}$
and if $\mathcal{L}_{NT}\left(\theta_{0},\Pi_{0}\right)=\bar{\mathcal{L}}\left(\theta_{0},\Pi_{0}\right)+o_{p}\left(1\right)$,
then $\left(\hat{\theta},\hat{\Pi}\right)\in\mathcal{B}$ with probability
approaching one.
\end{lemma}
\begin{proof}[Proof of Lemma \ref{lem:nc}]
Recall the definition of $\mathcal{B}$ is
\[
\mathcal{B}=\left\{ \left(\theta,\Pi\right)\in\Theta\times\Phi^{N\times T}:\left\Vert \Pi\right\Vert _{*}\leq\nu^{-1}\left(\bar{\mathcal{L}}\left(\theta_{0},\Pi_{0}\right)+1\right)+\left\Vert \Pi_{0}\right\Vert _{*}\right\}
\]
Since $\bar{\mathcal{L}}\left(\theta_{0},\Pi_{0}\right)$ is nonnegative,
the first statement is clear, i.e., $\left(\theta_{0},\Pi_{0}\right)\in\mathcal{B}$.
For the second statement, since $0\geq\mathcal{L}_{NT}\left(\hat{\theta},\hat{\Pi}\right)-\mathcal{L}_{NT}\left(\theta_{0},\Pi_{0}\right)+\nu\left(\|\hat{\Pi}\|_{*}-\|\Pi_{0}\|_{*}\right)$
and $\mathcal{L}_{NT}$ is nonnegative, we have $\nu\|\hat{\Pi}\|_{*}\leq\bar{\mathcal{L}}\left(\theta_{0},\Pi_{0}\right)+o_{p}\left(1\right)+\nu\|\Pi_{0}\|_{*}$.
Hence, $\left(\hat{\theta},\hat{\Pi}\right)\in\mathcal{B}$ with probability
approaching one.
\end{proof}
\begin{proof}[Proof of Lemma \ref{lem:consistency-2}]
Let $\eta>0$. By Assumption \ref{assu:id-2}, there is $\varepsilon>0$
such that for $\left(\theta,\Pi\right)\in\Theta\times\Phi^{N\times T}$
with $\left\Vert \theta-\theta_{0}\right\Vert ^{2}+\frac{1}{NT}\left\Vert \Pi-\Pi_{0}\right\Vert _{F}^{2}\geq\eta^{2}$,
we have $\bar{\mathcal{L}}\left(\theta,\Pi\right)-\bar{\mathcal{L}}\left(\theta_{0},\Pi_{0}\right)\geq\varepsilon$.
Thus, by the union bound, 
\begin{align*}
 & \mathrm{P}\left(\left\Vert \theta-\theta_{0}\right\Vert ^{2}+\frac{1}{NT}\left\Vert \Pi-\Pi_{0}\right\Vert _{F}^{2}\geq\eta^{2}\right)\\
 & \quad\leq\mathrm{P}\left(\left\{ \bar{\mathcal{L}}\left(\theta,\Pi\right)-\bar{\mathcal{L}}\left(\theta_{0},\Pi_{0}\right)\geq\varepsilon\right\} \cap\left\{ \left(\hat{\theta},\hat{\Pi}\right)\in\mathcal{B}\right\} \right)+\mathrm{P}\left(\left(\hat{\theta},\hat{\Pi}\right)\notin\mathcal{B}\right)
\end{align*}
Since $\left(\theta_{0},\Pi_{0}\right)\in\mathcal{B}$, we have $\mathcal{L}_{NT}\left(\theta_{0},\Pi_{0}\right)=\bar{\mathcal{L}}\left(\theta_{0},\Pi_{0}\right)+o_{p}\left(1\right)$
by Assumption \ref{assu:obj-2}, and $\mathrm{P}\left(\left(\hat{\theta},\hat{\Pi}\right)\notin\mathcal{B}\right)=o\left(1\right)$
by Proposition \ref{lem:nc}.

We have $\bar{\mathcal{L}}\left(\hat{\theta},\hat{\Pi}\right)\geq\bar{\mathcal{L}}\left(\theta_{0},\Pi_{0}\right)$
and $\mathcal{L}_{NT}\left(\hat{\theta},\hat{\Pi}\right)-\mathcal{L}_{NT}\left(\theta_{0},\Pi_{0}\right)\leq\nu\left(\|\Pi_{0}\|_{*}-\|\hat{\Pi}\|_{*}\right)$.
It follows that, on the event $\left\{ \left(\hat{\theta},\hat{\Pi}\right)\in\mathcal{B}\right\} $, 

\begin{align*}
0 & \leq\bar{\mathcal{L}}\left(\hat{\theta},\hat{\Pi}\right)-\bar{\mathcal{L}}\left(\theta_{0},\Pi_{0}\right)\\
 & \leq2\sup_{\left(\theta,\Pi\right)\in\mathcal{B}}\left|\mathcal{L}_{NT}\left(\theta,\Pi\right)-\bar{\mathcal{L}}\left(\theta,\Pi\right)\right|+\nu\left(\|\Pi_{0}\|_{*}-\|\hat{\Pi}\|_{*}\right)\\
 & \leq2\sup_{\left(\theta,\Pi\right)\in\mathcal{B}}\left|\mathcal{L}_{NT}\left(\theta,\Pi\right)-\bar{\mathcal{L}}\left(\theta,\Pi\right)\right|+\nu\|\Pi_{0}\|_{*}
\end{align*}
Therefore,
\begin{align*}
 & \mathrm{P}\left(\left\{ \bar{\mathcal{L}}\left(\theta,\Pi\right)-\bar{\mathcal{L}}\left(\theta_{0},\Pi_{0}\right)\geq\varepsilon\right\} \cap\left\{ \left(\hat{\theta},\hat{\Pi}\right)\in\mathcal{B}\right\} \right)\\
 & \quad\leq\mathrm{P}\left(\left\{ 2\sup_{\left(\theta,\Pi\right)\in\mathcal{B}}\left|\mathcal{L}_{NT}\left(\theta,\Pi\right)-\bar{\mathcal{L}}\left(\theta,\Pi\right)\right|+\nu\|\Pi_{0}\|_{*}\geq\varepsilon\right\} \cap\left\{ \left(\hat{\theta},\hat{\Pi}\right)\in\mathcal{B}\right\} \right)\\
 & \quad\leq\mathrm{P}\left(2\sup_{\left(\theta,\Pi\right)\in\mathcal{B}}\left|\mathcal{L}_{NT}\left(\theta,\Pi\right)-\bar{\mathcal{L}}\left(\theta,\Pi\right)\right|+\nu\|\Pi_{0}\|_{*}\geq\varepsilon\right)
\end{align*}

In spirit of Assumption \ref{assu:obj-2} and that $\nu\left\Vert \Pi_{0}\right\Vert _{*}=o\left(1\right)$,
the above term goes to 0. Thus, the desired result follow.
\end{proof}

 \subsection{Proof of Results in Section~\ref{subsec:rate}}

Recall the definition of norm $\rho\left(\cdot,\cdot\right)$ as $\rho\left(\delta,\Delta\right)=\left[\left\Vert \delta\right\Vert ^{2}+\frac{1}{NT}\left\Vert \Delta\right\Vert _{*}^{2}\right]^{1/2}$.
The empirical risk function over some bounded set $\mathcal{\overline{A}}$
is given by
\[
\epsilon(\gamma):=\underset{\substack{\left(\delta,\Delta\right)\in\mathcal{\mathcal{\overline{A}}}}
}{\sup}\left|\mathcal{\tilde{L}}_{NT}\left(\theta_{0}+\delta,\Pi_{0}+\Delta\right)-\mathcal{\tilde{L}}_{NT}\left(\theta_{0},\Pi_{0}\right)\right|
\]
where 
\[
\mathcal{\overline{A}}=\left\{ \left(\delta,\Delta\right)\in\mathbb{R}^{p}\times\mathbb{R}^{N\times T}:\left(\theta_{0}+\delta,\Pi_{0}+\Delta\right)\in\Theta\times\Phi^{N\times T}\right\} \cap\left\{ \rho^{2}\left(\delta,\Delta\right)\leq\gamma^{2}\right\} 
\]
The following auxiliary lemma is used to establish Proposition~\ref{prop:ep}.
\begin{lemma}
\label{lem:emp}Under Assumptions \ref{assu:data} and \ref{assu:Lip},
$\mathrm{P}\left(\epsilon(\gamma)\geq\kappa\right)=\zeta_{2,NT}$
for any $\kappa>0$, where \textup{$\zeta_{2,NT}$} is defined in
(\ref{eq:zeta2}). Furthermore, with probability at least $1-\zeta_{3,NT}$,
\[
\epsilon(\gamma)\leq\frac{\psi_{NT}\sqrt{c_{T}}}{\sqrt{N\land d_{T}}}\gamma
\]
where \textup{$\zeta_{3,NT}$ is defined in }(\ref{eq:zeta3}) and
is independent of $\gamma$. Moreover, \textup{$\zeta_{2,NT}$ and}
$\zeta_{3,NT}$ approach zero when $\left(N,T\right)\rightarrow\infty$.
\end{lemma}
\begin{proof}[Proof of Lemma \ref{lem:emp}] First, note that under Assumption~\ref{assu:Lip}, $L(W_{it})$ has a uniform fourth-moment bound for all $(i,t)$. Thus, define
$\Omega_1=\left\{\max _{i \in[N], t \in[T]} \mathbb{E}\left[L\left(W_{i t}\right)^4\right] \leq C_L^4\right\}
$, where $C_{L} < \infty$.
We want to bound the empirical risk function $\epsilon(\gamma)$.
To do so, we split the proof into five steps.

\textbf{Step 1:} To begin with, write $\tilde{\theta}=\theta_{0}+\delta$
and $\tilde{\pi}_{it}=\pi_{0,it}+\Delta_{it}$. For any $\kappa>0$,
we have 
\begin{equation}
\begin{aligned}\mathrm{P}(\epsilon(\gamma)\sqrt{NT}\geq\kappa)\leq & 2\mathrm{P}\left(\sup_{\left(\delta,\Delta\right)\in\mathcal{\overline{A}}}\frac{1}{\sqrt{Nd_{T}}}\left|\sum_{i=1}^{N}\sum_{l=1}^{d_{T}}\sum_{t\in H_{l}}\frac{Z_{it}^{*}(\delta,\Delta_{it})}{\sqrt{c_{T}}}\right|\geq\frac{\kappa}{3}\right)\\
+ & \mathrm{P}\left(\sup_{\left(\delta,\Delta\right)\in\mathcal{\overline{A}}}\frac{1}{\sqrt{Nd_{T}}}\left|\sum_{i=1}^{N}\sum_{t\in H_{R}}\frac{Z_{it}^{*}(\delta,\Delta_{it})}{\sqrt{c_{T}}}\right|\geq\frac{\kappa}{3}\right)+2pNd_{T}\beta\left(c_{T}\right)\\
\coloneqq & \;2A_{1}+A_{2}+2pNd_{T}\beta\left(c_{T}\right)
\end{aligned}
\label{lemma:eqn1}
\end{equation}
where we define $Z_{it}^{*}(\delta,\Delta_{it})$ as 
\[
Z_{it}^{*}(\delta,\Delta_{it})=\ell\left(W_{it}^{*};\tilde{\theta},\tilde{\pi}_{it}\right)-\ell\left(W_{it}^{*};\theta_{0},\pi_{0,it}\right)-\mathbb{E}\left[\ell\left(W_{it}^{*};\tilde{\theta},\tilde{\pi}_{it}\right)-\ell\left(W_{it}^{*};\theta_{0},\pi_{0,it}\right)\right]
\]
Here, we use Berbee's coupling lemma (see Lemma 2.1
in \citet{berbee_convergence_1987}. Note that the lemma is also valid for non-stationary
sequences. Next, we proceed to bound $A_{1}$ and $A_{2}$ in (\ref{lemma:eqn1}).

Before bounding $A_1$, note that $\|\Delta\|_F \leq\|\Delta\|_*$. Then, by Assumption~\ref{assu:Lip}, we have the following
bound on the variance of the process,
\begin{align*}
 & \sup _{(\delta, \Delta) \in \overline{\mathcal{A}}} \mathbb{E}\left[\frac{1}{N T} \sum_{i=1}^N \sum_{t=1}^T\left\{\ell\left(W_{i t}^* ; \tilde{\theta}, \tilde{\pi}_{i t}\right)-\ell\left(W_{i t}^*; \theta_0, \pi_{0, i t}\right)\right\}^2\right]
\\
\leq &
\sup _{(\delta, \Delta) \in \overline{\mathcal{A}}} \mathbb{E}\left[\frac{1}{N T} \sum_{i=1} \sum_{t=1}\left\{L\left(W_{i t}\right)\left[\left\|\theta_1-\theta_2\right\|+\left|\pi_1-\pi_2\right|\right]\right\}^2\right]
\leq \tilde{C}_{L}^{2}\gamma^{2}
\end{align*}
where $\tilde{C}_{L}=\sqrt{2}C_{L}$, and $\max_{i\in[N],t\in[T]}\mathbb{E}\left[L(W_{it})^{4}\right]\leq C_{L}^{4}$. 

For any fixed $\delta$ and $\Delta$ with $\left\Vert \delta\right\Vert ^{2}+\frac{1}{NT}\left\Vert \Delta\right\Vert _{*}^{2}\leq\gamma^{2}$,
we have the following bound on the variance of the process,
\[
\begin{aligned}\underset{\left(\delta,\Delta\right)\in\mathcal{\overline{A}}}{\sup}\operatorname{Var}\left(\sum_{i=1}^{N}\sum_{l=1}^{d_{T}}\sum_{t\in H_{l}}Z_{it}^{*}\left(\delta,\Delta_{it}\right)\right)\leq NT\tilde{C}_{L}^{2}\gamma^{2}\end{aligned}
\]

Let $\left\{ \varepsilon_{i,l}\right\} _{i\in[N],l\in\left[d_{T}\right]}$
be i.i.d Rademacher variables independent of the data. Given that
$\mathbb{E}\left[Z_{it}^{*}\left(\delta,\Delta_{it}\right)\right]=0$
by construction, we apply the symmetrization lemma as in Lemma 2.3.7
in \citet{van_der_vaart_weak_1996},
\[
\begin{aligned} & \mathrm{P}\left(\sup_{\left(\delta,\Delta\right)\in\mathcal{\overline{A}}}\frac{1}{\sqrt{Nd_{T}}}\left|\sum_{i=1}^{N}\sum_{l=1}^{d_{T}}\sum_{t\in H_{l}}\frac{Z_{it}^{*}(\delta,\Delta_{it})}{\sqrt{c_{T}}}\right|\geq\kappa\right)\\
 & \quad\leq\frac{\mathrm{P}\left(A^{0}(\gamma)\geq\frac{\kappa}{4}\mid\Omega_{1}\right)+P\left(\Omega_{1}^{c}\right)}{1-\frac{4}{NT\kappa^{2}}\sup_{\left(\delta,\Delta\right)\in\mathcal{\overline{A}}}\operatorname{Var}\left(\sum_{i=1}^{N}\sum_{l=1}^{d_{T}}\sum_{t\in H_{l}}Z_{it}^{*}(\delta,\Delta_{it})\right)}\\
 & \quad\leq\frac{\mathrm{P}\left(A^{0}(\gamma)\geq\frac{\kappa}{4}\mid\Omega_{1}\right)+P\left(\Omega_{1}^{c}\right)}{1-4\tilde{C}_{L}^{2}\gamma^{2}/\kappa^{2}},
\end{aligned}
\]

where
\[
A^{0}(\gamma)\coloneqq\sup_{\left(\delta,\Delta\right)\in\mathcal{\overline{A}}}\left|\frac{1}{\sqrt{Nd_{T}}}\sum_{i=1}^{N}\sum_{l=1}^{d_{T}}\varepsilon_{i,l}\left(\sum_{t\in H_{l}}\frac{\ell\left(W_{it}^{*};\tilde{\theta},\tilde{\pi}_{it}\right)-\ell\left(W_{it}^{*};\theta_{0},\pi_{0,it}\right)}{\sqrt{c_{T}}}\right)\right|
\]
Next, we define $B_{1}^{0}(\gamma)$ and $B_{2}^{0}(\gamma)$ as:
\[
B_{1}^{0}(\gamma)=\sup_{\left(\delta,\Delta\right)\in\mathcal{\overline{A}}}\left|\frac{1}{\sqrt{Nd_{T}}}\sum_{i=1}^{N}\sum_{l=1}^{d_{T}}\varepsilon_{i,l}\left(\sum_{t\in H_{l}}\frac{L\left(W_{it}^{*}\right)\left\Vert \delta \right\Vert}{\sqrt{c_{T}}}\right)\right|
\]
and 
\[
B_{2}^{0}(\gamma)=\sup_{\left(\delta,\Delta\right)\in\mathcal{\overline{A}}}\left|\frac{1}{\sqrt{Nd_{T}}}\sum_{i=1}^{N}\sum_{l=1}^{d_{T}}\varepsilon_{i,l}\left(\sum_{t\in H_{l}}\frac{L\left(W_{it}^{*}\right)\Delta_{it}}{\sqrt{c_{T}}}\right)\right|
\]
By Assumption~\ref{assu:Lip} and union bound we obtain that 
\begin{equation}
\mathrm{P}\left(A^{0}(\gamma)\geq\kappa\mid\Omega_{1}\right)\leq\mathrm{P}\left(B_{1}^{0}(\gamma)\geq\kappa\mid\Omega_{1}\right)+\mathrm{P}\left(B_{2}^{0}(\gamma)\geq\kappa\mid\Omega_{1}\right)\label{lemma:eqn3}
\end{equation}
It remains to bound two terms in the RHS of (\ref{lemma:eqn3}).

\textbf{Step 2:} We first consider the bound on $B_{1}^{0}(\gamma)$.
The argument is similar to the proof of Lemma D.3 in \citet{belloni_high-dimensional_2018}.
Note that 
\begin{align*}
B_{1}^{0}(\gamma) & \leq c_{T}\max_{m\in\left[c_{T}\right]}\sup_{\left(\delta,\Delta\right)\in\mathcal{\overline{A}}}\left|\frac{1}{\sqrt{Nd_{T}}}\sum_{i=1}^{N}\sum_{l=0}^{d_{T}-1}\varepsilon_{i,l}\frac{L\left(W_{i,2lc_{T}+m}^{*}\right)\left\Vert \delta \right\Vert}{\sqrt{c_{T}}}\right|\\
 & =\max_{m\in\left[c_{T}\right]}\sup_{\left(\delta,\Delta\right)\in\mathcal{\overline{A}}}\left|\frac{\sqrt{c_{T}}}{\sqrt{Nd_{T}}}\sum_{i=1}^{N}\sum_{l=0}^{d_{T}-1}\varepsilon_{i,l}L\left(W_{i,2lc_{T}+m}^{*}\right)\left\Vert \delta \right\Vert\right|
\end{align*}
Applying Markov's inequality and the union bound, we
get
\begin{align*}
 & \mathrm{P}\left( B_{1}^{0}(\gamma)\geq\kappa\right)\\
 & \quad\leq c_{T}\max_{m\in\left[c_{T}\right]}\mathrm{P}\left( \sup_{\left(\delta,\Delta\right)\in\mathcal{\overline{A}}}\left|\frac{\sqrt{c_{T}}}{\sqrt{Nd_{T}}}\sum_{i=1}^{N}\sum_{l=0}^{d_{T}-1}\varepsilon_{i,l}L\left(W_{i,2lc_{T}+m}^{*}\right)\left\Vert \delta \right\Vert\right|\geq\kappa\right)\\
 & \quad\leq c_{T}\max_{m\in\left[c_{T}\right]}\inf_{\tau>0}\exp\left(-\tau\kappa\right)\mathbb{E}\left[\exp\left(\tau\sup_{\left(\delta,\Delta\right)\in\mathcal{\overline{A}}}\left|\frac{\sqrt{c_{T}}}{\sqrt{Nd_{T}}}\sum_{i=1}^{N}\sum_{l=0}^{d_{T}-1}\varepsilon_{i,l}L\left(W_{i,2lc_{T}+m}^{*}\right)\left\Vert \delta \right\Vert\right|\right)\right]
\end{align*}
Now, fix $m\in\left[c_{T}\right]$. In this step, we will condition
throughout on $\left\{ W_{it}^{*}\right\} $ but omit explicit notation
for this conditioning to keep the notation lighter. We have
\begin{align*}
 & \mathbb{E}\left[\exp\left(\tau\sup_{\left(\delta,\Delta\right)\in\mathcal{\overline{A}}}\left|\frac{\sqrt{c_{T}}}{\sqrt{Nd_{T}}}\sum_{i=1}^{N}\sum_{l=0}^{d_{T}-1}\varepsilon_{i,l}L\left(W_{i,2lc_{T}+m}^{*}\right)\left\Vert \delta \right\Vert\right|\right)\right]\\
 & \quad\leq2\mathbb{E}\left[\exp\left(\frac{\tau\gamma\sqrt{c_{T}}}{\sqrt{Nd_{T}}}\sum_{i=1}^{N}\sum_{l=0}^{d_{T}-1}\varepsilon_{i,l}L\left(W_{i,2lc_{T}+m}^{*}\right)\right)\right]\\
 & \quad\leq2\exp\left(\tau^{2}c_{T}\gamma^{2}C_{L}^{2}\right)
\end{align*}
The first inequality follows from
the definition of $\mathcal{\overline{A}}$
and the last inequality follows from the law of iterated expectations
combined with both Hoeffding inequality and Chebyshev's inequality.
Hence, 
\begin{align*}
 & \mathrm{P}\left( B_{1}^{0}(\gamma)\geq\kappa\right)\\
 & \quad\leq2c_{T}\inf_{\tau>0}\exp\left(-\tau\kappa\right)\exp\left(\tau^{2}c_{T}\gamma^{2}C_{L}^{2}\right)\\
 & \quad\leq2 c_{T}\exp\left(-\frac{\kappa^{2}}{4c_{T}\gamma^{2}C_{L}^{2}}\right)
\end{align*}

\textbf{Step 3:} Here, we bound $B_{2}^{0}(\gamma)$. First, by the
definition of $\mathcal{\overline{A}}$, we have $\sup_{\left(\delta,\Delta\right)\in\mathcal{\overline{A}}}\|\Delta\|_{*}\leq\sqrt{NT}\gamma$
and we define $\zeta_{1,NT}=\sqrt{NT}$.

Notice that
\[
\begin{aligned}B_{2}^{0}(\gamma)\leq & c_{T}\max_{m\in\left[c_{T}\right]}\sup_{\left(\delta,\Delta\right)\in\mathcal{\overline{A}}}\left|\frac{1}{\sqrt{Nd_{T}}}\sum_{i=1}^{N}\sum_{l=0}^{d_{T}-1}\varepsilon_{i,l}\frac{L\left(W_{i,2lc_{T}+m}^{*}\right)\Delta_{i,2lc_{T}+m}}{\sqrt{c_{T}}}\right|\\
= & \max_{m\in\left[c_{T}\right]}\sup_{\left(\delta,\Delta\right)\in\mathcal{\overline{A}}}\left|\frac{\sqrt{c_{T}}}{\sqrt{Nd_{T}}}\sum_{i=1}^{N}\sum_{l=0}^{d_{T}-1}\varepsilon_{i,l}L\left(W_{i,2lc_{T}+m}^{*}\right)\Delta_{i,2lc_{T}+m}\right|
\end{aligned}
\]
We proceed to bound the moment generating function of $B_{2}^{0}(\gamma)$
and use that to obtain an upper bound on $B_{2}^{0}(\gamma)$. Now
fix $m\in\left[c_{T}\right]$ ,
\[
\begin{aligned} & \mathbb{E}\left[\exp\left(\tau\sup_{\left(\delta,\Delta\right)\in\mathcal{\overline{A}}}\left|\frac{\sqrt{c_{T}}}{\sqrt{Nd_{T}}}\sum_{i=1}^{N}\sum_{l=0}^{d_{T}-1}\varepsilon_{i,l}L\left(W_{i,2lc_{T}+m}^{*}\right)\Delta_{i,2lc_{T}+m}\right|\right)\right]\\
 & \quad\leq\mathbb{E}\left[\sup_{\left(\delta,\Delta\right)\in\mathcal{\overline{A}}}\exp\left(\tau\frac{\sqrt{c_{T}}}{\sqrt{Nd_{T}}}\sum_{i=1}^{N}\sum_{l=0}^{d_{T}-1}C_{L} \left\Vert \left(\varepsilon_{il}L\left(W_{i,2lc_{T}+m}^{*}\right)\right)_{i,l}\right\Vert  \|\Delta\|_{*}\right)\right]\\
 & \mathbb{\quad\leq E}\left[\sup_{\left(\delta,\Delta\right)\in\mathcal{\overline{A}}}\exp\left(\tau\frac{\sqrt{c_{T}}}{\sqrt{Nd_{T}}}C_{L}\mathbb{E}\left\Vert \left(\varepsilon_{il}L\left(W_{i,2lc_{T}+m}^{*}\right)\right)_{i,l}\right\Vert  \|\Delta\|_{*}\right)\right.\\
 & \quad\cdot\left.\sup_{\left(\delta,\Delta\right)\in\mathcal{\overline{A}}}\exp\left(\tau\frac{\sqrt{c_{T}}}{\sqrt{Nd_{T}}}C_{L}\left(\left\Vert \left(\varepsilon_{il}L\left(W_{i,2lc_{T}+m}^{*}\right)\right)_{i,l}\right\Vert -\mathbb{E}\left\Vert \left(\varepsilon_{il}L\left(W_{i,2lc_{T}+m}^{*}\right)\right)_{i,l}\right\Vert \right)\|\Delta\|_{*}\right)\right]\\
 & \quad\leq\exp\left(\overline{c}\tau\frac{\sqrt{c_{T}}}{\sqrt{Nd_{T}}}C_{L}\zeta_{1,NT}\left(\sqrt{N}\lor\sqrt{d_{T}}\right)\gamma\right)\exp\left(\overline{c}\tau^{2}\frac{c_{T}}{Nd_{T}}C_{L}^{2}\zeta_{1,NT}^{2}\gamma^{2}\right)
\end{aligned}
\]
for a positive constant $\overline{c}>0$. The first inequality holds
due to the duality between the spectral and nuclear norms, while the
second inequality follows from the triangle inequality. The two upper
bounds in the final inequality are derived from Theorem 3.4 in \citet{chatterjee_matrix_2015}
and Theorem 1.2 in \citet{vu_spectral_2007}, respectively.

Therefore, by Markov inequality, we have \[
\begin{aligned} & \mathrm{P}\left(B_{2}^{0}(\gamma)\geq\kappa\mid\Omega_{1}\right)\\
 & \quad\leq c_{T} \max_{m\in\left[c_{T}\right]}\mathrm{P}\left( \sup_{\left(\delta,\Delta\right)\in\mathcal{\overline{A}}}\left|\frac{\sqrt{c_{T}}}{\sqrt{Nd_{T}}}\sum_{i=1}^{N}\sum_{l=0}^{d_{T}-1}\varepsilon_{i,l}L\left(W_{i,2lc_{T}+m}^{*}\right)\Delta_{di,2lc_{T}+m}\right|\geq\kappa\right)\\
 & \quad\leq c_{T}\inf_{\tau>0}\exp\left(-\tau\kappa\right)\exp\left(\overline{c}\tau\frac{\sqrt{c_{T}}}{\sqrt{Nd_{T}}}C_{L}\zeta_{1,NT}\gamma\left(\sqrt{N}\lor\sqrt{d_{T}}\right)\right)\exp\left(\overline{c}\tau^{2}\frac{c_{T}}{Nd_{T}}C_{L}^{2}\zeta_{1,NT}^{2}\gamma^{2}\right)\\
 & \quad\leq\overline{c}c_{T}\exp\left(-\frac{\sqrt{N\land d_{T}}\kappa}{\sqrt{c_{T}}C_{L}\zeta_{1,NT}\gamma}\right)
\end{aligned}
\]
for a positive constant $\overline{c}>0$. Similarly, we repeat the
arguments above to bound $A_{2}$ in (\ref{lemma:eqn1}).

\textbf{Step 4: }To prove the first part of Lemma \ref{lem:emp},
note that for any positive $\gamma$ and $\kappa>0$, we have
\begin{align}
 & \mathrm{P}(\epsilon(\gamma)\geq\kappa)\nonumber \\
 &  \quad\leq3\frac{\overline{c}\exp\left(-\frac{NT\kappa^{2}}{c_{T}C_{L}^{2}\gamma^{2}}+\ln\left(pc_{T}\right)\right)+\overline{c}\exp\left(-\frac{\sqrt{N\land d_{T}}\kappa}{\sqrt{c_{T}}C_{L}\gamma}+\log\left(c_{T}\right)\right)}{1-\overline{c}\tilde{C}_{L}^{2}\gamma^{2}/\left(NT\kappa^{2}\right)}\nonumber \\
 & \quad+2pNd_{T}\beta\left(c_{T}\right)\coloneqq\zeta_{2,NT}\label{eq:zeta2}
\end{align}

\textbf{Step 5:} To prove the second part of Lemma \ref{lem:emp},
we set 
\[
\kappa=\frac{\psi_{NT}\sqrt{c_{T}}\zeta_{1,NT}}{\sqrt{N\land d_{T}}}\gamma
\]
From this, we obtain that

\begin{align}
 & \mathrm{P}(\epsilon(\gamma)\sqrt{NT}\geq\kappa)\nonumber \\
 &  \quad\leq3\frac{\overline{c}\exp\left(-\psi_{NT}^{2}c_{T}\left(N\lor d_{T}\right)+\log\left(pc_{T}\right)\right)+\overline{c}\exp\left(-\psi_{NT}+\log\left(c_{T}\right)\right)}{1-\overline{c}\left[\psi_{NT}^{2}c_{T}\left(N\lor T\right)\right]^{-1}}\nonumber \\
 & \quad+2pNd_{T}\beta\left(c_{T}\right)\coloneqq\zeta_{3,NT}\label{eq:zeta3}
\end{align}
For fixed $\kappa>0$, as $\left(N,T\right)\rightarrow\infty$, \eqref{eq:zeta3}
converges to zero by the definition of $\psi_{NT}$. Therefore, \[
\epsilon(\gamma)\leq\kappa/\sqrt{NT}=\frac{\psi_{NT}\sqrt{c_{T}}}{\sqrt{N\land d_{T}}}\gamma
\]
with probability at least $1-\zeta_{3,NT}$, where $\zeta_{3,NT}$
is independent of $\gamma$ and $\zeta_{3,NT}\rightarrow0$ as $\left(N,T\right)\rightarrow\infty$. 
\end{proof}

Before we proceed to prove Proposition~\ref{prop:ep}, we first bound
the empirical risk over some set $\mathcal{V}$. Define $\left|\mathcal{V}\right|$
as
\[
\left|\mathcal{V}\right|=\sup_{\left(\theta,\Pi\right)\in\mathcal{V}}\rho\left(\theta-\theta_{0},\Pi-\Pi_{0}\right)
\]

\begin{proof}[Proof of Proposition \ref{prop:ep}]
Define $\mathcal{C}=\left\{ \left(\theta,\Pi\right)\in\mathcal{V}:\rho\left(\theta-\theta_{0},\Pi-\Pi_{0}\right)\leq c_{\varepsilon,NT}\right\} $
and the collection of rings $\mathcal{C}_{k}=\left\{ \left(\theta,\Pi\right)\in\mathcal{V}:2^{k}\leq\rho\left(\theta-\theta_{0},\Pi-\Pi_{0}\right)\leq2^{k+1}\right\} ,k\in\mathbb{Z}$.
We denote

\[
 \begin{aligned} & A=\underset{\substack{\left(\theta,\Pi\right)\in\mathcal{C}}
}{\sup}\left|\mathcal{\tilde{L}}_{NT}(\theta,\Pi)-\mathcal{\tilde{L}}_{NT}(\theta_{0},\Pi_{0})\right|\\
 & A_{k}=\underset{\substack{\left(\theta,\Pi\right)\in\mathcal{C}_{k}}
}{\sup}\left|\mathcal{\tilde{L}}_{NT}(\theta,\Pi)-\mathcal{\tilde{L}}_{NT}(\theta_{0},\Pi_{0})\right|
\end{aligned}
\]
Define $u_{n}=\left\lfloor \frac{\log\left(c_{\varepsilon,NT}\right)}{\log(2)}\right\rfloor $
and $v_{n}=\left\lceil \frac{\log(\left|\mathcal{V}\right|)}{\log(2)}\right\rceil $,
assuming $c_{\varepsilon,NT}\leq\left|\mathcal{V}\right|$
(otherwise, only $A$ needs to be bounded). Since $\left\{ \rho\left(\delta,\Delta\right)\leq c_{\varepsilon,NT}\right\} \cup\left\{ \cup_{k=u_{n}}^{v_{n}}A_{k}\right\} $
covers the set $\mathcal{V}$ and because $\rho\left(\delta,\Delta\right)\geq2^{k}$
on $\mathcal{C}_{k}$, it holds that

\[
\frac{\left|\mathcal{\tilde{L}}_{NT}(\theta,\Pi)-\mathcal{\tilde{L}}_{NT}(\theta_{0},\Pi_{0})\right|}{\rho\left(\theta-\theta_{0},\Pi-\Pi_{0}\right)\vee c_{\varepsilon,NT}}\leq\left(A\left(c_{\varepsilon,NT}\right)^{-1}\right)\vee\left(\max_{u_{n}\leq k\leq v_{n}}A_{k}2^{-k}\right)
\]
We now examine individually each term on the right-hand side. By using
the fact that $\left\Vert \Delta\right\Vert \leq\left\Vert \Delta\right\Vert _{*}$
and setting $\zeta_{1,NT}=\sqrt{NT}$, we invoke Lemma~\ref{lem:emp}.
Hence, we have 
\[
\mathrm{P}\left(A\left(c_{\varepsilon,NT}\right)^{-1}\geq \psi_{NT}c_{\varepsilon,NT}\right)\leq\zeta_{3,NT}
\]
Similarly, it holds that $\mathrm{P}\left(A_{k}\geq \psi_{NT}c_{\varepsilon,NT}2^{k+1}\right)\leq\zeta_{3,NT}$.
The union bound then yields,
\[
\mathrm{P}\left(\max_{u_{n}\leq k\leq v_{n}}A_{k}2^{-k}\geq2\psi_{NT}c_{\varepsilon,NT}\right)\leq\left(v_{n}-u_{n}+1\right)\zeta_{3,NT}
\]
The union bound and the two inequalities together give that
\begin{align*}
 & \mathrm{P}\left(\sup_{\left(\theta,\Pi\right)\in\mathcal{V}}\frac{\left|\mathcal{\tilde{L}}_{NT}\left(\theta,\Pi\right)-\mathcal{\tilde{L}}_{NT}\left(\theta_{0},\Pi_{0}\right)\right|}{\rho\left(\theta-\theta_{0},\Pi-\Pi_{0}\right)\vee c_{\varepsilon,NT}}\geq \psi_{NT}c_{\varepsilon,NT}\right)\\
 & \quad\leq\left(v_{n}-u_{n}+2\right)\zeta_{2,NT}\\
 & \quad\leq\left(4+2\log_{2}\left(\sqrt{N\land T}\right)\right)\zeta_{2,NT}\\
 & \quad\leq3\frac{\overline{c}\exp\left(-\psi_{nt}^{2}\left(N\lor d_{T}\right)+\bar{c}\log\left(pc_{T}\right)\right)+\overline{c}\exp\left(-\psi_{nt}+\bar{c}\log\left(c_{T}\right)\right)}{1-\overline{c}\left[\psi_{nt}^{2}c_{T}\left(N\lor T\right)\right]^{-1}}\\
 & \quad+2pNT\beta\left(c_{T}\right)\\
 & \quad\coloneqq\zeta_{4,NT}
\end{align*}
for some constant $\overline{c}>0$. To derive the penultimate inequality,
we use the fact that $\left\lceil x\right\rceil \leq x+1$ and that
$\left|\mathcal{V}\right|\leq\left(\sqrt{p}+\sqrt{N\land T}\right)c_{l}^{\prime}$.
Furthermore, without loss of generality, we assume that the fixed
number of coefficients $p$ is less than $N$ and $T$. Hence, the last inequality follows. As a result, the desired outcome
follows, with $\zeta_{4,NT}$ tending to zero as $\left(N,T\right)\rightarrow\infty$.
\end{proof} \begin{lemma}
\label{lemma:matrix} 
Let $\mathcal{P}(\cdot)$ and $\mathcal{M}(\cdot)$ be the projection operators defined in \eqref{eq:PM-def}. 
For any $N\times T$ matrices $\Pi_{0}$ and $\Delta$, we have the following results\\
(i) $\left\Vert \Pi_{0}+\mathcal{M}(\Delta)\right\Vert _{*}=\left\Vert \Pi_{0}\right\Vert _{*}+\|\mathcal{M}(\Delta)\|_{*}$; \\
(ii) $\|\Delta\|_{F}^{2}=\|\mathcal{M}(\Delta)\|_{F}^{2}+\|\mathcal{P}(\Delta)\|_{F}^{2}$; \\
(iii) $\operatorname{rank}(\mathcal{P}(\Delta))\leq2 \operatorname{rank}(\Pi_{0})$; \\
(iv) $\|\Delta\|_{*}^{2}\leq\|\Delta\|_{F}^{2}\operatorname{rank}(\Delta)$; \\
(v) For any conformable matrices $\Delta_{1}$ and $\Delta_{2}$,
$\left|\operatorname{tr}\left(\Delta_{1}\Delta_{2}\right)\right|\leq\left\Vert \Delta_{1}\right\Vert \left\Vert \Delta_{2}\right\Vert _{*}$. 
\end{lemma}
\begin{proof}[Proof of Lemma \ref{lemma:matrix}]
The proof follows from that of Lemma C.2 in \citet{chernozhukov_inference_2019-1}.
\end{proof}
\begin{proof}[Proof of Theorem \ref{thm:main-3}]
Let
$$
\begin{aligned} & \mathcal{A}_1=\left\{\|\delta\|^2+\frac{1}{N T}\|\Delta\|_F^2>\gamma_{NT}^2 \vee c_{\varepsilon, n t}^{\prime}\right\} \\ & \mathcal{A}_2=\{(\hat{\theta}, \hat{\Pi}) \in \mathcal{V}\} \\ & \mathcal{A}_3=\left\{\left|\widetilde{\mathcal{L}}_{N T}(\hat{\theta}, \hat{\Pi})-\widetilde{\mathcal{L}}_{N T}\left(\theta_0, \Pi_0\right)\right| \leq \psi_{NT}c_{\varepsilon,NT} \rho\left(\hat{\theta}-\theta_0, \hat{\Pi}-\Pi_0\right)\right\}
\end{aligned}$$
with $\gamma_{NT}=\frac{2\left(c_{\nu}+1\right)\sqrt{r}}{c_{RSC}}\psi_{NT}c_{\varepsilon,NT}$. 
We have 
\[
 \begin{aligned}\mathrm{P}\left(\mathcal{A}_{1}\right) & =P(\mathcal{A}_{1}\cap\mathcal{A}_{2}\cap\mathcal{A}_{3})+P\left(\mathcal{A}_{1}\cap(\mathcal{A}_{2}\cap\mathcal{A}_{3})^{c}\right)\\
 & =P(\mathcal{A}_{1}\cap\mathcal{A}_{2}\cap\mathcal{A}_{3})+P\left(\mathcal{A}_{1}\cap\left(\mathcal{\mathcal{A}}_{2}^{c}\cup\mathcal{A}_{3}^{c}\right)\right)\\
 & \leq P(\mathcal{A}_{1}\cap\mathcal{A}_{2}\cap\mathcal{A}_{3})+P\left(\mathcal{A}_{1}\cap\mathcal{A}_{2}^{c}\right)+P\left(\mathcal{A}_{1}\cap\mathcal{A}_{3}^{c}\right)\\
 & :=P_{1}+P_{2}+P_{3}
\end{aligned}
\]
where the second inequality is due to the union bound. It holds $P_{2}\rightarrow0$
by Lemma \ref{lem:consistency-2}. Concerning $P_{3}$, when conditional
on $\mathcal{A}_{1}\cap\mathcal{A}_{3}^{c}$, it holds that $\rho\left(\hat{\theta}-\theta_{0},\hat{\Pi}-\Pi_{0}\right)\vee c_{\varepsilon,NT}=\rho\left(\hat{\theta}-\theta_{0},\hat{\Pi}-\Pi_{0}\right)$.
Therefore, over $\mathcal{A}_{1}\cap\mathcal{A}_{3}^{c}$, we find
that
\[
\left|\mathcal{\tilde{L}}_{NT}\left(\theta,\Pi\right)-\mathcal{\tilde{L}}_{NT}\left(\theta_{0},\Pi_{0}\right)\right|\geq \psi_{NT}c_{\varepsilon,NT}\left(\rho\left(\theta-\theta_{0},\Pi-\Pi_{0}\right)\vee c_{\varepsilon,NT}\right)
\]
which holds with probability approaching zero by Lemma \ref{lem:emp}.
Hence, it suffices to show that $P_{1}=0$. Note that 
\begin{align*}
0 & \geq\mathcal{L}_{NT}\left(\hat{\theta},\hat{\Pi}\right)-\mathcal{L}_{NT}\left(\theta_{0},\Pi_{0}\right)+\nu\left(\|\hat{\Pi}\|_{*}-\|\Pi_{0}\|_{*}\right)\\
 & \geq\mathcal{E}(\hat{\theta},\hat{\Pi})-\mathcal{\tilde{L}}_{NT}\left(\theta_{0},\Pi_{0}\right)+ \mathcal{\tilde{L}}_{NT}\left(\hat{\theta},\hat{\Pi}\right)+\nu\left(\|\mathcal{M}\left(\Delta\right)\|_{*}-\|\mathcal{P}\left(\Delta\right)\|_{*}\right)
\end{align*}
where the first inequality holds by the definition of $\mathcal{L}_{NT}\left(\hat{\theta},\hat{\Pi}\right)$
and the second inequality holds by Lemma~\ref{lemma:matrix}. Specifically,
\begin{equation}\label{ref:aux-nu}
\begin{aligned}\|\hat{\Pi}\|_{*}-\|\Pi_{0}\|_{*}= & \|\Pi_{0}+\mathcal{M}\left(\Delta\right)+\mathcal{P}\left(\Delta\right)\|_{*}-\|\Pi_{0}\|_{*}\\
\geq & \|\Pi_{0}+\mathcal{M}\left(\Delta\right)\|_{*}-\|\mathcal{P}\left(\Delta\right)\|_{*}-\|\Pi_{0}\|_{*}\\ 
= & \|\mathcal{M}\left(\Delta\right)\|_{*}-\|\mathcal{P}\left(\Delta\right)\|_{*}
\end{aligned} 
\end{equation}
Next, the definition of $\mathcal{A}_3$ implies that the empirical
risk function is bounded above. Since $\mathcal{E}(\hat{\theta},\hat{\Pi})\geq0$,
and $\nu=c_{\nu}\frac{\psi_{NT}c_{\varepsilon,NT}}{\sqrt{NT}}$, we have 
\begin{align*}
0 & \leq \psi_{NT}c_{\varepsilon,NT}\left(\left\Vert \delta\right\Vert +\frac{1}{\sqrt{NT}}\left\Vert \Delta\right\Vert _{*}\right)+c_{\nu}\frac{\psi_{NT}c_{\varepsilon,NT}}{\sqrt{NT}}\left(\|\mathcal{P}\left(\Delta\right)\|_{*}-\|\mathcal{M}\left(\Delta\right)\|_{*}\right)\\
 & \leq\sqrt{NT}\left\Vert \delta\right\Vert +\left(c_{\nu}+1\right)\|\mathcal{P}\left(\Delta\right)\|_{*}-\left(c_{\nu}-1\right)\|\mathcal{M}\left(\Delta\right)\|_{*}\\
 & \leq\sqrt{NT}\left\Vert \delta\right\Vert +2c_{\nu}\|\mathcal{P}\left(\Delta\right)\|_{*}-\left(c_{\nu}-1\right)\|\Delta\|_{*}
\end{align*}
Hence, within $\mathcal{A}_2\cap\mathcal{A}_3$, we have $\left(\delta,\Delta\right)\in\mathcal{A}$,
where $\mathcal{A}$ is defined in (\ref{ref:cone}). Therefore, by
Assumptions~\ref{assu:rsc}, the excess risk
function is bounded below by
\[
\begin{aligned} & \inf_{\substack{\left(\theta,\Pi\right)\in\mathcal{A}_{2}\mathcal{\cap}\mathcal{A}_{3}\\
\left\Vert \delta\right\Vert ^{2}+\frac{1}{NT}\left\Vert \Delta\right\Vert _{F}^{2}=\gamma_{NT}^{2}
}
}\mathcal{E}\left(\theta,\Pi\right)\\
 & \quad\geq\inf_{\substack{\left(\delta,\Delta\right)\in\mathcal{A}_{2}\cap\mathcal{A}_{3}\\
\left\Vert \delta\right\Vert ^{2}+\frac{1}{NT}\left\Vert \Delta\right\Vert _{F}^{2}=\gamma_{NT}^{2}
}
}c_{RSC}\left[\left\Vert \delta\right\Vert ^{2}+\frac{1}{NT}\|\Delta\|_{F}^{2}\right]\\
 & \quad=c_{RSC}\gamma_{NT}^{2}
\end{aligned}
\]
Summing these up, over the event $\mathcal{A}_{2}\cap\mathcal{A}_{3}$,
we have
\begin{align*}
c_{RSC}\gamma_{NT}^{2} & \leq \psi_{NT}c_{\varepsilon,NT}\left(\left\Vert \delta\right\Vert +\frac{1}{\sqrt{NT}}\left\Vert \Delta\right\Vert _{*}\right)+c_{\nu}\frac{\psi_{NT}c_{\varepsilon,NT}}{\sqrt{NT}}\left(\|\mathcal{P}\left(\Delta\right)\|_{*}-\|\mathcal{M}\left(\Delta\right)\|_{*}\right)\\
 & \leq \psi_{NT}c_{\varepsilon,NT}\left\Vert \delta\right\Vert +\frac{\psi_{NT}c_{\varepsilon,NT}}{\sqrt{NT}}\left\{ \left(c_{\nu}+1\right)\|\mathcal{P}\left(\Delta\right)\|_{*}-\left(c_{\nu}-1\right)\|\mathcal{M}\left(\Delta\right)\|_{*}\right\} \\
 & \leq \psi_{NT}c_{\varepsilon,NT}\left\Vert \delta\right\Vert +\frac{\psi_{NT}c_{\varepsilon,NT}}{\sqrt{NT}}\left(c_{\nu}+1\right)\|\mathcal{P}\left(\Delta\right)\|_{*}\\
 & \leq \psi_{NT}c_{\varepsilon,NT}\left(\left\Vert \delta\right\Vert +\frac{c_{\nu}+1}{\sqrt{NT}}\sqrt{2r}\left\Vert \Delta\right\Vert _{F}\right)\\
 & \leq\left(c_{\nu}+1\right)\sqrt{2r}\psi_{NT}c_{\varepsilon,NT}\left(\left\Vert \delta\right\Vert +\frac{1}{\sqrt{NT}}\left\Vert \Delta\right\Vert _{F}\right)\\
 & \leq\left(c_{\nu}+1\right)\sqrt{4r}\psi_{NT}c_{\varepsilon,NT}\gamma_{NT}
\end{align*}
The first inequality follows from \eqref{ref:aux-nu}. The third inequality holds because $\left\Vert \mathcal{M}(\Delta)\right\Vert_*$ is positive and $c_\nu \geq 2$. The fourth inequality is derived from Lemma~\ref{lemma:matrix}(iii) and (iv). Consequently, $\gamma_{NT} \leq \frac{(c_{\nu} + 1)\sqrt{4r}}{c_{RSC}}\psi_{NT}c_{\varepsilon,NT}$, leading to $P_1 = 0$ w.p.a.1.
\end{proof}
Corollary~\ref{cor:main-r} presents two key results. First, it establishes the consistency of our rank estimator $\hat{r}$. Second, it outlines the convergence rates of $\hat{\Lambda}$ and $\hat{F}$, which form the foundation for the iterative localized estimation in the second step. The proof for the second part is analogous to that of Lemma 2 in \citet*{chen_quantile_2021}, but we provide it here for the sake of completeness.
\begin{proof}[Proof of Corollary \ref{cor:main-r}]
Denote $\gamma_{NT}=\psi_{NT}\sqrt{\frac{rc_{T}}{N\land d_{T}}}$
and thus from Theorem \ref{thm:main-3}, we have $\frac{1}{\sqrt{NT}}\left\Vert \widehat{\Pi}-\Pi_{0}\right\Vert _{F}=O_{p}\left(\gamma_{NT}\right)$. Let $\sigma_{1} > \cdots > \sigma_{r} > 0$ denote the nonzero singular values of $\Pi_{0}$. Note that Assumption~\ref{assu:factor} implies that there exists a constant $c > 0$ such that, with probability approaching one, for all $s \leq r$:
$c^{-1} \sqrt{NT} \leq \sigma_{s}\leq c \sqrt{NT}$, and $\sigma_{s-1}^{2} - \sigma_{s}^{2}\geq c \sqrt{NT}$.

\textbf{Part (i):} To show that our rank estimator $\hat{r}$ is consistent,
we first use Weyl's inequality to obtain the error bound of singular
values. That is, w.p.a.1,
\[
\max_{s\in\{1,\ldots,N\wedge T\}}\left\{ \left|\sigma_{s}\left(\hat{\Pi}\right)-\sigma_{s}\right|\right\} \leq\left\Vert \hat{\Pi}-\Pi_{0}\right\Vert \leq\left\Vert \hat{\Pi}-\Pi_{0}\right\Vert _{F}
\]
Based on the results of Theorem \ref{thm:main-3} and the assumption
that $\sigma_{r+1}\left(\Pi_{0}\right)=\cdots=\sigma_{N\wedge T}\left(\Pi_{0}\right)=0$
(i.e., the true rank is $r$),
\[
\max\left\{ \max_{s\leq r}\left|\sigma_{s}\left(\hat{\Pi}\right)-\sigma_{s}\left(\Pi_{0}\right)\right|,\max_{r+1\leq s\leq N\wedge T}\left|\sigma_{s}\left(\hat{\Pi}\right)\right|\right\} \lesssim \sqrt{NT}\gamma_{NT}
\]
By Assumption \ref{assu:factor}, we have $\sigma_{r}\asymp\sqrt{NT}c_{r}$.
Thus,
\begin{align*}
\min_{s\leq r}\sigma_{s}\left(\hat{\Pi}\right) & \gtrsim\left(c_{r}-\gamma_{NT}\right)\sqrt{NT}=\left(c_{r}+o_{p}\left(1\right)\right)\sqrt{NT}\\
\max_{r+1\leq s\leq N\wedge T}\sigma_{s}\left(\hat{\Pi}\right) & \lesssim\sqrt{NT}\gamma_{NT}=\sqrt{NT}o_{P}\left(\gamma_{NT}^{1/2}\right) = o_p(\sqrt{NT}\big(\nu\|\hat\Pi\|\big)^{1/2})
\end{align*}
Combining the two yields \(P(\hat r=r)\to1\).
This proves the consistency of $\hat{r}$.

\textbf{Part (ii):}
To simplify notation, we omit the hat notation for the estimator and
instead define $\frac{1}{\sqrt{NT}}\left\Vert \Pi-\Pi_{0}\right\Vert _{F}=\gamma_{NT}$,
where $\Pi=\Lambda F^{\prime}$, with $\Lambda$ and $F$ normalized
such that $F^{\prime}F/T=\mathbb{I}_{r}$ and $\Lambda^{\prime}\Lambda/N$
is diagonal with non-increasing diagonal elements. Under this normalization,
$\Lambda$ and $F$ are uniquely identified.

First, let $U\in\mathbb{R}^{r\times r}$ be a diagonal matrix whose
diagonal elements are either $1$ or $-1$. Since $F^{\prime}F/T=F_{0}^{\prime}F_{0}/T=\mathbb{I}_{r}$
and $\left\Vert \Lambda_{0}\right\Vert_F /\sqrt{N}\leq\bar{c}$ by Assumption~\ref{assu:factor},
we have
\[
\begin{aligned}\left\Vert \Lambda-\Lambda_{0}U\right\Vert _{F}/\sqrt{N} & =\left\Vert \left(\Lambda-\Lambda_{0}U\right)F^{\prime}\right\Vert _{F}/\sqrt{NT}=\left\Vert \Lambda F^{\prime}-\Lambda_{0}F_{0}^{\prime}+\Lambda_{0}F_{0}^{\prime}-\Lambda_{0}UF^{\prime}\right\Vert _{F}/\sqrt{NT}\\
 & \leq\left\Vert \Lambda F^{\prime}-\Lambda_{0}F_{0}^{\prime}\right\Vert _{F}/\sqrt{NT}+\left\Vert \Lambda_{0}\right\Vert _{F}/\sqrt{N}\cdot\left\Vert F-F_{0}U\right\Vert _{F}/\sqrt{T}\\
 & \leq\gamma_{NT}+\bar{c}\left\Vert F-F_{0}U\right\Vert /\sqrt{T}.
\end{aligned}
\]
Thus, 
\[
 \left\Vert \Lambda-\Lambda_{0}U\right\Vert _{F}/\sqrt{N}+\left\Vert F-F_{0}U\right\Vert _{F}/\sqrt{T}\leq\gamma_{NT}+\left(1+\bar{c}\right)\left\Vert F-F_{0}U\right\Vert _{F}/\sqrt{T}
\]
Second, 
\[
\begin{aligned}\left\Vert F-F_{0}U\right\Vert _{F}/\sqrt{T} & =\left\Vert F_{0}U-F\left(F^{\prime}F_{0}U/T\right)+F\left(F^{\prime}F_{0}U/T\right)-F\right\Vert _{F}/\sqrt{T}\\
 & \leq\left\Vert F_{0}U-F\left(F^{\prime}F_{0}U/T\right)\right\Vert _{F}/\sqrt{T}+\left\Vert F\left(F^{\prime}F_{0}U/T\right)-F\right\Vert _{F}/\sqrt{T}\\
 & =\left\Vert M_{F}F_{0}\right\Vert _{F}/\sqrt{T}+\left\Vert F^{\prime}F_{0}/T-\mathrm{U}\right\Vert _{F}
\end{aligned}
\]
where $P_{A}=A\left(A^{\prime}A\right)^{-1}A^{\prime}$ and $M_{A}=\mathbb{I}-P_{A}$.

Third, we have 
\begin{align}
 \frac{1}{\sqrt{NT}}\left\Vert \left(\Lambda F^{\prime}-\Lambda_{0}F_{0}^{\prime}\right)M_{F}\right\Vert _{F} & \leq\sqrt{\operatorname{rank}\left[\left(\Lambda F^{\prime}-\Lambda_{0}F_{0}^{\prime}\right)M_{F}\right]}\cdot\left\Vert M_{F}\right\Vert \cdot\left\Vert \Lambda F^{\prime}-\Lambda_{0}F_{0}^{\prime}\right\Vert /\sqrt{NT}\nonumber \\
 &  \lesssim\left\Vert \Lambda F^{\prime}-\Lambda_{0}F_{0}^{\prime}\right\Vert _{F}/\sqrt{NT}=\gamma_{NT}\label{eq:factor1}
\end{align}
and since 
\begin{align}
\left\Vert \left(\Lambda F^{\prime}-\Lambda_{0}F_{0}^{\prime}\right)M_{F}\right\Vert _{F}/\sqrt{NT} & =\left\Vert \Lambda_{0}F_{0}^{\prime}M_{F}\right\Vert _{F}/\sqrt{NT}\nonumber \\
 & =\sqrt{\operatorname{Tr}\left[\left(\Lambda_{0}^{\prime}\Lambda_{0}/N\right)\cdot\left(F_{0}^{\prime}M_{F}F_{0}/T\right)\right]}\nonumber \\
 & \geq\sqrt{\sigma_{Nr}}\sqrt{\operatorname{Tr}\left(F_{0}^{\prime}M_{F}F_{0}/T\right)}\nonumber \\
 &  =\sqrt{\sigma_{Nr}}\left\Vert M_{F}F_{0}\right\Vert _{F}/\sqrt{T}\label{eq:factor2}
\end{align}
It follows from \eqref{eq:factor1} and \eqref{eq:factor2} that 
\begin{equation}
 \left\Vert M_{F}F_{0}\right\Vert _{F}/\sqrt{T}\lesssim\sqrt{\frac{1}{\sigma_{Nr}}}\gamma_{NT}\label{eq:factor6}
\end{equation}

Similarly, it can be shown that 
\begin{equation}
 \left\Vert M_{F_{0}}F\right\Vert _{F}/\sqrt{T}\lesssim\sqrt{\frac{1}{\mu_{\min}\left(\Lambda^{\prime}\Lambda/N\right)}}\gamma_{NT}\label{eq:factor9}
\end{equation}

Fourth, we have 
\[
 \frac{1}{\sqrt{NT}}\left\Vert \left(\Lambda F^{\prime}-\Lambda_{0}F_{0}^{\prime}\right)P_{F}\right\Vert _{F}\leq\frac{1}{\sqrt{NT}}\left\Vert \Lambda F^{\prime}-\Lambda_{0}F_{0}^{\prime}\right\Vert _{F}\cdot\left\Vert P_{F}\right\Vert _{F}=\sqrt{r}\gamma_{NT}
\]
so
\begin{equation}
 \begin{aligned} & \frac{1}{\sqrt{NT}}\left\Vert \left(\Lambda F^{\prime}-\Lambda_{0}F_{0}^{\prime}\right)P_{F}\right\Vert _{F}\\
 & \quad=\frac{1}{\sqrt{NT}}\left\Vert \Lambda F^{\prime}-\Lambda_{0}\left(F_{0}^{\prime}F/T\right)F^{\prime}\right\Vert _{F}\\
 & \quad=\frac{1}{\sqrt{N}}\left\Vert \Lambda-\Lambda_{0}\left(F_{0}^{\prime}F/T\right)\right\Vert _{F}\leq\sqrt{r}\gamma_{NT}
\end{aligned}
\label{eq:factor7}
\end{equation}
Likewise, it can be shown that 
\begin{equation}
 \frac{1}{\sqrt{N}}\left\Vert \Lambda_{0}-\Lambda\left(F^{\prime}F_{0}/T\right)\right\Vert _{F}\leq\sqrt{r}\gamma_{NT}\label{eq:factor8}
\end{equation}

Fifth, define $R_{T}=F^{\prime}F_{0}/T$. Note that $FR_{T}=FF^{\prime}F_{0}/T=P_{F}F_{0}$,
thus
\begin{equation}
\begin{aligned}\mathbb{I}_{r} & =F_{0}^{\prime}F_{0}/T=R_{T}^{\prime}\left(F^{\prime}F/T\right)R_{T}+F_{0}^{\prime}F_{0}/T-R_{T}^{\prime}\left(F^{\prime}F/T\right)R_{T}\\
 & =R_{T}^{\prime}R_{T}+F_{0}^{\prime}F_{0}/T-F_{0}^{\prime}FR_{T}/T+F_{0}^{\prime}FR_{T}/T-R_{T}^{\prime}\left(F^{\prime}F/T\right)R_{T}\\
 & =R_{T}^{\prime}R_{T}+F_{0}^{\prime}\left(F_{0}-FR_{T}\right)/T=R_{T}^{\prime}R_{T}+F_{0}^{\prime}M_{F}F_{0}/T
\end{aligned}
\label{eq:factor3}
\end{equation}
because 
\[
F_{0}^{\prime}FR_{T}/T-R_{T}^{\prime}\left(F^{\prime}F/T\right)R_{T}=\left(F_{0}-FR_{T}\right)^{\prime}FR_{T}/T=F_{0}^{\prime}M_{F}FR_{T}/T=0
\]
In addition, 
\begin{equation}
\begin{aligned} & \Lambda_{0}^{\prime}\Lambda_{0}/N\\
 & \quad=R_{T}^{\prime}\left(\Lambda^{\prime}\Lambda/N\right)R_{T}+\left(\Lambda_{0}^{\prime}\Lambda_{0}/N-R_{T}^{\prime}\left(\Lambda^{\prime}\Lambda/N\right)R_{T}\right)\\
 & \quad=R_{T}^{\prime}\left(\Lambda^{\prime}\Lambda/N\right)R_{T}+\Lambda_{0}^{\prime}\left(\Lambda_{0}-\Lambda R_{T}\right)/N+\left(\Lambda_{0}-\Lambda R_{T}\right)^{\prime}\Lambda R_{T}/N
\end{aligned}
\label{eq:factor4}
\end{equation}
Similar to the proof of \eqref{eq:factor3}, we have 
\begin{equation}
 \mathbb{I}_{r}=R_{T}R_{T}^{\prime}+F^{\prime}\left(F-F_{0}R_{T}^{\prime}\right)/T=R_{T}R_{T}^{\prime}+F^{\prime}M_{F_{0}}F/T\label{eq:factor5}
\end{equation}
From \eqref{eq:factor4}, it holds that 
\[
 \begin{aligned}\Lambda_{0}^{\prime}\Lambda_{0}/N= & R_{T}^{\prime}\left(\Lambda^{\prime}\Lambda/N\right)\left(R_{T}^{\prime}\right)^{-1}R_{T}^{\prime}R_{T}+\Lambda_{0}^{\prime}\left(\Lambda_{0}-\Lambda R_{T}\right)/N+\left(\Lambda_{0}-\Lambda R_{T}\right)^{\prime}\Lambda R_{T}/N\\
= & R_{T}^{\prime}\left(\Lambda^{\prime}\Lambda/N\right)\left(R_{T}^{\prime}\right)^{-1}+R_{T}^{\prime}\left(\Lambda^{\prime}\Lambda/N\right)\left(R_{T}^{\prime}\right)^{-1}\left(R_{T}^{\prime}R_{T}-\mathbb{I}_{r}\right)\\
 & +\Lambda_{0}^{\prime}\left(\Lambda_{0}-\Lambda R_{T}\right)/N+\left(\Lambda_{0}-\Lambda R_{T}\right)^{\prime}\Lambda R_{T}/N
\end{aligned}
\]
and then it follows from the above equation and \eqref{eq:factor3}
that 
\[
\left(\Lambda_{0}^{\prime}\Lambda_{0}/N+D_{NT}\right)R_{T}^{\prime}=R_{T}^{\prime}\left(\Lambda^{\prime}\Lambda/N\right)
\]
where 
\[
D_{NT}=R_{T}^{\prime}\left(\Lambda^{\prime}\Lambda/N\right)\left(R_{T}^{\prime}\right)^{-1}F_{0}^{\prime}M_{F}F_{0}/T-\Lambda_{0}^{\prime}\left(\Lambda_{0}-\Lambda R_{T}\right)/N-\left(\Lambda_{0}-\Lambda R_{T}\right)^{\prime}\Lambda R_{T}/N
\]
From \eqref{eq:factor6}, \eqref{eq:factor7} and \eqref{eq:factor8},
we have $ \left\Vert D_{NT}\right\Vert _{F}\lesssim\gamma_{NT}$.
Hence, by the Bauer-Fike theorem, it holds that 
\begin{equation}
 \left|\sigma_{\min}\left[\Lambda^{\prime}\Lambda/N\right]-\sigma_{\min}\left[\Lambda_{0}^{\prime}\Lambda_{0}/N\right]\right|\leq\left\Vert D_{NT}\right\Vert \leq\left\Vert D_{NT}\right\Vert _{F}\lesssim\gamma_{NT}\label{eq:factor10}
\end{equation}
Moreover, by Assumption~\ref{assu:factor} and the perturbation theory for eigenvectors,
\[
 \left\Vert R_{T}^{\prime}V_{T}\mathrm{~S}-\mathbb{I}_{r}\right\Vert _{F}=\left\Vert R_{T}^{\prime}V_{T}-\mathrm{S}\right\Vert _{F}\lesssim \gamma_{NT}
\]
where $V_{T}=\operatorname{diag}\left(\left(R_{T,1}R_{T,1}^{\prime}\right)^{-1/2},\ldots,\left(R_{T,r}R_{T,r}^{\prime}\right)^{-1/2}\right)$,
and $R_{T,j}^{\prime}$ is the $j$th column of $R_{T}^{\prime}$.

Furthermore, \eqref{eq:factor9} and \eqref{eq:factor10} imply that
$\sigma_{\min}\left[\Lambda^{\prime}\Lambda/N\right]$ is bounded
below by a positive constant, and that $ \left\Vert M_{F_{0}}F\right\Vert _{F}/\sqrt{T}\lesssim\gamma_{NT}$.
From \eqref{eq:factor5}, we then have
\[
\left\Vert V_{T}-\mathbb{I}_{r}\right\Vert _{F}\lesssim\left\Vert R_{T}R_{T}^{\prime}-\mathbb{I}_{r}\right\Vert _{F}=\left\Vert M_{F_{0}}F\right\Vert _{F}^{2}/T\lesssim\gamma_{NT}^{2}
\]

Note that the triangular inequality implies that
\[
 \left\Vert R_{T}^{\prime}-\mathrm{S}\right\Vert _{F}\leq\left\Vert R_{T}^{\prime}V_{T}-\mathrm{S}\right\Vert _{F}+\left\Vert R_{T}^{\prime}V_{T}-R_{T}\right\Vert _{F}\leq\left\Vert R_{T}^{\prime}V_{T}-\mathrm{S}\right\Vert _{F}+\left\Vert R_{T}\right\Vert _{F}\cdot\left\Vert V_{T}-\mathbb{I}_{r}\right\Vert _{F}
\]
Thus, 
\[
 \left\Vert F^{\prime}F_{0}/T-\mathrm{S}\right\Vert _{F}=\left\Vert R_{T}-\mathrm{S}\right\Vert _{F}\lesssim\gamma_{NT}
\]
Finally, setting $U=S$, we then have $ \left\Vert F-F_{0}\mathrm{S}\right\Vert _{F}/\sqrt{T}\lesssim\gamma_{NT}$
and similarly $\left\Vert \Lambda-\Lambda_{0}\mathrm{S}\right\Vert /\sqrt{N}\lesssim\gamma_{NT}$.
\end{proof}
 \subsection{Discussion of Assumptions in Single-Index Models}\label{apx:rsc}
In this section, we specialize to the case of single-index models and introduce two auxiliary assumptions that together lead to the Restricted Strong Convexity condition in Assumption~\ref{assu:rsc}. We further provide primitive, low-level conditions that verify one of these assumptions.

We first recall that the single-index specification assumes that the conditional distribution of $Y_{it}$ given $X_{it}$ depends only on the scalar index $X_{it}'\theta + \pi_{it}$. The common parameter $\theta$ captures the marginal effect of observed covariates, while $\pi_{it}$ accounts for latent, time-varying or individual-specific unobserved heterogeneity. Formally, the loss function can be written as
\[
\ell(W_{it};\theta,\pi_{it}) = \ell\bigl(Y_{it},\, X_{it}'\theta + \pi_{it}\bigr),
\]  
where $W_{it}=(Y_{it},X_{it})$. In this setting, the population model~\eqref{ref:Model} continues to hold, and estimation proceeds via the regularized sample criterion in~\eqref{ref:Estimation1}.

\paragraph{Restricted Strong Convexity}
We now introduce two auxiliary assumptions that together imply the high-level Restricted Strong Convexity (RSC) condition in Assumption~\ref{assu:rsc}. These conditions are formulated to separate local strong convexity of the objective function from separability of $\Delta$ from the covariate $X$.

\begin{assumption}\label{apx-assum:hl-2} 
There exist constants $c_{l}$ and $c^{\prime}_{l}$ such that
for all $\left(N,T\right)$, the inequality $\left\Vert \delta\right\Vert ^{2}+\frac{1}{NT}\left\Vert \Delta\right\Vert _{F}^{2}\leq c^{\prime}_{l}$, where $\delta$ is a p-dimensional vector and $\Delta$ is an $N \times T$ matrix, implies 
\[
 \mathcal{E}\left(\theta_{0}+\delta,\Pi_{0}+\Delta\right)\geq c_{l}\mathbb{E}\left[\frac{1}{NT}\left\Vert \sum_{j=1}^{p}X_{j}\delta_{j}+\Delta\right\Vert _{F}^{2}\right]
\]
\end{assumption} 

Assumption \ref{apx-assum:hl-2} is a local strong convexity condition, similar to those commonly used in the literature on nonconvex low-dimensional M-estimators. Specifically, the condition is imposed within an $\ell_2$-ball of fixed radius around the true parameters. In the low-dimensional setting, i.e., when $\Delta=0$, Assumption \ref{apx-assum:hl-2}
simplifies to the requirement that $\mathcal{E}\left(\theta_{0}+\delta,0\right)\geq\underline{c}\left\Vert \delta\right\Vert ^{2}$ for some constant $\underline{c}$, under the assumption that $\min_{i,t}\mu_{\min}\left(\mathbb{E}\left[X_{it}X_{it}^{\prime}\right]\right)\geq \bar{c}$. 
When fixed effects are present and the loss function is strongly convex and smooth (i.e., the second derivative exists and is bounded away from zero), Assumption~\ref{apx-assum:hl-2} is easily satisfied. For more general loss functions, sufficient conditions for Assumption~\ref{apx-assum:hl-2} are provided in Proposition~\ref{prop:convex} below.

\begin{assumption} \label{apx-assu:rsc} For all $\left(\delta,\Delta\right)\in\mathcal{A}$, there exists a universal constant $c_{RSC}>0$ such that 
\[
\mathbb{E}\left[\left\Vert \sum_{j=1}^{_{p}}X_{j}\delta_{j}+\Delta\right\Vert _{F}^{2}\right]\geq NT\cdot c^{\prime}_{RSC}\left\Vert \delta_{}\right\Vert ^{2}+c^{\prime}_{RSC}\left\Vert \Delta\right\Vert _{F}^{2}
\]
\end{assumption} 
Intuitively, Assumption \ref{apx-assu:rsc} requires
that $\Delta$ cannot be fully explained by $X$, and thus is separable.
Its verification is challenging without imposing further
conditions on parameter space. See also discussions in \citet*{moon_nuclear_2019}
and \citet*{miao_high-dimensional_2022} following their Assumption
1 and Assumption 2, respectively. 

\begin{proposition}\label{prop:A1A2to6}
Suppose Assumptions~\ref{apx-assum:hl-2} and \ref{apx-assu:rsc} hold, then Assumption~\ref{assu:rsc} holds.
\end{proposition}

\begin{proof}
The result follows directly from the definitions of Assumptions~\ref{apx-assum:hl-2} and \ref{apx-assu:rsc}.
\end{proof}

\paragraph{Low-Level Sufficient Conditions.}\label{sec:suff}
This subsection provides primitive sufficient conditions for Assumption~\ref{apx-assum:hl-2} in single-index settings, taking the high-level identification condition in Assumption~\ref{assu:id-2}(ii) as maintained. These conditions verify the local curvature requirement needed for the restricted strong convexity argument.

With a slight abuse of notation, we define the index function as $Z_{it} = X^{\prime}_{it}\theta + \pi_{it} \in \mathcal{Z}$, where $\mathcal{Z} = \{z \in \mathbb{R} \mid z = x^{\prime}\theta + \pi \text{ s.t. } (x, \theta, \pi) \in \mathcal{X} \times \Theta \times \Phi\}$. 
Define 
\[
m(z, x) = \mathbb{E}[\ell(Y_{it}, z) \mid X_{it} = x], 
\quad (z, x) \in \mathcal{Z} \times \mathcal{X},
\]
and let $Z_{0,it} = X_{it}'\theta_{0} + \pi_{0,it}$ denote the index function evaluated at the true parameter values.

We define the expected score function as $s(z,x) = \frac{\partial}{\partial z} m(z, x)$ and the expected Hessian as $H(z,x) = \frac{\partial^2}{\partial z \partial z^{\prime}} m(z, x)$. We denote $s_{it} = s(Z_{it}, X_{it})$ and $s_{0, it} = s(Z_{0, it}, X_{it})$. Similarly, we denote $H_{it} = H(Z_{it}, X_{it})$ and $H_{0, it} = H(Z_{0, it}, X_{it})$.

\begin{proposition}\label{prop:convex}
Suppose Assumption~\ref{assu:id-2}(ii) holds, along with the following conditions:\\
(i) There exists a constant $\bar{c}>0$ such that $\max_{i,t}\mu_{\max}\left(\mathbb{E}\left[X_{it}X_{it}^{\prime}\right]\right)\leq\bar{c}$;\\
(ii) The partial derivatives of $m\left(z,x\right)$ with respect to $z$ up to the third order are uniformly bounded in absolute value over $\left(z,x\right)\in\mathcal{Z}\times\mathcal{X}$; \\
(iii) 
It holds that $\sum_{i=1}^{N}\sum_{t=1}^{T}\mathbb{E}\left[s_{0,it}\left(Z_{it}-Z_{0,it}\right)\right]\geq0$;\\
(iv) The expected Hessian $H_{0,it}$ is positive definite and, in fact, $$\inf_{i\in\left[N\right],t\in\left[T\right]}\inf_{X_{it}\in\mathcal{X}}\mu_{min}\left(H_{0,it}\right)\geq c_{min}$$
for some constant $c_{min}>0$.\\
(v) The following holds,
$$0<q:=\inf_{(\delta,\Delta)\in\mathcal{V}_{1},\delta\neq0}\frac{\left(\mathbb{E}\left(\frac{1}{NT}\sum_{i=1}^{N}\sum_{t=1}^{T}\left(X_{it}^{\prime}\delta+\Delta_{it}\right)^{2}\right)\right)^{3/2}}{\mathbb{E}\left(\frac{1}{NT}\sum_{i=1}^{N}\sum_{t=1}^{T}\left|X_{it}^{\prime}\delta+\Delta_{it}\right|^{3}\right)}$$
where $\mathcal{V}_{1}=\left\{ \left(\delta,\Delta\right)\in\mathbb{R}^{p}\times\mathbb{R}^{N\times T}:\left\Vert \delta\right\Vert ^{2}+\frac{1}{NT}\left\Vert \Delta\right\Vert _{F}^{2}\leq\bar{c}_{l}\right\}$.\\
Then Assumption~\ref{apx-assum:hl-2} holds.
\end{proposition}
Condition (i) requires that the eigenvalues of $\mathbb{E}[X_{it}X_{it}']$ are bounded from above, 
a standard assumption in high-dimensional models. 
Condition (iv) ensures the positive definiteness of the expected Hessian matrix at the true parameter values, 
and Condition (v) imposes a nonlinearity condition as in \citet*{belloni_high_2023} and \citet*{belloni_1-penalized_2011}. 
\begin{proof}[Proof of Proposition \ref{prop:convex}]
Using a third-order Taylor expansion of $m(Z_{it},X_{it})$ around $Z_{0,it}=X_{it}'\theta_0+\pi_{0,it}$, we obtain
\[
\begin{aligned}
\mathcal{E}(Z)
&= \mathbb{E}\left[\frac{1}{NT}\sum_{i,t} 
\bigl(m(Z_{it},X_{it}) - m(Z_{0,it},X_{it})\bigr)\right] \\
&= \frac{1}{NT}\sum_{i,t}\mathbb{E}\Big[
s(Z_{0,it},X_{it})(Z_{it}-Z_{0,it})
+ \tfrac{1}{2}H(Z_{0,it},X_{it})(Z_{it}-Z_{0,it})^{2}
+ r_{Z,it}\Big],
\end{aligned}
\]
where the remainder term satisfies
\[
|r_{Z,it}|
\le \frac{1}{6}\sup_{z\in\mathcal{Z}}\left|
\frac{\partial^{3}}{\partial z^{3}}m(z,X_{it})\right|
|Z_{it}-Z_{0,it}|^{3}
\le \frac{c_M}{6}|Z_{it}-Z_{0,it}|^{3}
\]
for some constant $c_M>0$ by Condition~(ii). 

Since $\mathbb{E}[s(Z_{0,it},X_{it})]=0$ and the smallest eigenvalue of the expected Hessian satisfies 
$\mu_{\min}(H_{0,it})\ge c_{\min}>0$ by Condition~(iv), it follows that
\[
\mathcal{E}(Z)
\ge \frac{1}{NT}\sum_{i,t}\mathbb{E}\left[
\frac{c_{\min}}{2}(Z_{it}-Z_{0,it})^{2}
- \frac{c_{M}}{6}|Z_{it}-Z_{0,it}|^{3}\right].
\]
By the definition of $\mathcal{V}_{1}$ and the boundedness of $X_{it}$, 
we can choose a constant $c_l' \le \bar{c}_l$ such that
\[
\mathbb{E}\left[\frac{1}{NT}\sum_{i,t}(Z_{it}-Z_{0,it})^{2}\right]
\le \left(\frac{3 q c_{\min}}{c_{M}}\right)^{2}.
\]
Substituting this into the previous inequality yields
\[
\mathcal{E}(Z)
\ge \frac{c_{\min}}{6}
\mathbb{E}\left[\frac{1}{NT}\sum_{i,t}(Z_{it}-Z_{0,it})^{2}\right],
\]
which verifies Assumption~\ref{apx-assum:hl-2} with $c_l'$ as defined above and $c_l = c_{\min}/6$.
\end{proof}

 \section{Proofs of Inference Results}
\subsection{Notation and Choice of Norms\label{subsec:inf-norm}}

In this section, I adopt the notations introduced by \citet*{fernandez-val_individual_2016} and \citet*{chen_nonlinear_2021}. I extend their asymptotic
results for ``joint estimators'' to the framework of my iterative
procedure. We collect all these effects in the $r\left(N+T\right)$-vector
$\phi_{NT}=\left(\lambda_{1},\ldots,\lambda_{N},f_{1},\ldots,f_{T}\right)^{\prime}$.
The model parameter $\theta$ is the coefficient of interest, and
the vector $\phi_{NT}$ is treated as a nuisance parameter. The true
values of the parameters, denoted by $\theta_{0}$ and $\phi_{0}=\left(\lambda_{01}^{\prime},\ldots,\lambda_{0N}^{\prime},f_{01}^{\prime},\ldots,f_{0T}^{\prime}\right)^{\prime}$,
are the solution to the population conditional maximum likelihood.

As discussed in Remark 1, the factors and factor loadings are not
separately identified without suitable normalization. In the first-step
estimation we adopt a standard normalization commonly used in linear
and nonlinear factor models; however, for deriving the asymptotic
distribution we employ an alternative normalization, following \citet*{chen_nonlinear_2021}. Specifically,
we impose the restriction that $\sum_{i=1}^{N}\lambda_{0i}\lambda_{0i}^{\prime}=\sum_{t=1}^{T}f_{0t}f_{0t}^{\prime}$,
and define the restricted parameter set as
\[
 \Phi:=\left\{ \phi\in\mathbb{R}^{d_{\phi}}:\sum_{i=1}^{N}\lambda_{0i}\lambda_{i}^{\prime}=\sum_{t=1}^{T}f_{t}f_{0t}^{\prime}\right\} 
\]

Imposing $\hat{\phi}\in\Phi$ is not feasible in practice because
the true parameters appear in the definition of $\Phi$. Nonetheless,
all asymptotic results concern the estimator $\hat{\theta}$, which
is invariant to the choice of normalization for $\hat{\phi}$. Thus,
treating $\hat{\phi}\in\Phi$ as if it were imposed is a technical
device used solely for the proofs. With slight abuse of notation,
we redefine $L_{NT}\left(\theta,\Lambda,F\right)$ for $\mathcal{L}_{NT}\left(\theta,\Lambda,F\right)$
in the main context, and introduce the following normalized criterion
function
\[
\mathcal{L}_{NT}\left(\theta,\Lambda,F\right)=L_{NT}\left(\theta,\Lambda,F\right)-\frac{b}{2NT} \left\Vert \Lambda_{0}^{\prime}\Lambda-F^{\prime}F_{0}\right\Vert _{F}^{2}
\]
where $b>0$ for some constant.

To simplify notations, we suppress the dependence on $NT$ of all
the sequences of functions and parameters, e.g. we use $\mathcal{L}$
for $\mathcal{L}_{NT}$ and $\phi$ for $\phi_{NT}$. Partial derivatives
are denoted with subscripts, e.g. $\partial_{\theta}\mathcal{L}(\theta,\phi)$
denotes $\partial\mathcal{L}(\theta,\phi)/\partial\theta$ and so
on. Let 
\[
\mathcal{S}(\theta,\phi)=\partial_{\phi}\mathcal{L}(\theta,\phi),\quad\mathcal{H}(\theta,\phi)=\partial_{\phi\phi^{\prime}}\mathcal{L}(\theta,\phi)
\]
where $\partial_{x}f$ is the partial derivative of $f$ with respect
to $x$, and additional subscripts denote higher order partial derivatives.
We refer to the $\operatorname{dim}\phi$-vector $\mathcal{S}(\theta,\phi)$
as the incidental parameter score, and to the $\operatorname{dim}\phi\times\operatorname{dim}\phi$
matrix $\mathcal{H}(\theta,\phi)$ as the incidental parameter Hessian.
We drop arguments when evaluating at the true parameter values $\left(\theta_{0},\phi_{0}\right)$,
e.g. $\mathcal{H}=\mathcal{H}\left(\theta_{0},\phi_{0}\right)$. We
use a bar for expectations given $\phi$, such as $\partial_{\theta}\bar{\mathcal{L}}=\mathbb{E}\left[\partial_{\theta}\mathcal{L}\right]$,
and a tilde for variables in deviations from expectations, such as
$\partial_{\theta}\mathcal{\tilde{L}}=\partial_{\theta}\mathcal{L}-\partial_{\theta}\mathcal{\bar{L}}$.

We follow \citet{fernandez-val_individual_2016} (FVM) in using the
Euclidean norm $\|\cdot\|$ for vectors, and the norm induced by the
Euclidean norm for matrices and tensors, i.e. 
\[
\left\Vert \partial_{\theta\theta\theta}\mathcal{L}(\theta,\phi)\right\Vert =\max_{u,v\in\mathbb{R}^{\operatorname{dim}\theta}:\|u\|=1,\|v\|=1}\left\Vert \sum_{k,l=1}^{\operatorname{dim}\theta}u_{k}v_{l}\partial_{\theta\theta_{k}\theta_{l}}\mathcal{L}(\theta,\phi)\right\Vert 
\]
For matrices this induced norm is the spectral norm. Since the number
of fixed effect parameters in the model grows with $N$ and $T$,
the choice of norm for vectors and matrices is important. Following
FVW, we choose the $\ell_{q}$-norm for $\operatorname{dim}\phi$
vectors and the corresponding induced norms for matrices and tensors
\[
\left\Vert \partial_{\phi\phi\phi}\mathcal{L}(\theta,\phi)\right\Vert _{q}=\max_{u,v\in\mathbb{R}^{\operatorname{dim}\phi}:\|u\|=1,\|v\|=1}\left\Vert \sum_{k,l=1}^{\operatorname{dim}\beta}u_{k}v_{l}\partial_{\phi\phi_{k}\phi_{l}}\mathcal{L}(\theta,\phi)\right\Vert _{q}
\]
Note that for $w,x\in\mathbb{R}^{\operatorname{dim}\phi}$ and $q\geq2$,
\[
\left|w^{\prime}x\right|\leq\|w\|_{q}\|x\|_{q/(q-1)}\leq(\operatorname{dim}\phi)^{(q-2)/q}\|w\|_{q}\|x\|_{q}
\]
See FVW for more details on these norms. We also define the sets $\mathcal{B}\left(d,\theta_{0}\right)=\left\{ \theta:\left\Vert \theta-\theta_{0}\right\Vert \leq d\right\} $
and $\mathcal{B}\left(d,\phi_{0}\right)=\left\{ \phi:\left\Vert \phi-\phi_{0}\right\Vert \leq d\right\} $
for $d>0$.

In the remainder of the section, we assume that all assumptions from
Section~\ref{sec:main} hold for all lemmas, unless explicitly stated otherwise.
Additionally, we use the fact that the radius of our localized region
$d_{NT}=o_{p}\left(\left(NT\right)^{-1/4+\epsilon}\right)$, for $0<\epsilon\leq\frac{1}{8}-\frac{1}{2q}$
with $q>4$, when $N$ and $T$ go to infinity at the same rate.
\subsection{Proof of Theorem \ref{thm:inf-conv} and \ref{thm:inf-dist}}

The first lemma is important for the stochastic expansion of $\hat{\theta}^{\left(m+1\right)}$,
as it states that the expected incidental parameter Hessian matrix
is asymptotically diagonal.
\begin{lemma}
\label{lemma:inf-diagonal} We have
\[
\left\Vert \mathcal{\bar{H}}^{-1}-\mathcal{\bar{H}}_{d}^{-1}\right\Vert =O_{p}\left(1\right)
\]
where $\mathcal{\bar{H}}_{d}=\operatorname{diag}\left(\bar{\mathcal{H}}_{(\lambda\lambda)}^{*},\bar{\mathcal{H}}_{(ff)}^{*}\right)$.
Furthermore, for $\theta\in\mathcal{B}\left(d_{NT},\theta_{0}\right)$
and $\phi\in\mathcal{B}\left(\sqrt{N\lor T}d_{NT},\phi_{0}\right)$,
there exists a constant $c>0$ such that $\sqrt{NT}\mathcal{H}(\theta,\phi)\geq c\mathbb{I}_{r\left(N+T\right)}$
w.p.a.1.
\end{lemma}
\begin{proof}[Proof of Lemma \ref{lemma:inf-diagonal}]
The result follows directly from Lemma 2 and Lemma 3 in \citet*{chen_nonlinear_2021}; we therefore omit the proof.
\end{proof}
\begin{lemma}
\label{lemma:inf-aux}
Let $q=8$, $\epsilon=1/\left(16+2\iota\right)$, and $d_{NT}=o_{p}\left(\left(NT\right)^{-1/4+\epsilon}\right)$.
Then, we have \\
(i) For the $q$-norms defined before, we have
\[
\begin{aligned} & \|\mathcal{S}\|_{q}=O_{p}\left(\left(NT\right)^{-3/4+1/(2q)}\right),\quad\left\Vert \partial_{\theta}\mathcal{L}\right\Vert =O_{p}\left(\left(NT\right)^{-1/2}\right),\quad\|\mathcal{\tilde{H}}\|_{q}=o_{p}\left(\left(NT\right)^{-1/2}\right),\\
 & \left\Vert \partial_{\theta\phi^{\prime}}\mathcal{L}\right\Vert _{q}=O_{p}\left(\left(NT\right)^{-1/2+1/(2q)}\right),\quad\left\Vert \partial_{\theta\theta^{\prime}}\mathcal{L}\right\Vert =O_{p}\left(1\right),\quad\left\Vert \partial_{\theta\phi\phi}\mathcal{L}\right\Vert _{q}=O_{p}\left(\left(NT\right)^{-1/2+\epsilon}\right),\\
 & \left\Vert \partial_{\phi\phi\phi}\mathcal{L}\right\Vert _{q}=O_{p}\left(\left(NT\right)^{-1/2+\epsilon}\right),
\end{aligned}
\]
and 
\[
\begin{aligned} & \|\mathcal{S}\|=O_p\left(\left(NT\right)^{-1/2}\right),\quad\left\Vert \mathcal{H}^{-1}\right\Vert =O_p\left(\sqrt{NT}\right),\quad\left\Vert \mathcal{\bar{H}}^{-1}\right\Vert =O_p\left(\sqrt{NT}\right)\\
 & \left\Vert \mathcal{H}^{-1}-\mathcal{\bar{H}}^{-1}\right\Vert =o_{P}\left(\left(NT\right)^{3/8}\right),\quad\left\Vert \mathcal{H}^{-1}-\left(\mathcal{\bar{H}}^{-1}-\mathcal{\bar{H}}^{-1}\mathcal{\tilde{H}}\mathcal{\bar{H}}^{-1}\right)\right\Vert =o_{P}\left(\left(NT\right)^{1/4}\right),\\
 & \left\Vert \partial_{\theta\phi^{\prime}}\mathcal{L}\right\Vert =O_p\left((NT)^{-1/4}\right),\quad\left\Vert \partial_{\theta\phi\phi}\mathcal{L}\right\Vert =O_p\left(\left(NT\right)^{-1/2+\epsilon}\right)\\
 & \left\Vert \sum_{g}\partial_{\phi\phi^{\prime}\phi_{g}}\mathcal{L}\left[\mathcal{H}^{-1}\mathcal{S}\right]_{g}\right\Vert =O_p\left(\left(NT\right)^{-3/4+1/(2q)+\epsilon}\right)\\
 & \left\Vert \sum_{g}\partial_{\phi\phi^{\prime}\phi_{g}}\mathcal{L}\left[\mathcal{\bar{H}}^{-1}\mathcal{S}\right]_{g}\right\Vert =O_p\left(\left(NT\right)^{-3/4+1/(2q)+\epsilon}\right)
\end{aligned}
\]
(ii) Moreover,
\[
\begin{aligned}
    \|\tilde{\mathcal{H}}\| &= o_{P}\left((NT)^{-5/8}\right), \quad \left\Vert \partial_{\theta\theta^{\prime}}\tilde{\mathcal{L}}\right\Vert = o_{P}(1), \quad \left\Vert \partial_{\theta\phi\phi}\tilde{\mathcal{L}}\right\Vert = o_{P}\left((NT)^{-5/8}\right), \\
    \left\Vert \partial_{\theta\phi^{\prime}}\tilde{\mathcal{L}}\right\Vert &= O_p\left((NT)^{-1/2}\right), \quad \left\Vert \sum_{g,h=1}^{\operatorname{dim}\phi} \partial_{\phi\phi_{g}\phi_{h}} \tilde{\mathcal{L}} \left[\overline{\mathcal{H}}^{-1} \mathcal{S}\right]_{g} \left[\overline{\mathcal{H}}^{-1} \mathcal{S}\right]_{h} \right\Vert = o_{P}\left((NT)^{-3/4}\right).
\end{aligned}
\]
(iii) For all $\theta\in\mathcal{B}\left(d_{NT},\theta_{0}\right)$
and $\phi\in\mathcal{B}\left(\sqrt{N\lor T}d_{NT},\phi_{0}\right)$,
we have
\[
\begin{aligned} & \left\Vert \partial_{\theta\theta^{\prime}}\mathcal{L}(\theta,\phi)\right\Vert =O_p\left(1\right),\quad\left\Vert \partial_{\theta\phi^{\prime}}\mathcal{L}(\theta,\phi)\right\Vert _{q}=O_p\left(\left(NT\right)^{-1/2+1/(2q)}\right),\quad\left\Vert \partial_{\theta\theta\theta}\mathcal{L}(\theta,\phi)\right\Vert =O_p\left(1\right)\\
 & \left\Vert \partial_{\theta\theta\phi}\mathcal{L}(\theta,\phi)\right\Vert _{q}=O_p\left(\left(NT\right)^{-1/2+1/(2q)}\right),\quad\left\Vert \partial_{\theta\phi\phi}\mathcal{L}(\theta,\phi)\right\Vert _{q}=O_p\left(\left(NT\right)^{-1/2+\epsilon}\right),\\
 & \left\Vert \partial_{\phi\phi\phi}\mathcal{L}(\theta,\phi)\right\Vert _{q}=O_p\left(\left(NT\right)^{-1/2+\epsilon}\right),\quad\left\Vert \partial_{\theta\theta\phi\phi}\mathcal{L}(\theta,\phi)\right\Vert _{q}=O_p\left(\left(NT\right)^{-1/2+\epsilon}\right),\\
 & \left\Vert \partial_{\theta\phi\phi\phi}\mathcal{L}(\theta,\phi)\right\Vert _{q}=O_p\left(\left(NT\right)^{-1/2+\epsilon}\right),\quad\left\Vert \partial_{\phi\phi\phi\phi}\mathcal{L}(\theta,\phi)\right\Vert _{q}=O_p\left(\left(NT\right)^{-1/2+\epsilon}\right)
\end{aligned}
\]
\end{lemma}
\begin{proof}
The proofs follow those of Theorem 4 in \citet*{chen_nonlinear_2021} and Theorem C.1 and Lemma S.7 in \citet*{fernandez-val_individual_2016}.
\end{proof}

\begin{lemma}
\label{lemma:inf-exp-theta} The following stochastic expansion holds at the updated parameter $\theta$:
\[
\partial_{\theta}\mathcal{L}\left(\hat{\theta}^{\left(m+1\right)},\hat{\phi}^{\left(m\right)}\right)=\bar{W}\left(\theta^{\left(m+1\right)}-\theta_{0}\right)-U^{\left(0\right)}\left(\hat{\theta}^{\left(m\right)}-\theta_{0}\right)- U^{\left(1\right)}- U^{\left(2\right)}+R_{\theta}^{\left(m\right)}
\]
where
\[
\begin{aligned}\bar{W} & =\partial_{\theta\theta^{\prime}}\mathcal{\bar{L}}\\
 U^{\left(0\right)} & =\left(\partial_{\theta\phi^{\prime}}\mathcal{\bar{L}}\right)\mathcal{\bar{H}}^{-1}\left(\partial_{\phi\theta^{\prime}}\mathcal{\bar{L}}\right)\\
U^{\left(1\right)} & =-\partial_{\theta}\mathcal{L}+\left(\partial_{\theta\phi^{\prime}}\mathcal{\bar{L}}\right)\mathcal{\bar{H}}^{-1}\mathcal{S}\\
 U^{\left(2\right)} & =\left(\left[\partial_{\theta\phi^{\prime}}\mathcal{\tilde{L}}\right]\bar{\mathcal{H}}^{-1}\mathcal{S}-\left[\partial_{\theta\phi^{\prime}}\mathcal{\bar{L}}\right]\mathcal{\bar{H}}^{-1}\tilde{\mathcal{H}}\mathcal{\bar{H}}^{-1}\mathcal{S}\right)\\
 & +\frac{1}{2}\sum_{g\in\left[\dim\left(\phi\right)\right]}\left(\partial_{\theta\phi^{\prime}\phi_{g}}\mathcal{\bar{L}}+\left[\partial_{\theta\phi^{\prime}}\mathcal{\bar{L}}\right]\bar{\mathcal{H}}^{-1}\left[\partial_{\phi\phi^{\prime}\phi_{g}}\mathcal{\bar{L}}\right]\right)\mathcal{\bar{H}}^{-1}\mathcal{S}\left(\mathcal{\bar{H}}^{-1}\mathcal{S}\right)_{g}
\end{aligned}
\]
and the remainder term satisfies $R_{\theta}^{\left(m\right)}=o_{p}\left(\left\Vert \hat{\theta}^{\left(m+1\right)}-\theta_{0}\right\Vert \right)+o_{p}\left(\left\Vert \hat{\theta}^{\left(m\right)}-\theta_{0}\right\Vert \right)+o_{P}\left(\left(NT\right)^{-1/2}\right)$. 
\end{lemma}
\begin{proof}[Proof of Lemma~\ref{lemma:inf-exp-theta}]
 Based on Theorem B.1(i) in \citet{fernandez-val_individual_2016}, we have 
 \[
\begin{aligned} & \widehat{\phi}\left(\theta\right)-\phi_{0}\\
 & \quad=-\mathcal{H}^{-1}\left(\partial_{\phi\theta^{\prime}}\mathcal{L}\right)\left(\theta-\theta_{0}\right)-\mathcal{H}^{-1}\mathcal{S}-\frac{1}{2}\mathcal{H}^{-1}\sum_{g}\left(\partial_{\phi\phi^{\prime}\phi_{g}}\mathcal{L}\right)\mathcal{H}^{-1}\mathcal{S}\left(\mathcal{H}^{-1}\mathcal{S}\right)_{g}+R_{1\phi}\left(\theta\right)
\end{aligned}
\]
and the remainder term of the expansion satisfies 
\[
\begin{aligned}\left\Vert R_{1\phi}\left(\theta\right)\right\Vert _{q} & =\sup_{\|v\|_{q/(q-1)}=1}v^{\prime}R_{1\phi}(\theta)\\
 & =O_{P}\left[(NT)^{1/q+\epsilon}\|\theta-\theta_{0}\|^{2}+(NT)^{-1/4+1/q+\epsilon}\|\theta-\theta_{0}\|+(NT)^{-3/4+3/(2q)+2\epsilon}\right]\\
 & =o_{P}\left(\left(NT\right)^{1/8+1/(2q)}\left\Vert \theta-\theta_{0}\right\Vert ^{2}\right)+o_{P}\left(\left(NT\right)^{-1/8+1/(2q)}\left\Vert \theta-\theta_{0}\right\Vert \right)\\
 & +o_{P}\left(\left(NT\right)^{-1/2+1/(2q)}\right)
\end{aligned}
\]
We can further decompose the deviation of the incidental parameter estimator as
\begin{align}
 & \hat{\phi}(\theta)-\phi_{0}\nonumber \\
 & \quad=-\mathcal{\bar{H}}^{-1}\left(\partial_{\phi\theta^{\prime}}\mathcal{\bar{L}}\right)\left(\theta^{\left(m\right)}-\theta_{0}\right)-\mathcal{\bar{H}}^{-1}\mathcal{S}+\mathcal{\bar{H}}^{-1}\tilde{\mathcal{H}}\mathcal{\bar{H}}^{-1}\mathcal{S}\nonumber \\
 & \quad-\frac{1}{2}\mathcal{\bar{H}}^{-1}\sum_{g}\left(\partial_{\phi\phi^{\prime}\phi_{g}}\mathcal{\bar{L}}\right)\mathcal{\bar{H}}^{-1}\mathcal{S}\left(\mathcal{\bar{H}}^{-1}\mathcal{S}\right)_{g}+R_{1\phi}(\theta)+R_{2\phi}\left(\theta\right)\label{eq:inf-phi}
\end{align}
where 
\[
\begin{aligned} R_{2\phi}\left(\theta\right)
&\quad=\left[\mathcal{H}^{-1}\left(\partial_{\phi\theta^{\prime}}\mathcal{L}\right)-\mathcal{\bar{H}}^{-1}\left(\partial_{\phi\theta^{\prime}}\mathcal{\bar{L}}\right)\right]\left(\theta-\theta_{0}\right)+\left[\mathcal{H}^{-1}-\left(\mathcal{\bar{H}}^{-1}-\mathcal{\bar{H}}^{-1}\tilde{\mathcal{H}}\mathcal{\bar{H}}^{-1}\right)\right]\mathcal{S}\\
 & \quad+\frac{1}{2}\left[\mathcal{H}^{-1}\sum_{g}\left(\partial_{\phi\phi^{\prime}\phi_{g}}\mathcal{L}\right)\mathcal{H}^{-1}\mathcal{S}\left(\mathcal{H}^{-1}\mathcal{S}\right)_{g}-\mathcal{\bar{H}}^{-1}\sum_{g}\left(\partial_{\phi\phi^{\prime}\phi_{g}}\mathcal{\bar{L}}\right)\mathcal{\bar{H}}^{-1}\mathcal{S}\left(\mathcal{\bar{H}}^{-1}\mathcal{S}\right)_{g}\right]\\
\\\end{aligned}
\]
Again, by Lemma~\ref{lemma:inf-aux}, we have 
\[
\begin{aligned} & \|R_{\phi}\left(\theta\right)\|\\
 & \quad\leq\left\Vert \mathcal{H}^{-1}-\mathcal{\bar{H}}^{-1}\right\Vert \left\Vert \partial_{\phi^{\prime}\theta}\mathcal{L}\right\Vert \|\theta-\theta_{0}\|+\left\Vert \mathcal{\bar{H}}^{-1}\right\Vert \left\Vert \partial_{\phi^{\prime}\theta}\mathcal{\tilde{L}}\right\Vert \|\theta-\theta_{0}\|\\
 & \quad+\left\Vert \mathcal{H}^{-1}-\left(\mathcal{\bar{H}}^{-1}-\mathcal{\bar{H}}^{-1}\tilde{\mathcal{H}}\mathcal{\bar{H}}^{-1}\right)\right\Vert \|\mathcal{S}\|\\
 & \quad+ \frac{1}{2}\left(\left\Vert \mathcal{H}^{-1}\right\Vert +\left\Vert \mathcal{\bar{H}}^{-1}\right\Vert \right)\left\Vert \mathcal{H}^{-1}-\mathcal{\bar{H}}^{-1}\right\Vert \|\mathcal{S}\|\left\Vert \sum_{g}\partial_{\phi\phi^{\prime}\phi_{g}}\mathcal{L}\left[\mathcal{H}^{-1}\mathcal{S}\right]_{g}\right\Vert \\
 & \quad+\frac{1}{2}\left\Vert \mathcal{H}^{-1}-\mathcal{\bar{H}}^{-1}\right\Vert \left\Vert \mathcal{\bar{H}}^{-1}\right\Vert \|\mathcal{S}\|\left\Vert \sum_{g}\partial_{\phi\phi^{\prime}\phi_{g}}\mathcal{L}\left[\mathcal{\bar{H}}^{-1}\mathcal{S}\right]_{g}\right\Vert \\
 & \quad+\frac{1}{2}\left\Vert \mathcal{\bar{H}}^{-1}\right\Vert \|\mathcal{S}\|\left\Vert \sum_{g,h}\partial_{\phi\phi_{g}\phi_{h}}\mathcal{\tilde{L}}\left[\mathcal{\bar{H}}^{-1}\mathcal{S}\right]_{g}\mathcal{\tilde{L}}\left[\mathcal{\bar{H}}^{-1}\mathcal{S}\right]_{h}\right\Vert \\
 & \quad=o_{p}\left(\left(NT\right)^{1/8}\left\Vert \theta-\theta_{0}\right\Vert \right)+o_{p}\left(\left(NT\right)^{-1/4}\right)
\end{aligned}
\]
uniformly over $\theta\in\mathcal{B}\left(d_{NT},\theta_{0}\right)$.
Note that by Lemma~\ref{lemma:inf-aux}, we have 
\[
\begin{aligned}\left\Vert \widehat{\phi}\left(\theta\right)-\phi_{0}\right\Vert _{q} & \leq\left\Vert \mathcal{H}^{-1}\right\Vert _{q}\|\mathcal{S}\|_{q}+\left\Vert \mathcal{H}^{-1}\right\Vert _{q}\left\Vert \partial_{\phi\theta^{\prime}}\mathcal{L}\right\Vert _{q}\left\Vert \theta-\theta_{0}\right\Vert _{q}\\
 & +\frac{1}{2}\left\Vert \mathcal{H}^{-1}\right\Vert _{q}^{3}\left\Vert \partial_{\phi\phi\phi}\mathcal{L}\right\Vert _{q}\|\mathcal{S}\|_{q}^{2}+\left\Vert R_{1\phi}\left(\theta\right)\right\Vert _{q}\\
 & =O_{P}\left(\left(NT\right)^{-1/4+1/\left(2q\right)}\right)+O_{p}\left(\left(NT\right)^{1/\left(2q\right)}\left\Vert \theta-\theta_{0}\right\Vert \right)\\
 & =O_{p}\left(\left(NT\right)^{1/\left(2q\right)}\left\Vert \theta-\theta_{0}\right\Vert \right)
\end{aligned}
\]
Next, We expand $\partial_{\theta}\mathcal{L}\left(\hat{\theta}^{\left(m+1\right)},\hat{\phi}^{\left(m\right)}\right)$
around $\left(\theta_{0},\phi_{0}\right)$ and get 
\[
\begin{aligned}\partial_{\theta}\mathcal{L}\left(\hat{\theta}^{\left(m+1\right)},\hat{\phi}^{\left(m\right)}\right)= & \partial_{\theta}\mathcal{L}+\left(\partial_{\theta\theta^{\prime}}\mathcal{L}\right)\left(\hat{\theta}^{\left(m+1\right)}-\theta_{0}\right)+\left(\partial_{\theta\phi^{\prime}}\mathcal{L}\right)\left(\hat{\phi}^{\left(m\right)}-\phi_{0}\right)\\
+ & \frac{1}{2}\sum_{g\in\left[\dim\left(\phi\right)\right]}\left[\partial_{\theta\phi^{\prime}\phi_{g}}\mathcal{L}\right]\mathcal{H}^{-1}\mathcal{S}\left(\mathcal{H}^{-1}\mathcal{S}\right)_{g}+R_{1\theta}^{\left(m\right)}\end{aligned}
\]
and the remainder term satisfies 
\[
\begin{aligned}\left\Vert R_{1\theta}^{\left(m\right)}\right\Vert  & =\sup_{\|v\|=1}v^{\prime}R_{1\theta}^{\left(m\right)}\\
 & \leq\frac{1}{2}\left\Vert \sum_{g\in\left[\dim\left(\phi\right)\right]}\left[\partial_{\theta\phi^{\prime}\phi_{g}}\mathcal{L}\right]\left(\hat{\phi}^{\left(m\right)}-\phi_{0}\right)\left(\hat{\phi}^{\left(m\right)}-\phi_{0}\right)_{g}-\sum_{g\in\left[\dim\left(\phi\right)\right]}\left[\partial_{\theta\phi^{\prime}\phi_{g}}\mathcal{L}\right]\mathcal{H}^{-1}\mathcal{S}\left(\mathcal{H}^{-1}\mathcal{S}\right)_{g}\right\Vert _{q}\\
 & +\frac{1}{2}\left\Vert \partial_{\theta\theta\theta}\mathcal{L}\left(\theta^{*},\phi_{0}\right)\right\Vert \left\Vert \hat{\theta}^{\left(m+1\right)}-\theta_{0}\right\Vert ^{2}+\left(NT\right)^{1/2-1/q}\left\Vert \partial_{\theta\theta\phi}\mathcal{L}\left(\theta_{0},\phi^{*}\right)\right\Vert _{q}\\
 & \cdot\left\Vert \hat{\phi}^{\left(m\right)}-\phi_{0}\right\Vert _{q}\left\Vert \hat{\theta}^{\left(m+1\right)}-\theta_{0}\right\Vert +\frac{1}{6}\left(NT\right)^{1/2-1/q}\left\Vert \partial_{\theta\phi\phi\phi}\mathcal{L}\left(\theta^{*},\phi_{0}\right)\right\Vert _{q}\left\Vert \hat{\phi}^{\left(m\right)}-\phi_{0}\right\Vert _{q}^{3}\\
 & =o_{p}\left(\left\Vert \hat{\theta}^{\left(m+1\right)}-\theta_{0}\right\Vert \right)+o_{p}\left(\left\Vert \hat{\theta}^{\left(m\right)}-\theta_{0}\right\Vert \right)+O_{p}\left(\left\Vert \hat{\theta}^{\left(m+1\right)}-\theta_{0}\right\Vert ^{2}\right)+o_{P}\left(\left(NT\right)^{-1/2}\right)\\
 & =o_{p}\left(\left\Vert \hat{\theta}^{\left(m+1\right)}-\theta_{0}\right\Vert \right)+o_{p}\left(\left\Vert \hat{\theta}^{\left(m\right)}-\theta_{0}\right\Vert \right)+o_{P}\left(\left(NT\right)^{-1/2}\right)
\end{aligned}
\]
where $\theta^{*}$ lies between $\hat{\theta}^{\left(m+1\right)}$
and $\theta_{0}$ and $\phi^{*}$ lies between $\hat{\phi}^{\left(m\right)}$
and $\phi_{0}$.

We can further decompose 
\[
\begin{aligned} & \partial_{\theta}\mathcal{L}(\hat{\theta}^{\left(m+1\right)},\hat{\phi}^{\left(m\right)})\\
 & \quad=\partial_{\theta}\mathcal{L}+\left(\partial_{\theta\theta^{\prime}}\mathcal{\bar{L}}\right)\left(\hat{\theta}^{\left(m+1\right)}-\theta_{0}\right)+\left(\partial_{\theta\phi^{\prime}}\mathcal{\bar{L}}\right)\left(\hat{\phi}^{\left(m\right)}-\phi_{0}\right)\\
 & \quad+\frac{1}{2}\sum_{g}\left[\partial_{\theta\phi^{\prime}\phi_{g}}\mathcal{\bar{L}}\right]\mathcal{H}^{-1}\mathcal{S}\left(\mathcal{H}^{-1}\mathcal{S}\right)_{g}+R_{\theta}^{\left(m\right)}\\
 & \quad=\bar{W}\left(\hat{\theta}^{\left(m+1\right)}-\theta_{0}\right)- U^{\left(0\right)}\left(\hat{\theta}^{\left(m\right)}-\theta_{0}\right)- U^{\left(1\right)}- U^{\left(2\right)}+R_{\theta}^{\left(m\right)}
\end{aligned}
\]
where the last equality uses (\ref{eq:inf-phi}), and $R_{\theta}^{\left(m\right)}$ takes the form 
\[
\begin{aligned}
R_{\theta}^{\left(m\right)} & =R_{1\theta}^{\left(m\right)}+\left(\partial_{\theta\theta^{\prime}}\mathcal{\tilde{L}}\right)\left(\hat{\theta}^{\left(m+1\right)}-\theta_{0}\right)+\left(\partial_{\theta\phi^{\prime}}\mathcal{\tilde{L}}\right)\mathcal{\bar{H}}^{-1}\left(\partial_{\phi\theta^{\prime}}\mathcal{\bar{L}}\right)\left(\theta^{\left(m\right)}-\theta_{0}\right)-\left(\partial_{\theta\phi^{\prime}}\mathcal{\tilde{L}}\right) \mathcal{\bar{H}}^{-1}\tilde{\mathcal{H}}\mathcal{\bar{H}}^{-1}\mathcal{S}\\
 & +\frac{1}{2}\left[\sum_{g}\left(\partial_{\theta\phi^{\prime}\phi_{g}}\mathcal{L}\right)\mathcal{H}^{-1}\mathcal{S}\left(\mathcal{H}^{-1}\mathcal{S}\right)_{g}-\sum_{g}\left(\partial_{\theta\phi^{\prime}\phi_{g}}\mathcal{\bar{L}}\right)\mathcal{\bar{H}}^{-1}\mathcal{S}\left(\mathcal{\bar{H}}^{-1}\mathcal{S}\right)_{g}\right]\\
 & +\frac{1}{2}\left(\partial_{\theta\phi^{\prime}}\mathcal{\tilde{L}}\right)\mathcal{\bar{H}}^{-1}\sum_{g}\left(\partial_{\phi\phi^{\prime}\phi_{g}}\mathcal{\bar{L}}\right)\mathcal{\bar{H}}^{-1}\mathcal{S}\left(\mathcal{\bar{H}}^{-1}\mathcal{S}\right)_{g}+\left(\partial_{\phi\theta^{\prime}}\mathcal{\bar{L}}\right)\left(R_{1\phi}\left(\theta^{\left(m\right)}\right)+R_{2\phi}\left(\theta^{\left(m\right)}\right)\right)
\end{aligned}
\] 
Furthermore, from Lemma~\ref{lemma:inf-aux}, we have 
\[
\begin{aligned}\left\Vert R_{\theta}^{\left(m\right)}\right\Vert  & \leq\left\Vert R_{1\theta}^{\left(m\right)}\right\Vert +\left\Vert \partial_{\theta\theta^{\prime}}\mathcal{\tilde{L}}\right\Vert \left\Vert \hat{\theta}^{\left(m+1\right)}-\theta_{0}\right\Vert +\left\Vert \partial_{\theta\phi^{\prime}}\mathcal{\tilde{L}}\right\Vert \left\Vert \mathcal{\bar{H}}^{-1}\right\Vert \left\Vert \partial_{\phi\theta^{\prime}}\mathcal{\bar{L}}\right\Vert \left\Vert \hat{\theta}^{\left(m\right)}-\theta_{0}\right\Vert \\
 & + \left\Vert \mathcal{\bar{H}}^{-1}\right\Vert ^{2}\left\Vert \partial_{\theta\phi^{\prime}}\mathcal{\tilde{L}}\right\Vert \left\Vert \mathcal{\tilde{H}}\right\Vert \left\Vert \mathcal{S}\right\Vert \\
 & +\frac{1}{2}\left\Vert \partial_{\theta\phi\phi}\mathcal{L}\right\Vert \left(\left\Vert \mathcal{H}^{-1}\right\Vert +\left\Vert \mathcal{\bar{H}}^{-1}\right\Vert \right)\left\Vert \mathcal{H}^{-1}-\mathcal{\bar{H}}^{-1}\right\Vert \left\Vert \mathcal{S}\right\Vert ^{2}\\
 & +\frac{1}{2}\left\Vert \mathcal{\bar{H}}^{-1}\right\Vert ^{2}\left\Vert \partial_{\theta\phi\phi}\mathcal{\tilde{L}}\right\Vert \left\Vert \mathcal{S}\right\Vert ^{2}+\left(NT\right)^{1/2-1/q}\left\Vert \partial_{\phi\theta^{\prime}}\mathcal{\bar{L}}\right\Vert _{q}\left\Vert R_{1\phi}\left(\theta^{\left(m\right)}\right)\right\Vert _{q}+\left\Vert \partial_{\phi\theta^{\prime}}\mathcal{\bar{L}}\right\Vert \left\Vert R_{1\theta}^{\left(m\right)}\right\Vert \\
 & =\left\Vert R_{1\theta}^{\left(m\right)}\right\Vert +o_{P}\left(\left\Vert \hat{\theta}^{\left(m+1\right)}-\theta_{0}\right\Vert \right)+O_{P}\left(\left(NT\right)^{-1/4}\left\Vert \hat{\theta}^{\left(m\right)}-\theta_{0}\right\Vert \right)\\
 & +o_{P}\left(\left(NT\right)^{-1/8+\epsilon}\left\Vert \hat{\theta}^{\left(m\right)}-\theta_{0}\right\Vert \right)+o_{P}\left(\left(NT\right)^{-1/2}\right)\\
 & =o_{p}\left(\left\Vert \hat{\theta}^{\left(m+1\right)}-\theta_{0}\right\Vert \right)+o_{p}\left(\left\Vert \hat{\theta}^{\left(m\right)}-\theta_{0}\right\Vert \right)+o_{P}\left(\left(NT\right)^{-1/2}\right)
\end{aligned}
\]
The desirable results then follow. 
\end{proof}

Before proceeding to the next lemma, we introduce several quantities. For any $N\times T$ matrix $A$ we define the $N\times T$ matrix
$\mathbb{P}A$ as follows 
\[
\left(\mathbb{P}A\right){}_{it}=\lambda_{i}^{*\prime}f_{0t}+\lambda_{0i}^{\prime}f_{t}^{*},\quad \left(\lambda^{*},f^{*}\right)\in\underset{\lambda_{i},f_{t}}{\arg\min}\sum_{i,t}\mathbb{E}\left[\partial_{\pi^{2}}\ell_{it}\right]\left(A_{it}-\lambda'_{i}f_{0t}-\lambda_{0i}^{\prime}f_{t}\right)^{2}
\]
Note that $\mathbb{PP}=\mathbb{P}$. We also define the linear operator
$\mathbb{\tilde{P}}$ as 
\begin{equation}
\label{aux:Atilde}
\mathbb{\tilde{P}}A=\mathbb{P}\tilde{A},\quad\tilde{A}_{it}=\frac{A_{it}}{\mathbb{E}\left[\partial_{\pi^{2}}\ell_{it}\right]}
\end{equation}

The next lemma gives a convenient algebraic result, which is used
extensively in deriving the asymptotic expansion for the estimated
common parameter. This is a straightforward extension of Lemma S.8
in FVW.
\begin{lemma}
\label{lemma:inf-algebraic} Let $A,B$ and $C$ be $N\times T$ matrices,
and let the expected incidental parameter Hessian $\mathcal{\bar{H}}$
be invertible. Define the $r\left(N+T\right)$ vectors $\mathcal{A}$
and $\mathcal{B}$ and the $r\left(N+T\right)\times r\left(N+T\right)$
matrix $\mathcal{C}$ as follows\footnote{Here, $\left(\sum_{t}A_{it}f_{0t}\right)_{i\in\left[N\right]}$ is
a $rN$-dimensional vector with its $i$-th block to be a $r$-dimensional
vector, $\sum_{t}A_{it}f_{0t}$ .}
\[
\mathcal{A}=\frac{1}{NT}\left(\begin{array}{c}
\left(\sum_{t}A_{it}f_{0t}\right)_{i\in\left[N\right]}\\
\left(\sum_{i}A_{it}\lambda_{0i}\right)_{t\in\left[T\right]}
\end{array}\right),\quad\mathcal{B}=\frac{1}{NT}\left(\begin{array}{c}
\left(\sum_{t}B_{it}f_{0t}\right)_{i\in\left[N\right]}\\
\left(\sum_{i}B_{it}\lambda_{0i}\right)_{t\in\left[T\right]}
\end{array}\right)
\]
and 
$$\mathcal{C}=\frac{1}{NT}\left(\begin{array}{cc}
\operatorname{diag}\left(\left(\sum_{t}C_{it}f_{0t}f_{0t}^{\prime}\right)_{i\in\left[N\right]}\right) & \left(C_{it}\lambda_{0i}f_{0t}^{\prime}\right)_{i\in\left[N\right],t\in\left[T\right]}\\
\left(C_{it}\lambda_{0i}f_{0t}^{\prime}\right)_{i\in\left[N\right],t\in\left[T\right]}^{\prime} & \operatorname{diag}\left(\left(\sum_{i}C_{it}\lambda_{0i}\lambda_{0i}^{\prime}\right)_{t\in\left[T\right]}\right)
\end{array}\right)$$
Then \\
(i) $\mathcal{A}^{\prime}\mathcal{\bar{H}}^{-1}\mathcal{B}=\frac{1}{NT}\sum_{i=1}^{N}\sum_{t=1}^{T}\left(\mathbb{\tilde{P}}A\right)_{it}B_{it}=\frac{1}{NT}\sum_{i=1}^{N}\sum_{t=1}^{T}\left(\mathbb{\tilde{P}}B\right)_{it}A_{it}$; \\ 
(ii) $\mathcal{A}^{\prime}\mathcal{\bar{H}}^{-1}\mathcal{B}=\frac{1}{NT}\sum_{i=1}^{N}\sum_{t=1}^{T}\mathbb{E}\left[\partial_{\pi^{2}}\ell_{it}\right]\left(\mathbb{\tilde{P}}A\right)_{it}\left(\mathbb{\tilde{P}}B\right)_{it}$; \\
(iii) $\mathcal{A}^{\prime}\mathcal{\bar{H}}^{-1}\mathcal{C}^{-1}\mathcal{\bar{H}}^{-1}\mathcal{B}=\frac{1}{NT}\sum_{i=1}^{N}\sum_{t=1}^{T}\left(\mathbb{\tilde{P}}A\right)_{it}C_{it}\left(\mathbb{\tilde{P}}B\right)_{it}$. 
\end{lemma}
\begin{proof}
To simplify the notations, we consider the case $r=1$, but
the proof can be easily generalized to the case $r>1$. Let $\lambda_{i}^{*}f_{0t}+\lambda_{0i}f_{t}^{*}=\left(\mathbb{P}\tilde{A}\right){}_{it}=\left(\mathbb{\tilde{P}}A\right){}_{it}$,
with $\tilde{A}_{it}$ defined in \eqref{aux:Atilde}. The first order condition of
the minimization problem in the definition of $\left(\mathbb{P}\tilde{A}\right){}_{it}$
can be written as $\mathcal{\bar{H}}\left(\begin{array}{c}
\lambda^{*}\\
f^{*}
\end{array}\right)=\mathcal{A}$. One solution to this is $\left(\begin{array}{c}
\lambda^{*}\\
f^{*}
\end{array}\right)=\mathcal{\bar{H}}^{-1}\mathcal{A}$. Therefore, $\mathcal{A}^{\prime}\mathcal{\bar{H}}^{-1}\mathcal{B}=\left(\begin{array}{c}
\lambda^{*}\\
f^{*}
\end{array}\right)\mathcal{B}=\frac{1}{NT}\sum_{i=1}^{N}\sum_{t=1}^{T}\left(\mathbb{\tilde{P}}A\right){}_{it}B_{it}$. This is the first equality of the statement (i), and the second
equality in statement (i) follows by symmetry. Statement (ii) is a
special case of statement (iii) with $\mathcal{C}=\mathcal{\bar{H}}$.
Let $\lambda_{i}f_{0t}+\lambda_{0i}f_{t}=\left(\mathbb{P}\tilde{B}\right){}_{it}=\left(\mathbb{\tilde{P}}B\right){}_{it}$
with $\tilde{B}_{it}=\frac{B_{it}}{\mathbb{E}\left[\partial_{\pi^{2}}\ell_{it}\right]}$.
Analogous to the above, choose $\left(\begin{array}{c}
\lambda\\
f
\end{array}\right)=\mathcal{\bar{H}}^{-1}\mathcal{B}$ as one solution to the minimization problem. Then $\mathcal{A}^{\prime}\mathcal{\bar{H}}^{-1}\mathcal{C}^{-1}\mathcal{\bar{H}}^{-1}\mathcal{B}=\frac{1}{NT}\sum_{i=1}^{N}\sum_{t=1}^{T}\left(\mathbb{\tilde{P}}A\right)_{it}C_{it}\left(\mathbb{\tilde{P}}B\right)_{it}$.
\end{proof} \begin{lemma}
\label{lemma:inf-U} The approximate Hessian and the terms of the
score defined in Lemma~\ref{lemma:inf-exp-theta} can be written
as 
\[
\begin{aligned}\bar{W} & =\frac{1}{NT}\sum_{i=1}^{N}\sum_{t=1}^{T}\mathbb{E}\left[\partial_{\theta\theta}\ell_{it}\right]\\
U^{\left(0\right)} & =\frac{1}{NT}\sum_{i=1}^{N}\sum_{t=1}^{T}\mathbb{E}\left[\partial_{\pi^{2}}\ell_{it}\Xi_{it}\Xi_{it}^{\prime}\right]\\
U^{\left(1\right)} & =\frac{1}{NT}\sum_{i=1}^{N}\sum_{t=1}^{T}-D_{\theta}\ell_{it}\\
U^{\left(2\right)} & =-\frac{1}{NT}\sum_{i=1}^{N}\sum_{t=1}^{T}\Lambda_{it}\left(D_{\theta\pi}\ell_{it}-\mathbb{E}\left[D_{\theta\pi}\ell_{it}\right]\right)+\frac{1}{2NT}\sum_{i=1}^{N}\sum_{t=1}^{T}\Lambda_{it}^{2}\mathbb{E}\left[D_{\theta\pi^{2}}\ell_{it}\right]
\end{aligned}
\]
where $D_{\theta\pi^{q}}\ell_{it}=\partial_{\theta\pi^{q}}\ell_{it}-\Xi_{it}\partial_{\pi^{q+1}}\ell_{it}$
with $q=0,1,2$. \end{lemma}
\begin{proof}[Proof of Lemma \ref{lemma:inf-U}]
 First, it is by definition that 
\[
\bar{W}=\frac{1}{NT}\sum_{i=1}^{N}\sum_{t=1}^{T}\mathbb{E}\left[\partial_{\theta\theta}\ell_{it}\right]
\]
The terms $U^{\left(0\right)}$ is easy by Lemma \ref{lemma:inf-algebraic}(ii).
Since $\mathbb{E}\left[\partial_{\theta}\mathcal{L}\right]=0$ and
$\mathbb{E}\left[\mathcal{S}\right]=0$, we have $\mathbb{E}\left[U^{\left(1\right)}\right]=0$.
Also, from Lemma \ref{lemma:inf-exp-theta} and Lemma \ref{lemma:inf-algebraic}(i),
we have 
\[
U^{\left(1\right)}=\frac{1}{NT}\sum_{i=1}^{N}\sum_{t=1}^{T}\left(-\partial_{\theta}\ell_{it}+\Xi_{it}\partial_{\pi}\ell_{it}\right)=\frac{1}{NT}\sum_{i=1}^{N}\sum_{t=1}^{T}-D_{\theta}\ell_{it}
\]
It remains to show the decomposition of $U^{\left(2\right)}$. Define
$\Lambda_{d,it}$ to be 
\begin{equation}
\Lambda_{it}\coloneqq -\frac{1}{NT}\sum_{j=1}^{N}\sum_{\tau=1}^{T}\left(f_{0t}^{\prime}\mathcal{\bar{H}}_{(\lambda\lambda)ij}^{-1}f_{0\tau}+f_{0t}^{\prime}\mathcal{\bar{H}}_{(\lambda f)i\tau}^{-1}\lambda_{0j}+\lambda_{0i}^{\prime}\mathcal{\bar{H}}_{(f\lambda)tj}^{-1}f_{0\tau}+\lambda_{0i}^{\prime}\mathcal{\bar{H}}_{(ff)t\tau}^{-1}\lambda_{0j}\right)\partial_{\pi}\ell_{j\tau}\label{eqn:inf-def-Lambda}
\end{equation}
In addition, with Lemma \ref{lemma:inf-algebraic}(iii), 
\begin{align}
U^{\left(2\right)} &  =\left[\partial_{\theta\phi^{\prime}}\mathcal{\tilde{L}}\right]\bar{H}^{-1}\mathcal{S}-\left[\partial_{\theta\phi^{\prime}}\mathcal{\bar{L}}\right]\mathcal{\bar{H}}^{-1}\tilde{\mathcal{H}}\mathcal{\bar{H}}^{-1}\mathcal{S}\nonumber \\
 & +\frac{1}{2}\sum_{g=1}\left(\partial_{\theta\phi^{\prime}\phi_{g}}\mathcal{\bar{L}}+\left[\partial_{\theta\phi^{\prime}}\mathcal{\bar{L}}\right]\bar{\mathcal{H}}^{-1}\left[\partial_{\phi\phi^{\prime}\phi_{g}}\mathcal{\bar{L}}\right]\right)\mathcal{\bar{H}}^{-1}\mathcal{S}\left(\mathcal{\bar{H}}^{-1}\mathcal{S}_{g}\right)\nonumber \\
 & = -\frac{1}{NT}\sum_{i=1}^{N}\sum_{t=1}^{T}\Lambda_{it}\left(\partial_{\theta\pi}\tilde{\ell}_{it}+\Xi_{it}\partial_{\pi^{2}}\tilde{\ell}_{it}\right)\nonumber \\
 &  +\frac{1}{2NT}\sum_{i=1}^{N}\sum_{t=1}^{T}\Lambda_{it}^{2}\left(\mathbb{E}\left[\partial_{\theta\pi^{2}}\ell_{it}\right]+\left(\partial_{\theta\phi^{\prime}}\mathcal{\bar{L}}\right)\bar{\mathcal{H}}^{-1}\mathbb{E}\left[\partial_{\phi}\partial_{\pi^{2}}\ell_{it}\right]\right)\nonumber \\
 & =\underbrace{-\frac{1}{NT}\sum_{i=1}^{N}\sum_{t=1}^{T}\Lambda_{it}\left(D_{\theta\pi}\ell_{it}-\mathbb{E}\left[D_{\theta\pi}\ell_{it}\right]\right)}_{U^{\left(2a\right)}}\underbrace{+\frac{1}{2NT}\sum_{i=1}^{N}\sum_{t=1}^{T}\Lambda_{it}^{2}\mathbb{E}\left[D_{\theta\pi^{2}}\ell_{it}\right]}_{U^{\left(2b\right)}}\label{eqn:inf-U-decomp}
\end{align}
where the penultimate equality uses Lemma \ref{lemma:inf-algebraic}(i).
For each $i$ and $t$, $\partial_{\phi}\partial_{\pi^{2}}\ell_{it}$
is a dim-$\phi$ vector, which can be written as $\partial_{\phi}\partial_{\pi^{2}}\ell_{it}=\left(\begin{array}{c}
A1_{rT}\\
A^{\prime}1_{rN}
\end{array}\right)$ for an $r\left(N\times T\right)$ matrix $A$ with elements $A_{j\tau}=\partial_{\pi^{3}}\ell_{j\tau}$
if $j=i$ and $\tau=t$, and $A_{j\tau}=0$ otherwise. Thus, Lemma
\ref{lemma:inf-algebraic}(i) gives $\left(\partial_{\theta\phi^{\prime}}\mathcal{\bar{L}}\right)\bar{\mathcal{H}}^{-1}\partial_{\phi}\partial_{\pi^{2}}\ell_{it}= -\sum_{j,\tau}\Xi_{j\tau}\delta_{(i=j)}\delta_{(t=\tau)}\partial_{\pi^{3}}\ell_{it}=-\Xi_{it}\partial_{\pi^{3}}\ell_{it}$. 
\end{proof}
\begin{lemma}
\label{lem:inf-bias}The terms of the score defined in Lemma~\ref{lemma:inf-exp-theta}
have the following limiting behaviors, $U^{\left(0\right)}=O_{p}\left(1/\sqrt{NT}\right)$,
$U^{(1)}=O_{p}\left(1/\sqrt{NT}\right)$, and \textup{$U^{\left(2\right)}\stackrel{p}{\longrightarrow}\bar{B}_{\infty}/T+\bar{D}_{\infty}/N$},
where \textup{$\bar{B}_{\infty}$ and $\bar{D}_{\infty}$ are defined
in Theorem} \ref{thm:inf-dist}.
\end{lemma}
\begin{proof}[Proof of Lemma \ref{lem:inf-bias}]
It is clear that $U^{(0)}=O_{p}\left(\frac{1}{\sqrt{NT}}\right)$,
and by CLT, $\sqrt{NT}U^{(1)}\stackrel{d}{\longrightarrow}N\left(0,\bar{\Sigma}_{\infty}\right)$,
where 
\[
\bar{\Sigma}_{\infty}=\mathbb{\lim}_{N,T\rightarrow\infty} \frac{1}{NT}\sum_{i=1}^{N}\sum_{t=1}^{T}\mathbb{E}\left[\left(D_{\theta}\ell_{it}\right)\left(D_{\theta}\ell_{it}\right)^{\prime}\right]
\]
As for $U^{\left(2\right)}$, the main focus is to show that the bias
formulas by taking account the specific structure of the incidental
parameters Hessian. Decompose 
\[
\Lambda_{it}=\Lambda_{it}^{(1)}+\Lambda_{it}^{(2)}+\Lambda_{it}^{(3)}+\Lambda_{it}^{(4)}
\]
with 
\[
\begin{aligned}\Lambda_{it}^{(1)} & =-\frac{1}{NT}\sum_{j=1}^{N}f_{0t}^{\prime}\mathcal{\bar{H}}_{(\lambda\lambda)ij}^{-1}\sum_{\tau=1}^{T}\partial_{\pi}\ell_{j\tau}f_{0\tau}\\
\Lambda_{it}^{(2)} &  =-\frac{1}{NT}\sum_{\tau=1}^{T}f_{0t}^{\prime}\mathcal{\bar{H}}_{(\lambda f)i\tau}^{-1}\sum_{j=1}^{N}\partial_{\pi}\ell_{j\tau}\lambda_{0j}\\
\Lambda_{it}^{(3)} &  =-\frac{1}{NT}\sum_{j=1}^{N}\lambda_{0i}^{\prime}\mathcal{\bar{H}}_{(f\lambda)tj}^{-1}\sum_{\tau=1}^{T}\partial_{\pi}\ell_{j\tau}f_{0\tau}\\
\Lambda_{it}^{(4)} &  =-\frac{1}{NT}\sum_{\tau=1}^{T}\lambda_{0i}^{\prime}\mathcal{\bar{H}}_{(ff)t\tau}^{-1}\sum_{j=1}^{N}\partial_{\pi}\ell_{j\tau}\lambda_{0j}
\end{aligned}
\]
We look at two decomposition terms $U^{(2a)}$ and $U^{(2b)}$ as
specified in (\ref{eqn:inf-U-decomp}). Let 
\[
U^{(2a)}=U^{(2a,1)}+U^{(2a,2)}+U^{(2a,3)}+U^{(2a,4)}
\]
where 
\[
\begin{aligned} & U^{(2a,1)}=\frac{1}{\left(NT\right){}^{2}}\sum_{i,j,t}f_{0t}^{\prime}\mathcal{\bar{H}}_{(\lambda\lambda)ij}^{-1}\left(\sum_{\tau}\partial_{\pi}\ell_{j\tau}f_{0\tau}\right)\left(D_{\theta\pi}\ell_{it}-\mathbb{E}\left[D_{\theta\pi}\ell_{it}\right]\right)\\
 & U^{(2a,2)}=\frac{1}{\left(NT\right){}^{2}}\sum_{i,\tau,t}f_{0t}^{\prime}\mathcal{\bar{H}}_{(\lambda f)i\tau}^{-1}\left(\sum_{j}\partial_{\pi}\ell_{j\tau}\lambda_{0j}\right)\left(D_{\theta\pi}\ell_{it}-\mathbb{E}\left[D_{\theta\pi}\ell_{it}\right]\right)\\
 & U^{(2a,3)}=\frac{1}{\left(NT\right){}^{2}}\sum_{t,j,i}\lambda_{0i}^{\prime}\mathcal{\bar{H}}_{(f\lambda)tj}^{-1}\left(\sum_{\tau}\partial_{\pi}\ell_{j\tau}f_{0\tau}\right)\left(D_{\theta\pi}\ell_{it}-\mathbb{E}\left[D_{\theta\pi}\ell_{it}\right]\right)\\
 & U^{(2a,4)}=\frac{1}{\left(NT\right){}^{2}}\sum_{t,\tau,i}\lambda_{0i}^{\prime}\mathcal{\bar{H}}_{(ff)t\tau}^{-1}\left(\sum_{j}\partial_{\pi}\ell_{j\tau}\lambda_{0j}\right)\left(D_{\theta\pi}\ell_{it}-\mathbb{E}\left[D_{\theta\pi}\ell_{it}\right]\right)
\end{aligned}
\]

Use $\left\Vert \mathcal{\bar{H}}^{-1}-\mathcal{\bar{H}}_{d}^{-1}\right\Vert =O_{p}\left(1\right)$
by Lemma \ref{lemma:inf-diagonal}, and apply the Cauchy-Schwarz inequality
to the sum over $t$ in $U^{(2a,3)}$, and since both $\lambda_{0i}^{\prime}\mathcal{\bar{H}}_{(f\lambda)}^{-1}\partial_{\pi}\ell_{j\tau}f_{0\tau}$
and $\left(D_{\theta\pi}\ell_{it}-\mathbb{E}\left[D_{\theta\pi}\ell_{it}\right]\right)$
are mean zero, independent across $i$, 
\[
\begin{aligned}\left(U^{(2a,3)}\right)^{2} & \leq\frac{1}{(NT)^{4}}\left[\sum_{t}\left(\sum_{j,\tau}\lambda_{0i}^{\prime}\mathcal{\bar{H}}_{(f\lambda)tj}^{-1}\partial_{\pi}\ell_{j\tau}f_{0\tau}\right)^{2}\right]\left[\sum_{t}\left(\sum_{i}\left(D_{\theta\pi}\ell_{it}-\mathbb{E}\left[D_{\theta\pi}\ell_{it}\right]\right)\right)^{2}\right]\\
 & =\frac{1}{(NT)^{4}}\left[\sum_{t}O_{p}(NT)\right]\left[\sum_{t}O_{p}(N)\right]=O_{p}\left(1/\left(N^{2}T\right)\right)=o_{p}\left(\frac{1}{NT}\right)
\end{aligned}
\]
Thus, $U^{(2a,3)}=o_{p}(1/\sqrt{NT})$. Analogously, $U^{(2a,2)}=o_{p}(1/\sqrt{NT})$.

According to Lemma~\ref{lemma:inf-diagonal}, $\bar{\mathcal{H}}_{(\lambda\lambda)}^{-1}=\operatorname{diag}\left[\left(\frac{1}{NT}\sum_{t=1}^{T}\mathbb{E}\left[\partial_{\pi^{2}}\ell_{it}\right]f_{0t}f_{0t}^{\prime}\right)^{-1}\right]+O_{p}(1)$.
Analogously to the proof of $U^{(2a,3)}$, the $O_{p}(1)$ part of
$\bar{\mathcal{H}}_{(\lambda\lambda)}^{-1}$ has an asymptotically
negligible contribution to $U^{(2a,1)}$. 
\[
\begin{aligned} & U^{(2a,1)}\\
 & \quad=\frac{1}{(NT)^{2}}\sum_{i,j}f_{0t}^{\prime}\bar{\mathcal{H}}_{(\lambda\lambda)ij}^{-1}\left(\sum_{\tau}\partial_{\pi}\ell_{j\tau}f_{0\tau}\right)\sum_{t}\left(D_{\theta\pi}\ell_{it}-\mathbb{E}\left[D_{\theta\pi}\ell_{it}\right]\right)\\
 & \quad=\frac{1}{NT}\sum_{i}\left\{ f_{0t}^{\prime}\left(\frac{1}{NT}\sum_{t=1}^{T}\mathbb{E}\left[\partial_{\pi^{2}}\ell_{it}\right]f_{0t}f_{0t}^{\prime}\right)^{-1}\right.\cdot\\
 & \quad\quad\left.\left(\sum_{\tau}\partial_{\pi}\ell_{i\tau}f_{0\tau}\right)\sum_{t}\left(D_{\theta\pi}\ell_{it}-\mathbb{E}\left[D_{\theta\pi}\ell_{it}\right]\right)\right\} +o_{p}\left(\frac{1}{\sqrt{NT}}\right)\\
 & \quad=\frac{1}{T}\cdot\underbrace{\frac{1}{N}\sum_{i=1}^{N}\sum_{t=1}^{T}\sum_{\tau=t}^{T}f_{0t}^{\prime}\left(\sum_{t=1}^{T}\mathbb{E}\left[\partial_{\pi^{2}}\ell_{it}\right]f_{0t}f_{0t}^{\prime}\right)^{-1}f_{0\tau}\mathbb{E}\left[\partial_{\pi}\ell_{it}D_{\theta\pi}\ell_{i\tau}\right]}_{\eqqcolon\bar{B}^{(1)}}+o_{p}(\frac{1}{\sqrt{NT}})
\end{aligned}
\]

Note, previous assumptions guarantee that $\mathbb{E}\left[\left(U_{i}^{(2a,1)}\right)^{2}\right]=O_{p}\left(1/\left(NT\right)\right)$
uniformly over $i$. For the numerator, both $\partial_{\pi}\ell_{i\tau}f_{0\tau}$
and $\left(D_{\theta\pi}\ell_{it}-\mathbb{E}\left[D_{\theta\pi}\ell_{it}\right]\right)$
are mean zero weakly correlated processes, hence the sum over which
is of order $\sqrt{T}$ each. The denominator of of $U_{i}^{(1a,1)}$
is of order $T$ as it sums over $T$. The last equality applies the
WLLN over $i$, $\frac{1}{N}\sum_{i}U_{i}^{(2a,1)}=\frac{1}{N}\mathbb{E}U_{i}^{(2a,1)}+o_{P}(1)$,
and by using $\mathbb{E}\left[\partial_{\pi}\ell_{it}D_{\theta\pi}\ell_{i\tau}\right]=0$
for $t>\tau$. Analogously, 
\[
U^{(2a,4)}=\frac{1}{N}\cdot\underbrace{\frac{1}{T}\sum_{t=1}^{T}\sum_{i=1}^{N}\lambda_{0i}^{\prime}\left(\sum_{i=1}^{N}\mathbb{E}\left[\partial_{\pi^{2}}\ell_{it}\right]\lambda_{0i}\lambda_{0i}^{\prime}\right)^{-1}\lambda_{0i}\mathbb{E}\left[\partial_{\pi}\ell_{it}D_{\theta\pi}\ell_{it}\right]}_{\eqqcolon\bar{D}^{(1)}}+o_{p}\left(\frac{1}{\sqrt{NT}}\right)
\]
The definition of $\Lambda_{it}$ induces the following decomposition
of $U^{\left(2b\right)}$, 
\[
U^{(2b)}=\sum_{s,l=1}^{4}U^{(2b,s,l)},\quad U^{(2b,s,l)}=\frac{1}{2NT}\sum_{i,t}\Lambda_{it}^{(s)}\Lambda_{it}^{(l)}\mathbb{E}\left[D_{\theta\pi^{2}}\ell_{it}\right]
\]
Due to the symmetry $U^{(2b,s,l)}=U^{(2b,l,s)}$, this decomposition
has 10 distinct terms. Starting with $ U^{(2b,1,3)}$,
we have 

\[
\begin{aligned}U^{(2b,1,3)} & =\frac{1}{NT}\sum_{i=1}^{N}U_{i}^{(2b,1,3)}\\
 U_{i}^{(2b,1,3)} & =\frac{1}{2T}\sum_{t=1}^{T}\mathbb{E}\left[D_{\theta\pi^{2}}\ell_{it}\right]f_{0t}^{\prime}U_{it}^{(2b,1,3)}\lambda_{0i}\\
 U_{it}^{(2b,1,3)} & =\frac{1}{N^{2}}\sum_{j_{1},j_{2}=1}^{N}\left\{ \bar{\mathcal{H}}_{(\lambda\lambda)ij_{1}}^{-1}\left(\frac{1}{\sqrt{T}}\sum_{\tau=1}^{T}\partial_{\pi}\ell_{j_{1}\tau}f_{0\tau}\right)\left(\frac{1}{\sqrt{T}}\sum_{\tau=1}^{T}\partial_{\pi}\ell_{j_{2}\tau}f_{0\tau}^{\prime}\right)\bar{\mathcal{H}}_{(f\lambda)tj_{2}}^{-1}\right\} 
\end{aligned}
\]
Given that $\mathbb{E}\left[\sum_{t}\partial_{\pi}\ell_{it}f_{0t}\right]=0$
and $\mathbb{E}\left[\sum_{t}\partial_{\pi}\ell_{it}f_{0t}\sum_{\tau}\partial_{\pi}\ell_{j\tau}f_{0\tau}\right]=0$
for $i\neq j$, along with the properties of the inverse expected
Hessian from Lemma~\ref{lemma:inf-diagonal}, we have $\mathbb{E}\left[U_{i}^{(2b,1,3)}\right]=O_{p}(1/N)$
uniformly over $i$, $\mathbb{E}\left[\left(U_{i}^{(2b,1,3)}\right)^{2}\right]=O_{p}(1)$
uniformly over $i$, and $\mathbb{E}\left[U_{i}^{(2b,1,3)}U_{j}^{(2b,1,3)}\right]=O_{p}(1/N)$
uniformly over $ i\neq j$. This implies that $\mathbb{E}\left[\sqrt{NT}U^{(2b,1,3)}\right]=O_{p}(1/N)$,
and $\mathbb{E}\left[NT\left(U^{(2b,1,3)}-\mathbb{E}\left[U^{(2b,1,3)}\right]\right)^{2}\right]=O_{p}(1/\sqrt{N})$.
Therefore $ U^{(2b,1,3)}=o_{p}\left(1/\sqrt{NT}\right)$.
By similar arguments, we have $U^{(2b,s,l)}=o_{p}\left(1/\sqrt{NT}\right)$
for all combinations of $s,l=1,2,3,4$, except for $s=l=1$ and $s=l=4$.
For $s=l=1$, 
\[
\begin{aligned}U^{(2b,1,1)} & =\frac{1}{NT}\sum_{i=1}^{N}U_{i}^{(2b,1,1)}\\
 U_{i}^{(2b,1,1)} & = \frac{1}{2T}\sum_{t=1}^{T}\mathbb{E}\left[D_{\theta\pi^{2}}\ell_{it}\right]f_{0t}^{\prime}U_{it}^{(2b,1,1)}f_{0t}\\
U_{it}^{(2b,1,1)} & =\frac{1}{N^{2}}\sum_{j_{1},j_{2}=1}^{N}\left\{ \bar{\mathcal{H}}_{(\lambda\lambda)ij_{1}}^{-1}\left(\frac{1}{\sqrt{T}}\sum_{\tau=1}^{T}\partial_{\pi}\ell_{j_{1}\tau}f_{0\tau}\right)\left(\frac{1}{\sqrt{T}}\sum_{\tau=1}^{T}\partial_{\pi}\ell_{j_{2}\tau}f_{0\tau}^{\prime}\right)\bar{\mathcal{H}}_{(\lambda\lambda)ij_{2}}^{-1}\right\} 
\end{aligned}
\]
Analogous to the result for $U^{(2b,1,3)}$, we have $\mathbb{E}\left[NT\left(U^{(2b,1,1)}-\mathbb{E}U^{(2b,1,1)}\right)^{2}\right]=O_{p}(1/\sqrt{N})$.
Thus, 
\[
\begin{aligned}U^{(2b,1,1)} & =\mathbb{E}\left[U^{(2b,1,1)}\right]+o_{p}\left(\frac{1}{\sqrt{NT}}\right)\\
 & =-\frac{1}{2NT}\sum_{i=1}^{N}\sum_{t=1}^{T}\mathbb{E}\left[D_{\theta\pi^{2}}\ell_{it}\right]f_{0t}^{\prime}\left(\sum_{t=1}^{T}\mathbb{E}\left[\partial_{\pi^{2}}\ell_{it}\right]f_{0t}f_{0t}^{\prime}\right)^{-1}\mathbb{E}\left[\left(\partial_{\pi}\ell_{it}f_{0t}\right)\left(\partial_{\pi}\ell_{it}f_{0t}\right)^{\prime}\right]\\
 & \quad\cdot\left(\sum_{t=1}^{T}\mathbb{E}\left[\partial_{\pi^{2}}\ell_{it}\right]f_{0t}f_{0t}^{\prime}\right)^{-1}f_{0t}+o_{p}\left(\frac{1}{\sqrt{NT}}\right)\\
 & =\frac{1}{T}\cdot\underbrace{\frac{1}{2N}\sum_{i=1}^{N}\sum_{t=1}^{T}f_{0t}^{\prime}\left(\sum_{t=1}^{T}\mathbb{E}\left[\partial_{\pi^{2}}\ell_{it}\right]f_{0t}f_{0t}^{\prime}\right)^{-1}f_{0t}\mathbb{E}\left[D_{\theta\pi^{2}}\ell_{it}\right]}_{\eqqcolon\bar{B}^{(2)}}+o_{p}\left(\frac{1}{\sqrt{NT}}\right)
\end{aligned}
\]
Similarly, 
\[
\begin{aligned}U^{(2b,4,4)} & =\mathbb{E}U^{(2b,4,4)}+o_{p}\left(\frac{1}{\sqrt{NT}}\right)\\
 & =\frac{1}{N}\cdot\underbrace{\frac{1}{2T}\sum_{t=1}^{T}\sum_{i=1}^{N}\lambda_{0i}^{\prime}\left(\sum_{i=1}^{N}\mathbb{E}\left[\partial_{\pi^{2}}\ell_{it}\right]\lambda_{0i}\lambda_{0i}^{\prime}\right)^{-1}\lambda_{0i}\mathbb{E}\left[D_{\theta\pi^{2}}\ell_{it}\right]}_{\eqqcolon\bar{D}^{(2)}}+o_{p}\left(\frac{1}{\sqrt{NT}}\right)
\end{aligned}
\]
Summing up the above terms, we have $U^{(2a)}=\frac{1}{T}\bar{B}^{(1)}+\frac{1}{N}\bar{D}^{(1)}+o_{p}\left(\frac{1}{\sqrt{NT}}\right)$
and $U^{(2b)}=\frac{1}{T}\bar{B}^{(2)}+\frac{1}{N}\bar{D}^{(2)}+o_{p}\left(\frac{1}{\sqrt{NT}}\right)$.
Since $\bar{B}_{\infty}=\lim_{N,T\rightarrow\infty}\left(\bar{B}^{(1)}+\bar{B}^{(2)}\right)$
and $\bar{D}_{\infty}=\lim_{N,T\rightarrow\infty}\left(\bar{D}^{(1)}+\bar{D}^{(2)}\right)$,
then $U^{(2)}=\frac{1}{T}\bar{B}_{\infty}+\frac{1}{N}\bar{D}_{\infty}+o_{p}(1)$.
We have shown that 
\begin{equation}
U^{(2)}\stackrel{p}{\longrightarrow}\frac{1}{T}\bar{B}_{\infty}+\frac{1}{N}\bar{D}_{\infty}\label{eqn:U2}
\end{equation}
\end{proof}
Now, we can proceed to prove our main theorems. 
\begin{proof}[Proof of Theorem~\ref{thm:inf-conv}(i)]
Now we have 
\[
\begin{aligned} & \theta^{\left(m+1\right)}-\theta_{0}\\
 & \quad=\bar{W}^{-1}U^{\left(0\right)}\left(\theta^{\left(m\right)}-\theta_{0}\right)+\bar{W}^{-1}U^{\left(1\right)}+\bar{W}^{-1}U^{\left(2\right)}\\
 & \quad+o_{p}\left(\left\Vert \theta^{\left(m+1\right)}-\theta_{0}\right\Vert \right)+o_{p}\left(\left\Vert \theta^{\left(m\right)}-\theta_{0}\right\Vert \right)+o_{P}\left(\left(NT\right)^{-1/2}\right)
\end{aligned}
\]
Let $\bar{C}^{\left(0\right)}=\bar{W}^{-1}\left(\partial_{\theta\phi^{\prime}}\mathcal{\bar{L}}\right)\mathcal{\bar{H}}^{-1}\left(\partial_{\phi\theta^{\prime}}\mathcal{\bar{L}}\right)$.
We want to show that $\sigma_{\max}\left(\bar{C}^{\left(0\right)}\right)\in\left(0,1\right]$.
Note that by the Bartlett identities, $\mathbb{E}\left[\partial_{\theta}\mathcal{L}\partial_{\theta^{\prime}}\mathcal{L}\right]=\partial_{\theta\theta^{\prime}}\mathcal{\bar{L}}$,
$\mathbb{E}\left[\partial_{\theta}\mathcal{L}\mathcal{S}^{\prime}\right]= \partial_{\theta\phi^{\prime}}\mathcal{\bar{L}}$,
and $\mathbb{E}\left[\mathcal{SS}^{\prime}\right]=\bar{\mathcal{H}}$,
then 
\[
\begin{aligned} & \left(\partial_{\theta}\mathcal{L}-\left(\partial_{\theta\phi^{\prime}}\mathcal{\bar{L}}\right)\mathcal{\bar{H}}^{-1}\mathcal{S}\right)\left(\partial_{\theta}\mathcal{L}-\left(\partial_{\theta\phi^{\prime}}\mathcal{\bar{L}}\right)\mathcal{\bar{H}}^{-1}\mathcal{S}\right)'\\
 & \quad=\partial_{\theta}\mathcal{L}\partial_{\theta'}\mathcal{L}-2\partial_{\theta}\mathcal{L}\mathcal{S}\mathcal{\bar{H}}^{-1}\left(\partial_{\theta\phi^{\prime}}\mathcal{\bar{L}}\right)'+\left(\partial_{\theta\phi^{\prime}}\mathcal{\bar{L}}\right)\mathcal{\bar{H}}^{-1}\mathcal{S}\mathcal{S}'\mathcal{\bar{H}}^{-1}\left(\partial_{\theta\phi^{\prime}}\mathcal{\bar{L}}\right)'\\
 & \quad\stackrel{p}{\longrightarrow}\partial_{\theta\theta^{\prime}}\mathcal{\bar{L}}-\left(\partial_{\theta\phi^{\prime}}\mathcal{\bar{L}}\right)\mathcal{\bar{H}}^{-1}\left(\partial_{\phi\theta^{\prime}}\mathcal{\bar{L}}\right)
\end{aligned}
\]
Therefore, $\bar{W}\left(\mathbb{I}_{p}-\bar{C}^{\left(0\right)}\right)=\partial_{\theta\theta^{\prime}}\mathcal{\bar{L}}-\left(\partial_{\theta\phi^{\prime}}\mathcal{\bar{L}}\right)\mathcal{\bar{H}}^{-1}\left(\partial_{\phi\theta^{\prime}}\mathcal{\bar{L}}\right)>0$.
By Fact 8.14.20 in Bernstein (2005, p.329), 
\begin{align*}
\sigma_{\min}\left(\mathbb{I}_{p}-\bar{C}^{\left(0\right)}\right) & \geq\sigma_{\max}\left(\bar{W}^{-1}\right)\sigma_{\min}\left(\bar{W}\left(\mathbb{I}_{p}-\bar{C}^{\left(0\right)}\right)\right)\\
 & =\sigma_{\max}^{-1}\left(\bar{W}\right)\sigma_{\min}\left(\bar{W}\left(\mathbb{I}_{p}-\bar{C}^{\left(0\right)}\right)\right)\coloneqq\bar{\rho}\in\left(0,1\right]
\end{align*}
where we use that $\bar{W}-\bar{W}\left(\mathbb{I}_{p}-\bar{C}^{\left(0\right)}\right)=\bar{W}\bar{C}^{\left(0\right)}\geq0$.
This implies that $\|\bar{C}^{\left(0\right)}\|=\sigma_{\max}\left(\bar{C}^{\left(0\right)}\right)=1-\sigma_{\min}\left(\mathbb{I}_{p}-\bar{C}^{\left(0\right)}\right)\le 1-\bar{\rho}\in\left[0,1\right)$
w.p.a.1. 
\end{proof}
\begin{proof}[Proof of Theorem~\ref{thm:inf-conv}(ii)]
 By Theorem~\ref{thm:inf-conv}(i) and the fact that $\bar{W}^{-1} U^{\left(1\right)}+\bar{W}^{-1} U^{\left(2\right)}=O_{p}\left(\left(NT\right)^{-1/2}\right)$,
we have $\left\Vert \hat{\theta}^{\left(m+1\right)}-\theta_{0}\right\Vert =O_{p}\left(\left\Vert \hat{\theta}^{\left(m\right)}-\theta_{0}\right\Vert \right)$.
Continuous backward substitutions then give \begin{align}
\hat{\theta}^{\left(m+1\right)}-\theta_{0} & =\left[\bar{W}^{-1} U^{\left(0\right)}+o_{p}\left(1\right)\right]^{m+1}\left(\hat{\theta}^{\left(0\right)}-\theta_{0}\right)\label{eqn:inf-m1}\\
 & +\sum_{s=0}^{m}\left[\bar{W}^{-1} U^{\left(0\right)}\right]^{s}\left\{ \bar{W}^{-1} U^{\left(1\right)}+\bar{W}^{-1} U^{\left(2\right)}+o_{P}\left(\left(NT\right)^{-1/2}\right)\right\} \label{eqn:inf-m2}
\end{align}
Thus, when $m+1\geq-\left(\frac{1}{2}\log\left(NT\right)+\log\left(\gamma_{NT}\right)\right)/\log\left(\|\bar{C}^{\left(0\right)}\|\right)$,
we can readily have 
\[
\left[\bar{W}^{-1}U^{\left(0\right)}+o_{p}\left(1\right)\right]^{m+1}\left(\hat{\theta}^{\left(0\right)}-\theta_{0}\right)=O_{p}\left(\left(NT\right)^{-1/2}\right)
\]
by Theorem~\ref{thm:main-3}. This, in conjunction with the fact that
$\bar{W}^{-1} U^{\left(1\right)}+\bar{W}^{-1} U^{\left(2\right)}=O_{p}\left(\left(NT\right)^{-1/2}\right)$,
implies $\theta^{\left(m+1\right)}-\theta_{0}=O_{p}\left(\left(NT\right)^{-1/2}\right)$. 
\end{proof}
\begin{proof}[Proof of Theorem~\ref{thm:inf-dist}]
For any $m\geq-\frac{1}{2}\log(NT)/\log\left(\|\bar{C}^{\left(0\right)}\|\right)-1$,
the cumulative effect (\ref{eqn:inf-m1}) can be controlled as follows
\[
\left\Vert \text{(\ref{eqn:inf-m1})}\right\Vert \leq\sigma_{\max}^{m+1}\left(\bar{C}^{(0)}\right)O_{p}\left(\gamma_{NT}\right)=O_{p}\left(\left(NT\right)^{-1/2}\gamma_{NT}\right)
\]
Then 
\[
\begin{aligned}\sqrt{NT}\left(\hat{\theta}^{\left(m+1\right)}-\theta_{0}\right) & =\sum_{s=0}^{m}\left[\bar{W}^{-1}U^{\left(0\right)}\right]^{s}\left\{ \bar{W}^{-1}U^{\left(1\right)}+\bar{W}^{-1}U^{\left(2\right)}\right\} +o_{p}\left(1\right)\\
 & =\left(\mathbb{I}-\bar{W}^{-1}U^{\left(0\right)}\right)^{-1}\left(\mathbb{I}-\left[\bar{W}^{-1}U^{\left(0\right)}\right]^{m+1}\right)\bar{W}^{-1}\sqrt{NT}\left(U^{\left(1\right)}+U^{\left(2\right)}\right)+o_{p}\left(1\right)\\
 & =\left(\mathbb{I}-\bar{W}^{-1}U^{\left(0\right)}\right)^{-1}\bar{W}^{-1}\sqrt{NT}\left(U^{\left(1\right)}+U^{\left(2\right)}\right)+o_{p}\left(1\right)\\
 & =\left(\bar{W}-U^{\left(0\right)}\right)^{-1}\sqrt{NT}\left( U^{\left(1\right)}+ U^{\left(2\right)}\right)+o_{p}\left(1\right)
\end{aligned}
\]
where the second equality follows from the summation formula for geometric
sequences and the penultimate equality holds since for $m\geq-\frac{1}{2}\log(NT)/\log\left(\|\bar{C}^{\left(0\right)}\|\right)-1$,
\[
\left[\bar{W}^{-1} U^{\left(0\right)}\right]^{m+1}\bar{W}^{-1}\sqrt{NT}\left(U^{\left(1\right)}+ U^{\left(2\right)}\right)=O_{p}\left(\bar{W}^{-1}\left( U^{\left(1\right)}+ U^{\left(2\right)}\right)\right)=O_{p}\left(\left(NT\right)^{-1/2}\right)
\]
Thus, by the definition of $\bar{W}$, $U^{\left(0\right)}$, and
$\Xi_{it}$, we have 
\[
\bar{W}-U^{\left(0\right)}=\frac{1}{NT}\sum_{i=1}^{N}\sum_{t=1}^{T}\mathbb{E}\left[\partial_{\theta\theta}\ell_{it}-\partial_{\pi^{2}}\ell_{it}\Xi_{it}\Xi_{it}^{\prime}\right]
\]
By CLT and Bartlett identities, 
\[
\sqrt{NT}U^{\left(1\right)}=\frac{1}{\sqrt{NT}}\sum_{i=1}^{N}\sum_{t=1}^{T}-D_{\theta}\ell_{it}\stackrel{d}{\longrightarrow}\mathcal{N}\left(0,\bar{W}_{\infty}\right)
\]
We also have $U^{(2)}\stackrel{p}{\longrightarrow}\frac{1}{T}\bar{B}_{\infty}+\frac{1}{N}\bar{D}_{\infty}$
from (\ref{eqn:U2}). The desired result then follows. 
\end{proof}

\begin{proof}[Proof of Theorem~\ref{thm:inf-ABC}]
Under the assumptions in Theorem~\ref{thm:inf-dist},   $\hat{B}\stackrel{p}{\longrightarrow}\bar{B}_{\infty}$,  $\hat{D}\stackrel{p}{\longrightarrow}\bar{D}_{\infty}$, and $\hat{W}\stackrel{p}{\longrightarrow}\bar{W}_{\infty}$. The results follow directly from applying the same argument as in the proof of Lemma S.1 and Theorem 4.3 in the Supplementary Material of \citet*{fernandez-val_individual_2016}. In particular, it relies on repeated applications of the weak law of large numbers and Slutsky’s theorem.

Given the definition of $\hat{\theta}_{A B C}$, we can write
$$
\begin{aligned} \sqrt{N T}\left(\hat{\theta}_{A B C}-\theta_0\right)= & \sqrt{N T}\left(\hat{\theta}^{(m+1)}-\frac{1}{T} \hat{W}^{-1} \hat{B}-\frac{1}{N} \hat{W}^{-1} \hat{D}-\theta_0\right) \\ = & \underbrace{\sqrt{N T}\left(\hat{\theta}^{(m+1)}-\theta_0-\frac{1}{T} \bar{W}_{\infty}^{-1} \bar{B}_{\infty}-\frac{1}{N} \bar{W}_{\infty}^{-1} \bar{D}_{\infty}\right)}_{(a)} \\ & -\underbrace{\sqrt{N T}\left[\frac{1}{T}\left(\hat{W}^{-1} \hat{B}-\bar{W}_{\infty}^{-1} \bar{B}_{\infty}\right)+\frac{1}{N}\left(\hat{W}^{-1} \hat{D}-\bar{W}_{\infty}^{-1} \bar{D}_{\infty}\right)\right]}_{(b)}
\end{aligned}
$$
By Theorem~\ref{thm:inf-dist}, we have $(a)\xrightarrow{d} \mathcal{N}\left(0, \bar{W}_{\infty}^{-1}\right)$. Since we have established the consistency of $\hat{B}$, $\hat{D}$, and $\hat{W}$,  Slutsky’s theorem implies that $\hat{W}^{-1} \hat{B} \xrightarrow{p} \bar{W}_{\infty}^{-1} \bar{B}_{\infty}$ and $\hat{W}^{-1} \hat{D} \xrightarrow{p} \bar{W}_{\infty}^{-1} \bar{D}_{\infty}$. Consequently, term $(b)$ is $o_p(1)$. This completes the proof.

\end{proof} \section{Additional Details on Empirical Application}
\label{apx:empirical}
\subsection{Data Sources and Sample Construction}
\label{apx:emp-data}
Our empirical analysis is based on firm and market data following the framework of \citet{ma_nonfinancial_2019}. The firm-level dataset combines information from Compustat, CRSP, IBES, TRACE, Datastream, and Mergent’s Fixed Income Securities Database (FISD). Flow variables such as net equity repurchases, net debt issuance, and capital expenditures are measured at the quarterly frequency and normalized by lagged total assets, while stock variables such as cash holdings and leverage are normalized by contemporaneous assets. Outliers are winsorized at the 1\% level in all firm-level analyses. The sample excludes financial firms (SIC 6000–6999) and firms with missing bond or equity identifiers. See Table~\ref{tab:var_def} for detailed variable definitions.

Firm-level bond variables are constructed from TRACE and FISD. Specifically, the credit spread is the face-value-weighted difference between each bond’s yield and the yield on the nearest-maturity Treasury, excluding convertible, asset-backed, and foreign-currency bonds, as well as bonds with less than one year to maturity. When yields are unavailable after November 2008, they are imputed from TRACE prices and FISD coupon information. The term spread is constructed analogously, using Treasury yields of different maturities. On the equity side, we construct the value-to-price (V/P) ratio following \citet*{dong_overvalued_2012}, combining Compustat book equity, CRSP prices, and IBES earnings forecasts.

Quarterly net equity repurchases are defined as share repurchases minus equity issuance (PRSTKC–SSTK), and net debt issuance is defined as long-term debt issued minus retired (DLTIS–DLTR), each scaled by lagged assets. Additional controls include net income, cash holdings, capital expenditures, and deviations from target leverage, all from Compustat and defined as in \citet*{ma_nonfinancial_2019}. The aggregate bond yields and returns come from the Barclays Capital / Lehman Brothers indices compiled by \citet*{greenwood_issuer_2013}, updated with Bank of America Merrill Lynch data. Treasury yields are sourced from FRED, and stock market data is sourced from Robert Shiller's historical dataset.

To ensure comparability in the estimation of firm-specific effects, we focus on the period 2003-2024, during which TRACE coverage becomes comprehensive. The earlier years used in \citet*{ma_nonfinancial_2019}, particularly the pre-2000 period, are highly unbalanced––especially for bond-level variables––and exhibit systematic patterns of missingness related to limited TRACE reporting and incomplete data integration across sources. We define an observation as missing if at least one of the key or control variables is unavailable for a given firm–quarter in the panel. Consequently, we restrict the sample to the period 2003Q1–2024Q4 and construct a high-coverage panel comprising 212 nonfinancial firms observed over 88 quarters in total. Firms are retained in the sample as long as they have at least 44 quarters of observed data over the sampled period.

Figure \ref{fig:missing_pattern} visualizes the distribution of missing observations across firms and quarters in the final sample. Large blocks of missing observations near the sample edges primarily correspond to firms entering the panel mid-period or exiting early, while the few short gaps observed in the middle of the sample likely reflect random reporting noise rather than systematic non-coverage. Accordingly, we proceed under the assumption that the data are missing at random (MAR) in our empirical application.\footnote{Missing observations are common in empirical corporate‐finance and asset‐pricing panels. The objective of this paper is to achieve consistent estimation in high-coverage panels rather than to model systematic missingness. For broader treatments of structured or non‐MAR missingness in financial panels, see \citet*{freyberger_missing_2025}, \citet*{cahan_factor-based_2023}, and \citet*{bryzgalova_missing_2022}. Extending the proposed framework to accommodate such forms of missingness is a promising direction for future research.}

Our updated estimation algorithm explicitly handles these limited gaps under the assumption that missing observations are missing at random (MAR). The first-step estimator employs the Soft-Impute algorithm of \citet*{mazumder_spectral_2010}, while the second-step estimator inherently accommodates MAR data, following \citet*{bai_panel_2009} and \citet*{chen_nonlinear_2021}.

\begin{figure}[h] \centering \includegraphics[width=0.8\textwidth]{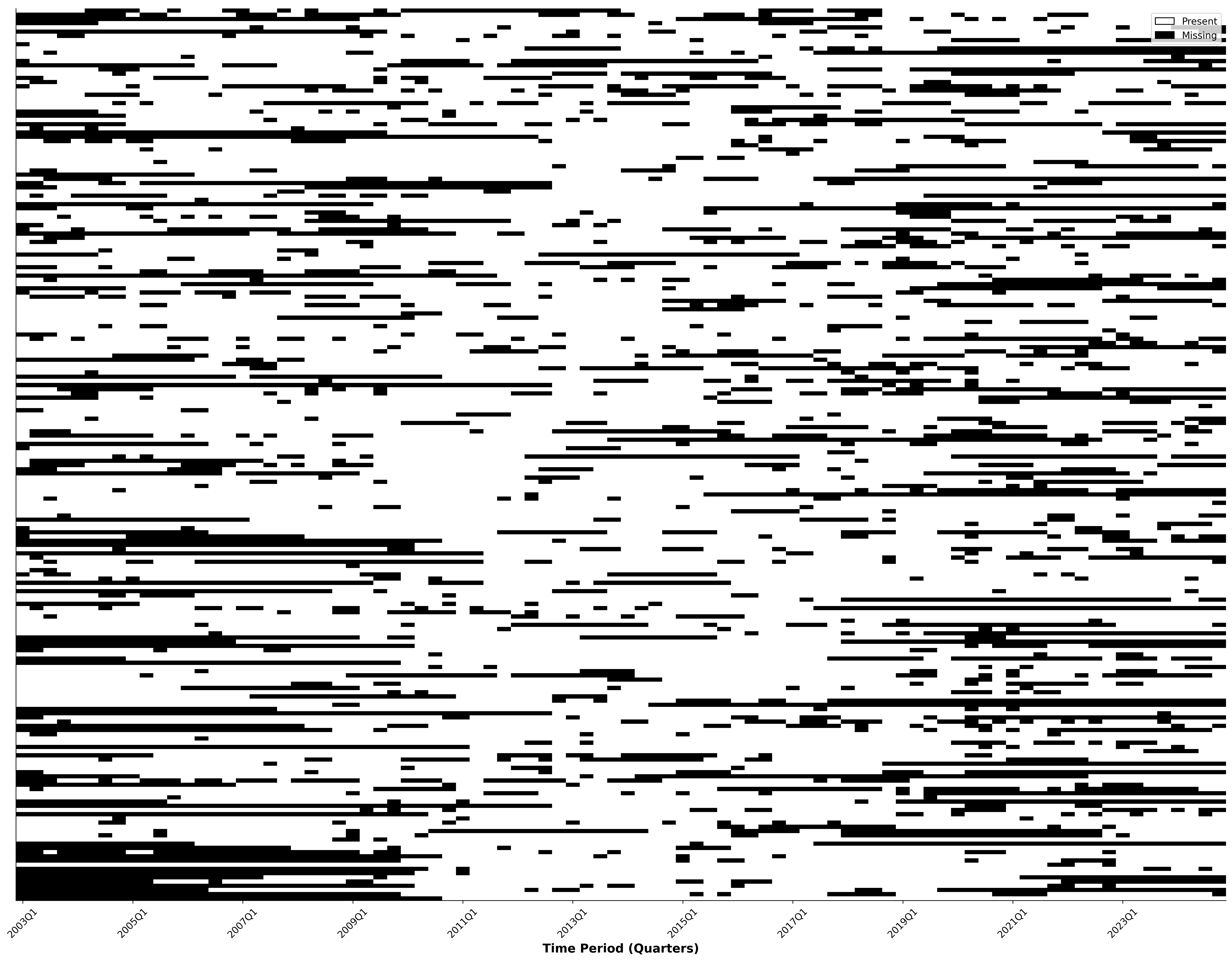} \caption{\textbf{Data availability across firms and quarters.} Each cell represents the availability of firm‐level observations in the final estimation sample (212 firms, 2003Q1–2024Q4). White cells indicate available observations, and black cells denote missing quarters.} 
\label{fig:missing_pattern} 
\end{figure}

\begin{table}[htbp]
\begin{center}
\renewcommand{\arraystretch}{1.4}
\caption{Firm-Level Variable Definitions}
\label{tab:var_def}
\begin{tabular}{lp{8.2cm}p{2.8cm}}
\toprule
Variable & Description & Main Source \\
\midrule
Credit spread
& Face-value-weighted difference between a firm’s corporate bond yield and the yield on the nearest-maturity Treasury
& TRACE, FISD, FRED \\

Term spread
& Face-value-weighted difference between the yield on the nearest-maturity Treasury and the 3-month Treasury yield
& TRACE, FISD, FRED \\

V/P ratio
& Value-to-price ratio constructed following \citet*{dong_overvalued_2012}, using forecasted ROE and book equity
& Compustat, CRSP, IBES \\

Cash holdings
& Cash and short-term investments
& Compustat \\

Asset growth
& Quarterly growth rate of total assets
& Compustat \\

Size
& Natural logarithm of book value of total assets
& Compustat \\

Net income
& Net income
& Compustat \\

Capital expenditures
& Capital expenditures (CAPX)
& Compustat \\
\bottomrule
\end{tabular}
\end{center}
\footnotesize{
Note: This table reports definitions of firm-level variables used in the empirical analysis. For firm financials, flow variables are measured quarterly and normalized by lagged total assets, while stock variables are normalized by contemporaneous total assets. Datastream is used as a supplementary data source where applicable. Detailed construction follows the Internet Appendix of \citet*{ma_nonfinancial_2019}. }
\end{table}
 
\subsection{Additional Results}
\begin{figure}[H]
    \centering
    \includegraphics[width=0.8\textwidth]{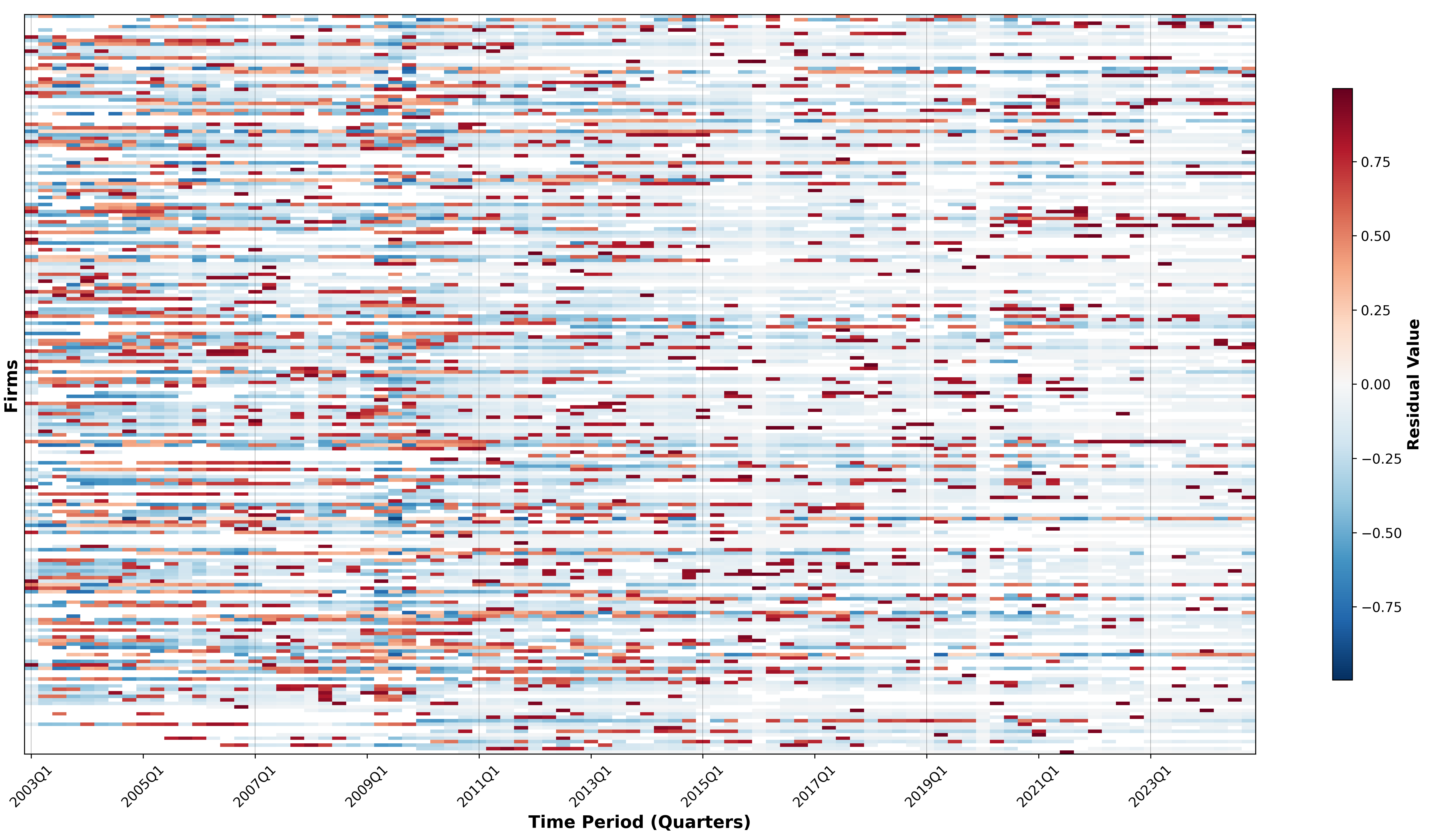}
    \caption{\textbf{Residual heatmap from TWFE estimation.}
    The figure plots firm–quarter residuals from the TWFE estimator following \citet*{fernandez-val_individual_2016}. }
    \label{fig:resid_heatmap}
\end{figure}

\end{appendices}

\end{document}